\documentclass[a4paper]{article}
\usepackage{tikz}
\tikzstyle{trepp}=[circle,draw,minimum size=2.25mm,inner sep=0.1pt]
\newcommand{\etq}[1]{%
\draw (#1) node {\tiny $#1$};}

\usepackage[utf8]{inputenc}
\usepackage{amssymb,amsthm,amsmath,enumerate,graphicx,mathtools,hyperref,amsfonts,latexsym,url,xcolor}
\usepackage{cancel}

\newcommand{\NN}{\mathbb{N}}

\newcommand{\RR}{\mathbb{R}}
\newcommand{\QQ}{\mathbb{Q}}
\renewcommand{\leq}{\leqslant}
\renewcommand{\geq}{\geqslant}

\theoremstyle{plain}
\newtheorem{thm}{Theorem}
\newtheorem{lemma}[thm]{Lemma}
\newtheorem{proposition}[thm]{Proposition}
\newtheorem{cor}[thm]{Corollary}
\theoremstyle{definition}
\newtheorem{rem}[thm]{Remark}
\newtheorem{example}{Example}

\newcounter{exs}

\newcommand{\aadm}{\alpha^{(d,m)}}

\begin{document}
 
\title{Explicit solution of divide-and-conquer  dividing by a half recurrences with polynomial independent term}

\author{\textbf{Tom\'as M. Coronado, Arnau Mir, Francesc Rossell\'o}\\[1ex]
Department of Mathematics and Computer Science,\\
University of the Balearic Islands, E-07122 Palma, Spain\\
and\\
Balearic Islands Health Research Institute (IdISBa), E-07010 Palma, Spain}

\maketitle

\begin{abstract}
Divide-and-conquer   dividing by a half recurrences, of the form
$$
x_n =a\cdot x_{\left\lceil{n}/{2}\right\rceil}+a\cdot x_{\left\lfloor{n}/{2}\right\rfloor}+p(n),\quad n\geq 2,
$$
appear in many areas of applied mathematics, from the analysis of algorithms to the optimization of phylogenetic balance indices. The Master Theorems that solve these equations do  not provide the solution's explicit expression, only its big-$\Theta$ order of growth. In this paper we give an explicit expression ---in terms of the binary decomposition of $n$--- for the solution $x_n$ of a recurrence of this form, with given initial condition $x_1$, when the independent term $p(n)$ is a polynomial in  $\lceil{n}/{2}\rceil$ and $\lfloor{n}/{2}\rfloor$.
\end{abstract}

\section{Introduction}

\emph{Divide-and-conquer dividing by a half} recurrences, of the form
\begin{equation}
x_n =a\cdot x_{\left\lceil{n}/{2}\right\rceil}+a\cdot x_{\left\lfloor{n}/{2}\right\rfloor}+p(n),\quad n\geq 2,
\label{eqn:def}
\end{equation}
appear in many areas of applied mathematics. As evidence, for instance, the \emph{On-line Encyclopedia of Integer Sequences} (OEIS, \url{https://oeis.org}) contains more than 200 integer sequences with the keyword ``divide and conquer', many of which satisfy a recurrence like (\ref{eqn:def}) \cite{StephanOEIS}.

The most popular and best studied setting where such recurrences arise is in Computer Science, and more specifically in the analysis of balanced divide and conquer algorithms.  This type of algorithms solve a problem by splitting its input into  two or more parts of the same size, solving (recursively) the  problem on these parts, and finally combining these solutions into a solution for the global instance \cite[\S 2.6--7]{AHU}. Typical examples of this strategy are  the heapsort and mergesort algorithms and several fast integer and matrix multiplication methods. When the input is split into two parts of  the same size (up to a unit of difference, when the size of the original input is odd) and both parts contribute equally to the final solution, the algorithm's running time $t_n$ on an instance of size $n$ satisfies a recurrence of the form
$$
t_n=a\cdot t_{\left\lceil{n}/{2}\right\rceil}+a\cdot t_{\left\lfloor{n}/{2}\right\rfloor}+p(n)
$$
where the independent term $p(n)$, called in this context the \emph{toll function}, represents the cost of combining the solutions of the subproblems into a solution for the original problem.

Our interest in this type of recurrences stems from the study of phylogenetic balance indices. A \emph{phylogenetic tree}, the standard representation of the joint evolutionary history of a group of extant species (or other Operational Phylogenetic Units, OPU, like genes, languages, or myths), is, from the formal point of view, a leaf-labeled rooted tree. In a phylogenetic tree, its leaves represent  the species under study, its internal nodes represent their common ancestors, the root represents the most recent common ancestor of all of them, and the arcs represent direct descendance through mutations \cite{fel:04,semple}.  A phylogenetic tree is \emph{bifurcating}, or \emph{fully resolved}, when every internal node has two direct descendants, or \emph{children}. 

Biologists use the \textit{shape} of phylogenetic trees, that is, their raw branching structure, to deduce information on the  forces beneath the speciation and extinction processes that have taken place \cite{futuyma}. A popular set of tools used in the analysis of phylogenetic tree shapes are the \textit{shape indices} \cite{QGT}, and among them the \textit{balance indices}, which measure the propensity of the direct descendants of any given node in a tree to have the same number of descendant leaves. Many balance indices have been proposed in the literature: see, for instance,  \cite{colless82, fel:04, Fischer2015, Fusco95, kirk-slatk, cherries, cophenetic, collesslike, sackin, shao} and the recent survey \cite{surveybal} and the references therein.  Let us recall here the four indices that appear below in the Examples section:
\begin{itemize}
\item The \emph{Sackin index} $S(T)$ of a tree $T$ \cite{sackin, shao} is the sum of its leaves' \emph{depths} (that is, of the lengths of the paths from the root to the leaves).
\item The \emph{Colless index} $C(T)$ of a bifurcating tree $T$ \cite{colless82} is the sum, over all its internal nodes $v$, of the absolute value of the difference between the numbers of descendant leaves of the pair of children of $v$.
\item The  \emph{Total Cophenetic index} $\Phi(T)$ of a tree $T$ \cite{cophenetic} is the sum, over all its pairs of internal nodes, of the depth of their lowest common ancestor.
\item The \emph{rooted Quartet index} $\textit{rQI}(T)$  of a bifurcating tree $T$ \cite{rQI} is the number of its subtrees of 4 leaves that are fully symmetric.
\end{itemize}
Any sensible balance index should classify, for each number of leaves $n$, as most unbalanced trees the \emph{rooted caterpillars} with $n$ leaves $K_n$ ---the bifurcating trees such that all their internal nodes have a leaf child--- and as most balanced bifurcating  trees the \emph{maximally balanced trees} with $n$ leaves $B_n$ ---where, for each internal node, the numbers of descendant leaves of its pair of children differ at most in 1--- (cf. Fig. \ref{fig:intro}). The four indices described above satisfy this condition \cite{surveybal}.

\begin{figure}[htb]
\begin{center}
$$
\begin{array}{ccc}
\begin{tikzpicture}[thick,>=stealth,scale=0.5]
\draw(0,0) node [trepp] (1) {}; \etq{1}
\draw(1,0) node [trepp] (2) {}; \etq{2}
\draw(2,0) node [trepp] (3) {}; \etq{3}
\draw(3,0) node [trepp] (4) {};\etq{4}
\draw(4,0) node [trepp] (5) {}; \etq{5}
\draw(5,0) node [trepp] (6) {};\etq{6}
\draw(6,0) node [trepp] (7) {};\etq{7}
\draw(0.5,1/2) node [trepp] (v1) {};
\draw(1.5,2/2) node [trepp] (v2) {};
\draw(2.5,3/2) node [trepp] (v3) {};
\draw(3.5,4/2) node [trepp] (v4) {};
\draw(4.5,5/2) node [trepp] (v5) {};
\draw(5.5,6/2) node [trepp] (v6) {};
\draw  (v1)--(1);
\draw  (v1)--(2);
\draw  (v2)--(3);
\draw  (v3)--(4);
\draw  (v4)--(5);
\draw  (v5)--(6);
\draw  (v6)--(7);
\draw  (v6)--(v5);
\draw  (v5)--(v4);
\draw  (v4)--(v3);
\draw  (v3)--(v2);
\draw  (v2)--(v1);
\draw (3,-1) node {$K_7$};
\end{tikzpicture} 
&
\begin{tikzpicture}[thick,>=stealth,scale= 0.5]
\draw(0,0) node [trepp] (1) {}; \etq{1}
\draw(1,0) node [trepp] (2) {}; \etq{2}
\draw(2,0) node [trepp] (3) {}; \etq{3}
\draw(3,0) node [trepp] (4) {};\etq{4}
\draw(4,0) node [trepp] (5) {}; \etq{5}
\draw(5,0) node [trepp] (6) {};\etq{6}
\draw(6,0) node [trepp] (7) {};\etq{7}
\draw(0.5,1) node [trepp] (v1) {};
\draw(2.5,1) node [trepp] (v2) {};
\draw(4.5,1) node [trepp] (v3) {};
\draw(1.5,2) node [trepp] (w1) {};
\draw(5.5,2) node [trepp] (w2) {};
\draw(3.5,3) node [trepp] (r) {}; 
\draw  (v1)--(1);
\draw  (v1)--(2);
\draw  (v2)--(3);
\draw  (v2)--(4);
\draw  (v3)--(5);
\draw  (v3)--(6);
\draw  (w1)--(v1);
\draw  (w1)--(v2);
\draw  (w2)--(v3);
\draw  (w2)--(7);
\draw  (r)--(w1);
\draw  (r)--(w2);
\draw (3,-1) node {$B_7$};
\end{tikzpicture} 
\end{array}
$$
\end{center}
\caption{A rooted caterpillar with 7 leaves (left) and a maximally balanced with 7 leaves (right).\label{fig:intro}}
\end{figure}
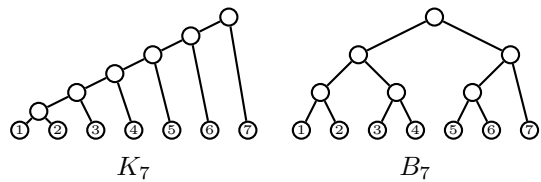

Now, given a balance index $I$, it is unreasonable  to expect it to allow for the comparison of the balance of trees with different numbers of leaves. A possible way to circumvent this difficulty is to normalize it  to $[0,1]$ for every number of leaves  $n$  by subtracting its minimum  value for this number of leaves and dividing by its range of values on the space of trees with $n$ leaves \cite{shao,Stam02}. To perform this normalization for bifurcating trees,  it is necessary to know the value of $I$ for the rooted caterpillars $K_n$ and the maximally balanced trees $B_n$. It turns out that, while computing $I(K_n)$ is usually easy, to compute $I(B_n)$ one is led in most cases to solve a recurrence of the form
$$
I(B_n)=I(B_{\left\lceil{n}/{2}\right\rceil})+I(B_{\left\lfloor{n}/{2}\right\rfloor})+p_I(n)
$$
with $p_I(n)$ a function that, in some sense, measures the contribution of the root to the value of $I$. For the four aforementioned indices, this independent term is a polynomial in  $\lceil{n}/{2}\rceil$ and $\lfloor{n}/{2}\rfloor$: for the Sackin index, $p_S(n)=n=\lceil{n}/{2}\rceil+\lfloor{n}/{2}\rfloor$; for the Colless index, $p_C(n)=\lceil{n}/{2}\rceil-\lfloor{n}/{2}\rfloor$; for the Total Cophenetic index, $p_\Phi(n)=\binom{n}{2}=\binom{\lceil{n}/{2}\rceil+\lfloor{n}/{2}\rfloor}{2}$; and for the rooted Quartet index, $p_{rQI}(n)=\binom{\lceil{n}/{2}\rceil}{2}\binom{\lfloor{n}/{2}\rfloor}{2}$  \cite{surveybal}. At the moment of writing this paper, the values of $S(B_n)$ and $C(B_n)$ were known (see Examples \ref{ex:Sackin} and \ref{ex:Colless} in Section \ref{sec:exm}), but not those of $\Phi(B_n)$ or $\textit{rQI}(B_n)$.

In the context of the analysis of algorithms, divide-and-conquer recurrences like (\ref{eqn:def}) and more general  ones are ``solved'' by finding the big-$\Theta$ order of growth of $x_n$ using some \emph{Master Theorem}. The original Master Theorem was obtained by Bentley, Haken, and Saxe \cite{Bentley} and extended in \cite[\S 4.5--6]{Cormen}, and since then several improved versions have been obtained by other researchers: see, for instance, \cite{Akra,Drmota, Guedj, Kao,Roura,Verma1,Verma2,Wang}. Tipically, a Master Theorem  deduces the asymptotic behaviour of a solution $x_n$ of  (\ref{eqn:def}) from that of the independent term of the recurrence (and, in more general recurrences, from the number of parts into which the input is divided, which we fix here to 2, and the contribution of each subproblem to the general problem, which we assume here to be equal and represented by the coefficient $a$). Thus, Master Theorems do not provide explicit solutions of recurrences, only their asymptotic behaviour.

In the analysis of algorithms, knowing the growing order of the computational cost of an algorithm  on an instance of size $n$ is usually enough. But in other applications, like for instance in order to normalize balance indices as we explained above, an explicit expression for the solution is needed. 
Many specific divide-and-conquer recurrences are explicitly solved when needed, like for instance the cost of the mergesort algorithm in the worst case, solved as Ex. 34 in \cite[Ch. 3]{Knuth2}, but no general solution is known. To our knowledge, the only attempt to find an explicit solution of Eqn. (\ref{eqn:def}) is made by Hwang, Janson and Tsai \cite{HJT17} by proving that, when $a=1$ and under very general conditions on the independent term $p(n)$, the solution $x_n$ has the form
$$
x_n=nP(\log_2(n))+F(n)-Q(n)
$$
with $P$ continuous and 1-periodic and $F,Q$ of precise growing orders, and giving explicit expressions for $P,F,Q$ in terms of series expansions. 

In this paper we consider the more restrictive case when $p(n)$ is a polynomial  in  $\lceil{n}/{2}\rceil$ and $\lfloor{n}/{2}\rfloor$. In this case, we are able to give an explicit finite formula for $x_n$ in terms of the binary decomposition of $n$. Although for our applications the case when  $a=1$ was enough, our formula works for the arbitrary $a$ case.

The rest of this paper is organized as follows. In Section \ref{sec:mr} we state our main result and in Section \ref{sec:exm} we use it to solve several examples. As a proof of concept,  in some of them  we obtain formulas that were already known, but we also include several examples for which no explicit formula was known so far. Finally, the bulk of the paper (Section \ref{sec:proof}) is devoted to prove our main result. 

We have implemented the formula provided by our main theorem in Python using SymPy, a Python library for symbolic mathematics. This implementation  is available at the GitHub repository \url{https://github.com/biocom-uib/divide_and_conquer}.\medskip 

\noindent\textbf{Acknowledgments.}
This research was partially supported by the Spanish Ministry of Science, Innovation and Universities and the European Regional Development Fund through project PGC2018-096956-B-C43 (FEDER/MICINN/AEI). We thank G. Valiente for some useful suggestions.

\section{Main result}\label{sec:mr}

\noindent For every $n\in \NN_{\geq 1}$, let its binary decomposition be $n=\sum_{j=1}^{s_n} 2^{q_j(n)}$, with $0\leq q_1(n)<\cdots<q_{s_{n}}(n)$,  and then, for every $i\geq 1$, let $M_i(n)=\sum_{j=i}^{s_n} 2^{q_j(n)}$. Notice that $M_1(n)=n$ and  $M_{s_n}=2^{q_{s_n}(n)}$.
In order to simplify the notations, we set $q_0(n) = 0$ and $M_0(n)=M_1(n)+2^{q_0(n)}=n+1$.

Additionally, for every $d,n\in \NN$ and $x\in \RR_{\neq 0}$ let
$$
T(d,n,x)=\sum_{k=0}^{n-1} k^dx^k
$$
(with the convention, which we use throughout this paper, that  $0^0=1$). In particular,
$$
T(0,n,x)=\left\{
\begin{array}{ll}
n & \text{if $x=1$}\\[2ex]
\displaystyle  \frac{x^n-1}{x-1}& \text{if $x\neq 1$}
\end{array}\right.
\qquad
T(1,n,x)=\left\{
\begin{array}{ll}
\displaystyle \binom{n}{2}& \text{if $x=1$}\\[2ex]
\displaystyle  \frac{nx^n(x-1)-x(x^n-1)}{(x-1)^2}& \text{if $x\neq 1$}
\end{array}\right.
$$
For every $m\in \NN$, let $B_m$ denote the $m$-th Bernoulli number of the first kind \cite{Berni}. In particular, $B_0=1$, $B_1=-1/2$, and $B_{2k+1}=0$ for every $k\geq 1$.

For $d=0,1$, for every $a\in \RR$, for every $m\in \NN$, and for every $n\in \NN_{\geq 1}$, let
\begin{align*}
& \aadm_{n}(a) \\ &\quad =\frac{1}{2(m+1)}\sum_{i=1}^{s_n}\sum_{j=0}^{m}\binom{m+1}{j}B_j 2^{j}M_{i}(n)^{m+1-j}\big(T(d,q_i(n),a2^{j-m})-T(d,q_{i-1}(n),a2^{j-m})\big)\nonumber\\
&\qquad\qquad +\sum_{i=1}^{s_n-1}q_i(n)^d(a2^{-m})^{q_i(n)} (n-M_{i}(n))M_{i+1}(n)^m -T(d,q_{s_n}(n),2a)\cdot \delta_{m=0}
\end{align*}
with $\delta_{m=0}=1$ if $m=0$ and $\delta_{m=0}=0$ if $m>0$. 

For every $a\in \RR_{>0}$ and for every $r,t\in \NN$, let $\ell_{a,t}=\log_2(a)-t$ and
$$
\delta_{\ell}(a,r,t)=\left\{\begin{array}{ll}
1 & \text{ if $r>0$, and $\ell_{a,t}\in \{0,\ldots,r-1\}$}\\
0 & \text{ otherwise}
\end{array}\right.
$$
If $a\leq 0$, then $\ell_{a,t}$ is undefined and $\delta_{\ell}(a,r,t)=0$.

The main result in this paper is the following:

\medskip\noindent\textbf{Theorem. } {\it Let $a\in \RR_{\neq 0}$ and let $P(x,y)=\sum_{r,t\geq 0} b_{r,t}x^ry^t\in \RR[x,y]$ be a bivariate polynomial.  Then, the solution of the recurrence equation
$$
x_n =a\cdot x_{\left\lceil{n}/{2}\right\rceil}+a\cdot x_{\left\lfloor{n}/{2}\right\rfloor}+P\big(\lceil{n}/{2}\rceil,\lfloor{n}/{2}\rfloor\big),\quad n\geq 2,
$$
with initial condition $x_1$ is
$$
x_n=\sum_{r,t}b_{r,t} x_n^{(r,t)}(a)+\big((2a)^{q_{s_n}(n)}+(2a-1)(na^{q_{s_n}(n)}-(2a)^{q_{s_n}(n)})\big)x_1
$$
with, for every $n\geq 2$,
\begin{align*}
&x^{(r,t)}_n(a) =   \sum_{k=1}^{r+t} \Bigg(\sum_{i=k\atop i\neq t+\ell_{a,t}+1}^{r+t}\frac{\binom{r}{i-t-1}\binom{i}{k}B_{i-k}}{i(2^{i-1}-a)} \Bigg)n^k
+\frac{1}{a-1}\cdot\delta_{r>0,t=0,a\neq 1}\nonumber \\
&\qquad +\Bigg(1-\sum_{l=0\atop l\neq \ell_{a,t}}^{r-1}\frac{\binom{r}{l}}{2^{t+l}-a} \Bigg)\big(T(0,q_{s_n}(n),2a)+na^{q_{s_n}(n)}-(2a)^{q_{s_n}(n)}\big)\nonumber\\
 &\qquad+\sum_{i=0}^{r+t-1} \Bigg(2^{-i}\binom{r+t}{i}-2^{-i+1}\binom{r}{i-t}-\sum_{l=i-t+1\atop l\neq \ell_{a,t}}^{r-1}\frac{\binom{r}{l}\binom{t+l}{i}}{2^{t+l}-a} \Bigg)\alpha^{(0,i)}_{n}(a)\\
&\qquad +\frac{\delta_{\ell}(a,r,t)}{a}\binom{r}{\ell_{a,t}}\Bigg(T(1,q_{s_n}(n),2a)+(na^{q_{s_n}(n)}-(2a)^{q_{s_n}(n)})q_{s_{n}}(n)\\
& \qquad\hphantom{+\frac{\delta_{\ell}(a,r,t)}{a}\binom{r}{\ell_{a,t}}\Bigg(T}+\sum_{i=0}^{t+\ell_{a,t}-1}\binom{t+\ell_{a,t}}{i}\alpha^{(1,i)}_{n}(a)\Bigg)
\end{align*}
where $\delta_{r>0,t=0,a\neq 1}=1$ if $r>0$, $t=0$, and $a\neq 1$, and $\delta_{r>0,t=0,a\neq 1}=0$ otherwise.}\medskip

In particular:
\begin{enumerate}[(a)]
\item If $r=0$,
\begin{align*}
x_n^{(0,t)}(a) &=\sum_{i=0}^{t-1} 2^{-i}\binom{t}{i} \alpha^{(0,i)}_{n}(a)+T(0,q_{s_n}(n),2a)+na^{q_{s_n}(n)}-(2a)^{q_{s_n}(n)}\\
&=
\left\{\begin{array}{ll}
\displaystyle\sum_{i=0}^{t-1} 2^{-i}\binom{t}{i} \alpha^{(0,i)}_{n}(1/2)+q_{s_n}(n)+n\cdot 2^{-q_{s_n}(n)}-1 & \text{ if $a=\frac{1}{2}$}\\[1ex]
\displaystyle\sum_{i=0}^{t-1} 2^{-i}\binom{t}{i} \alpha^{(0,i)}_{n}(a)+\dfrac{(2a)^{q_{s_n}(n)}-1}{2a-1}+n\cdot a^{q_{s_n}(n)}-(2a)^{q_{s_n}(n)}& \text{ if $a\neq \frac{1}{2}$}
\end{array}\right.
\end{align*}

\item If $r\geq 1$ and $t=0$:
\begin{enumerate}
\item[(b.1)] If  $a=1/2$,
\begin{align*}
& x_n^{(r,0)}(1/2) =2\sum_{k=1}^{r}\Big(\sum_{j=k}^{r}\frac{\binom{r}{j-1}\binom{j}{k}B_{j-k}}{j(2^{j}-1)}\Big)n^k-2\\
&\qquad +\Big(1-2\sum_{l=0}^{r-1}\frac{\binom{r}{l}}{2^{l+1}-1}\Big)\Big(q_{s_n}(n)+n\cdot 2^{-q_{s_n}(n)}-1\Big)\\
&\qquad -\sum_{i=0}^{r-1}\Big(2^{-i}\binom{r}{i}+2\sum_{l=i+1}^{r-1}\frac{\binom{r}{l}\binom{l}{i}}{2^{l+1}-1}\Big)\alpha^{(0,i)}_{n}(1/2)
\end{align*}

\item[(b.2)] If $a=1$,
\begin{align*}
&x_n^{(r,0)}(1) =\sum_{k=2}^{r}\Big(\sum_{j=k}^{r}\frac{\binom{r}{j-1}\binom{j}{k}B_{j-k}}{j(2^{j-1}-1)}\Big)n^k+\Big(q_{s_n}(n)+1+\sum_{j=1}^{r-1}\frac{\binom{r}{j}(B_{j}-1)}{2^{j}-1}\Big)n\\
 &\qquad  +1+\sum_{j=1}^{r-1}\frac{\binom{r}{j}}{2^{j}-1}-2^{q_{s_n}(n)+1}-\sum_{i=0}^{r-1}\Big(2^{-i}\binom{r}{i}+\sum_{l=i+1}^{r-1}\frac{\binom{r}{l}\binom{l}{i}}{2^{l}-1}\Big)\alpha^{(0,i)}_{n}(1)
\end{align*}

\item[(b.3)] If $a\notin\{1/2,1,\ldots,2^{r-1}\}$
\begin{align*}
&x_n^{(r,0)}(a) =\sum_{k=1}^{r}\Big(\sum_{j=k}^{r}\frac{\binom{r}{j-1}\binom{j}{k}B_{j-k}}{j(2^{j-1}-a)}\Big)n^k+\frac{1}{a-1}\\
&\qquad+\Big(1-\sum_{l=0}^{r-1}\frac{\binom{r}{l}}{2^{l}-a}\Big)\Big(\dfrac{(2a)^{q_{s_n}(n)}-1}{2a-1}+na^{q_{s_n}(n)}-(2a)^{q_{s_n}(n)}\Big)\\
&\qquad -\sum_{i=0}^{r-1}\Big(2^{-i}\binom{r}{i}+\sum_{l=i+1}^{r-1}\frac{\binom{r}{l}\binom{l}{i}}{2^{l}-a}\Big)\alpha^{(0,i)}_{n}(a)
\end{align*}

\item[(b.4)] If $a=2^\ell$ with $\ell\in\{1,\ldots,r-1\}$,
\begin{align*}
& x^{(r,0)}_n(a) = \sum_{k=1}^{r} \Big(\sum_{i=k\atop i\neq \ell+1}^{r}\frac{\binom{r}{i-1}\binom{i}{k}B_{i-k}}{i(2^{i-1}-a)} \Big)n^k
+\frac{1}{a-1}\\ 
&\quad
+\Big(1-\sum_{l=0\atop l\neq \ell}^{r-1}\frac{\binom{r}{l}}{2^{l}-a} \Big)\Big(\dfrac{(2a)^{q_{s_n}(n)}-1}{2a-1}+na^{q_{s_n}(n)}-(2a)^{q_{s_n}(n)}\Big)\nonumber\\
 &\quad   -\sum_{i=0}^{r-1} \Big(2^{-i}\binom{r}{i}+\sum_{l=i+1\atop l\neq \ell}^{r-1}\frac{\binom{r}{l}\binom{l}{i}}{2^{l}-a} \Big)\alpha^{(0,i)}_{n}(a)\\
&\quad +\frac{1}{a}\binom{r}{\ell}\Big(\dfrac{{((2a-1)q_{s_n}(n)-2a)(2a)^{q_{s_n}(n)}+2a}}{(2a-1)^2}+q_{s_n}(n)a^{q_{s_n}(n)}(n-2^{q_{s_n}(n)})\Big)\\
&\quad +\frac{1}{a}\binom{r}{\ell}\sum_{i=0}^{\ell-1}\binom{\ell}{i}\alpha^{(1,i)}_{n}(a)
\end{align*}
\end{enumerate}

\item If $r\geq 1$ and $t\geq 1$:
\begin{enumerate} 
\item[(c.1)] If $a=1/2$,
\begin{align*}
&x_n^{(r,t)}(1/2) =2\sum_{k=1}^{r+t}\Big(\sum_{j=k}^{r+t}\frac{\binom{r}{j-t-1}\binom{j}{k}B_{j-k}}{j(2^{j}-1)}\Big)n^k\\
& \qquad +\Big(1-2\sum_{l=0}^{r-1}\frac{\binom{r}{l}}{2^{t+l+1}-1}\Big)\Big(q_{s_n}(n)+n\cdot 2^{-q_{s_n}(n)}-1\Big)\\
&\qquad +\sum_{i=0}^{r+t-1}\Big(2^{-i}\binom{r+t}{i}-2^{-i+1}\binom{r}{i-t}-2\sum_{l=i-t+1}^{r-1}\frac{\binom{r}{l}\binom{t+l}{i}}{2^{t+l+1}-1}\Big)\alpha^{(0,i)}_{n}(1/2)
\end{align*}

\item[(c.2)] If $a = 1$,
\begin{align*}
&x_n^{(r,t)}(1) =\sum_{k=2}^{r+t}\Big(\sum_{j=k}^{r+t}\frac{\binom{r}{j-t-1}\binom{j}{k}B_{j-k}}{j(2^{j-1}-1)}\Big)n^k\\ &\qquad +{\Big(1+\sum_{j=1}^{r+t-1}\frac{\binom{r}{j-t}(B_{j}-1)}{2^{j}-1}\Big)n+\sum_{l=0}^{r-1}\frac{\binom{r}{l}}{2^{t+l}-1}-1}\\
&\qquad +\sum_{i=0}^{r+t-1}\Big(2^{-i}\binom{r+t}{i}-2^{-i+1}\binom{r}{i-t}-\sum_{l=i-t+1}^{r-1}\frac{\binom{r}{l}\binom{t+l}{i}}{2^{t+l}-1}\Big)\alpha^{(0,i)}_{n}(1)
\end{align*}

\item[(c.3)] If  $a\notin\{1/2,1,2^t,\ldots,2^{r+t-1}\}$,
\begin{align*}
&x_n^{(r,t)}(a) =\sum_{k=1}^{r+t}\Big(\sum_{j=k}^{r+t}\frac{\binom{r}{j-t-1}\binom{j}{k}B_{j-k}}{j(2^{j-1}-a)}\Big)n^k\\
&\qquad+\Big(1-\sum_{l=0}^{r-1}\frac{\binom{r}{l}}{2^{t+l}-a}\Big)\Big(\dfrac{(2a)^{q_{s_n}(n)}-1}{2a-1}+na^{q_{s_n}(n)}-(2a)^{q_{s_n}(n)}\Big)\\
&\qquad +\sum_{i=0}^{r+t-1}\Big(2^{-i}\binom{r+t}{i}-2^{-i+1}\binom{r}{i-t}-\sum_{l=i-t+1}^{r-1}\frac{\binom{r}{l}\binom{t+l}{i}}{2^{t+l}-a}\Big)\alpha^{(0,i)}_{n}(a)
\end{align*}

\item[(c.4)] If $a=2^{t+\ell}$, for some $\ell\in\{0,\ldots,r-1\}$,
\begin{align*}
& x_n^{(r,t)}(a) =\sum_{k=1}^{r+t}\Big(\sum_{j=k\atop j\neq t+\ell+1}^{r+t}\frac{\binom{r}{j-t-1}\binom{j}{k}B_{j-k}}{j(2^{j-1}-a)}\Big)n^k\\
&\qquad +\frac{1}{a}\binom{r}{\ell}\Big(\dfrac{q_{s_n}(n)(2a)^{q_{s_n}(n)}-(2a)^{q_{s_n}(n)+1}+2a}{(2a-1)^2}+q_{s_n}(n)a^{q_{s_n}(n)}(n-2^{q_{s_n}(n)})\Big)\\
&\qquad +\Big(1-\sum_{l=0\atop l\neq \ell}^{r-1}\frac{\binom{r}{l}}{2^{t+l}-a}\Big)\Big(\dfrac{(2a)^{q_{s_n}(n)}-1}{2a-1}+na^{q_{s_n}(n)}-(2a)^{q_{s_n}(n)}\Big)\nonumber\\
&\qquad +\sum_{i=0}^{r+t-1}\Big(2^{-i}\binom{r+t}{i}-2^{-i+1}\binom{r}{i-t}-\sum_{l=i-t+1\atop l\neq \ell}^{r-1}\frac{\binom{r}{l}\binom{t+l}{i}}{2^{t+l}-a}\Big)\alpha^{(0,i)}_{n}(a)\\
&\qquad+\frac{1}{a}\binom{r}{\ell}\sum_{i=0}^{t+\ell-1}\binom{t+\ell}{i}\alpha^{(1,i)}_{n}(a)
\end{align*}
\end{enumerate}
\end{enumerate}

\section{Examples}\label{sec:exm}

 In this section we gather several applications  of our main theorem. The sequences' identifiers, when available, refer to the  OEIS.

\begin{example}\label{ex:constant}
When the independent term in the divide-and-conquer recurrence is constant,
$$
x_n =a\cdot x_{\left\lceil{n}/{2}\right\rceil}+a\cdot x_{\left\lfloor{n}/{2}\right\rfloor}+c
$$
with $c\in \RR$, the solution with initial condition $x_1$ is
\begin{align*}
x_n & =c\cdot x_n^{(0,0)}(a) + \big((2a)^{q_{s_n}(n)}+(2a-1)(na^{q_{s_n}(n)}-(2a)^{q_{s_n}(n)})\big)x_1\\
& = 
\left\{\begin{array}{ll}
c(q_{s_n}(n)+n\cdot 2^{-q_{s_n}(n)}-1)+x_1 & \text{ if $a=\frac{1}{2}$}\\[1ex]
\displaystyle {((2a-1)x_1+c)\Big(na^{q_{s_n}(n)}-\frac{2a-2}{2a-1}(2a)^{q_{s_n}(n)}\Big)-\dfrac{c}{2a-1}}& \text{ if $a\neq \frac{1}{2}$}
\end{array}\right.
\end{align*}
In particular, if $a=1$,
$$
x_n=(x_1+c)n-c.
$$
\end{example}

\begin{example}\label{ex:Sackin}
The minimum total Sackin index $S_n$ of a rooted bifurcating tree with $n$ leaves (the binary entropy function: sequence A003314) satisfies the recurrence
$$
S_n = S_{\left\lceil{n}/{2}\right\rceil} + S_{\left\lfloor{n}/{2}\right\rfloor} + n =
S_{\left\lceil{n}/{2}\right\rceil} + S_{\left\lfloor{n}/{2}\right\rfloor} + \left\lceil{n}/{2}\right\rceil+\left\lfloor{n}/{2}\right\rfloor
$$
with $S_1=0$. Then, according to our main theorem,
$$
S_n = x_n^{(1,0)}(1)+x_n^{(0,1)}(1)
$$
where 
\begin{align*}
x_n^{(1,0)}(1)& =(q_{s_n}(n)+1)n+1-2^{q_{s_n}(n)+1}-\alpha^{(0,0)}_{n}(1)\\[1ex]
x_n^{(0,1)}(1) & = \alpha^{(0,0)}_{n}(1)+n-1
\end{align*}
and therefore
\begin{align*}
S_n & = (q_{s_n}(n)+1)n+1-2^{q_{s_n}(n)+1}-\alpha^{(0,0)}_{n}(1)+\alpha^{(0,0)}_{n}(1)+n-1\\ 
& = (q_{s_n}(n)+2)n-2^{q_{s_n}(n)+1}
\end{align*}
in agreement with previously published formulas: cf.  \cite{FisherS,HJT17}.
\end{example}

\begin{example}\label{ex:Colless}
The minimum Colless index $c_n$ of a rooted bifurcating tree with $n$ leaves (sequence A296062) satisfies the recurrence
$$
c_n  =c_{\lceil{n}/{2}\rceil}+c_{\lfloor{n}/{2}\rfloor}+\left\lceil\frac{n}{2}\right\rceil-\left\lfloor\frac{n}{2}\right\rfloor
$$
with $c_1=0$. Therefore, according to our main theorem
\begin{align*}
c_n &=x_n^{(1,0)}(1)-x_n^{(0,1)}(1)\\
& = (q_{s_n}(n)+1)n+1-2^{q_{s_n}(n)+1}-\alpha^{(0,0)}_{n}(1)-(\alpha^{(0,0)}_{n}(1)+n-1)\\
& = q_{s_n}(n)n-2^{q_{s_n}(n)+1}+2-2\alpha^{(0,0)}_{n}(1)
\end{align*}
Now,
\begin{align}
\alpha^{(0,0)}_{n}(1) & =\frac{1}{2}\sum_{i=1}^{s_n}M_{i}(n)(q_i(n)-q_{i-1}(n)) +\sum_{i=1}^{s_n-1}(n-M_{i}(n))-2^{q_{s_n}(n)}+1\label{eqn:prealpha00}\\
& =\sum_{i=1}^{s_n} 2^{q_i(n)-1}q_{i}(n) +(s_n-1)n-\sum_{j=1}^{s_n}\sum_{i=j}^{s_n} 2^{q_i(n)}+1\nonumber \\& 
=\sum_{i=1}^{s_n} 2^{q_i(n)-1}q_{i}(n) +(s_n-1)n-\sum_{i=1}^{s_n}i 2^{q_i(n)}+1\nonumber \\& = (s_n-1)n+\sum_{i=1}^{s_n}2^{q_i(n)-1}(q_i(n)-2i)+1
\label{eqn:alpha00}
\end{align}
and hence, finally,
\begin{align*}
c_n &= q_{s_n}(n)n-2^{q_{s_n}(n)+1}+2-2\Big((s_n-1)n+\sum_{i=1}^{s_n}2^{q_i(n)-1}(q_i(n)-2i)+1\Big)\nonumber\\
& =\sum_{i=1}^{s_n-1} 2^{q_i(n)} (q_{s_n}(n) - q_i(n) - 2(s_n - i-1)) 
\end{align*}
as it was proved in \cite{MinColless} by showing by induction that this sequence satisfies the recurrence above.
\end{example}

\begin{rem}
The solution of 
$$
x_n =a\cdot x_{\lceil{n}/{2}\rceil}+a\cdot x_{\lfloor{n}/{2}\rfloor}+(\lceil{n}/{2}\rceil-\lfloor{n}/{2}\rfloor)
$$
and of
$$
x_n =a\cdot x_{\lceil{n}/{2}\rceil}+a\cdot x_{\lfloor{n}/{2}\rfloor}+(\lceil{n}/{2}\rceil-\lfloor{n}/{2}\rfloor)^m
$$
for $m\geq 1$, is the same. Therefore, for every $m\geq 1$,
$$
x^{(1,0)}(a)-x^{(0,1)}(a)=\sum_{p=0}^m \binom{m}{p}(-1)^{m-p}x^{(p,m-p)}(a)
$$
\end{rem}

\begin{example}\label{ex:n2}
The sequence $x_n=n^2$ satisfies the equation
$$
x_n  =2x_{\left\lceil{n}/{2}\right\rceil}+2x_{\left\lfloor{n}/{2}\right\rfloor}+\left\lfloor{n}/{2}\right\rfloor-\left\lceil{n}/{2}\right\rceil
$$
with $x_1=1$. Therefore,
$$
x_n=x_n^{(0,1)}(2)-x_n^{(1,0)}(2)+3n2^{q_{s_n}(n)}-2\cdot 4^{q_{s_n}(n)}
$$
Now,
\begin{align*}
x_n^{(0,1)}(2) &=\alpha^{(0,0)}_{n}(2)+\frac{1}{3}\big(3n\cdot 2^{q_{s_n}(n)}-2\cdot 4^{q_{s_n}(n)}-1\big)\\[1ex]
x_n^{(1,0)}(2) & =-n+1 +\frac{2}{3}\big(3n\cdot 2^{q_{s_n}(n)}-2\cdot 4^{q_{s_n}(n)}-1\big)-\alpha^{(0,0)}_{n}(2)\\
\end{align*}
and
\begin{align*}
 &\alpha^{(0,0)}_{n}(2)  =\frac{1}{2}\sum_{i=1}^{s_n} M_{i}(n)\big((2^{q_i(n)}-1)-(2^{q_i(n)-1}-1)\big)+\sum_{i=1}^{s_n-1}2^{q_i(n)} (n-M_{i}(n)) -\frac{4^{q_i(n)}-1}{3}\\
 &\qquad = \frac{1}{2}\sum_{i=1}^{s_n} M_{i}(n)(2^{q_i(n)}-2^{q_{i-1}(n)})+\sum_{i=1}^{s_n}2^{q_i(n)} (n-M_{i}(n)) -\frac{1}{3}(3n2^{q_{s_n}(n)}-2\cdot 4^{q_{s_n}(n)}-1)\\
  &\qquad = n^2-\frac{1}{2}\sum_{i=1}^{s_n} M_{i}(n)2^{q_i(n)}-\frac{1}{2}\sum_{i=1}^{s_n} M_{i}(n)2^{q_{i-1}(n)}-\frac{1}{3}(3n2^{q_{s_n}(n)}-2\cdot 4^{q_{s_n}(n)}-1)\\
&\qquad = n^2-\frac{1}{2}\sum_{i=1}^{s_n} M_{i}(n)2^{q_i(n)}-\frac{1}{2}\sum_{i=0}^{s_n} M_{i+1}(n)2^{q_{i}(n)}-\frac{1}{3}(3n2^{q_{s_n}(n)}-2\cdot 4^{q_{s_n}(n)}-1)\\
 &\qquad = n^2-\frac{1}{2}\sum_{i=1}^{s_n} 2^{q_i(n)}(M_{i}(n)+M_{i+1}(n))-\frac{1}{2}n-\frac{1}{3}(3n2^{q_{s_n}(n)}-2\cdot 4^{q_{s_n}(n)}-1)\\
  &\qquad = n^2-\frac{1}{2}\sum_{i=1}^{s_n} (M_{i}(n)-M_{i+1}(n))(M_{i}(n)+M_{i+1}(n))-\frac{1}{2}n-\frac{1}{3}(3n2^{q_{s_n}(n)}-2\cdot 4^{q_{s_n}(n)}-1)\\
  &\qquad = n^2-\frac{1}{2}\sum_{i=1}^{s_n} (M_{i}(n)^2-M_{i+1}(n)^2)-\frac{1}{2}n-\frac{1}{3}(3n2^{q_{s_n}(n)}-2\cdot 4^{q_{s_n}(n)}-1)\\
 &\qquad = \frac{1}{2}n^2-\frac{1}{2}n-\frac{1}{3}(3n2^{q_{s_n}(n)}-2\cdot 4^{q_{s_n}(n)}-1)
\end{align*}
and therefore
\begin{align*}
x_n & = \alpha^{(0,0)}_{n}(2)+\frac{1}{3}\big(3n\cdot 2^{q_{s_n}(n)}-2\cdot 4^{q_{s_n}(n)}-1\big)\\
& \qquad -\Big(-n+1 +\frac{2}{3}\big(3n\cdot 2^{q_{s_n}(n)}-2\cdot 4^{q_{s_n}(n)}-1\big)-\alpha^{(0,0)}_{n}(2)\Big)+3n2^{q_{s_n}(n)}-2\cdot 4^{q_{s_n}(n)}\\
& = 2\alpha^{(0,0)}_{n}(2)+\frac{2}{3}\big(3n\cdot 2^{q_{s_n}(n)}-2\cdot 4^{q_{s_n}(n)}-1\big)+n\\
& = n^2-n-\frac{2}{3}(3n2^{q_{s_n}(n)}-2\cdot 4^{q_{s_n}(n)}-1)+\frac{2}{3}\big(3n\cdot 2^{q_{s_n}(n)}-2\cdot 4^{q_{s_n}(n)}-1\big)+n\\
& = n^2
\end{align*}
indeed.
\end{example}

\begin{example}
The sequence $\lambda_n$ of Lebesgue constants of the Walsh system
satisfies the recurrence 
$$
\lambda_n =\frac{1}{2}\lambda_{\left\lceil{n}/{2}\right\rceil}+\frac{1}{2}\lambda_{\left\lfloor{n}/{2}\right\rfloor}+\frac{1}{2}\left\lceil\frac{n}{2}\right\rceil-\frac{1}{2}\left\lfloor\frac{n}{2}\right\rfloor
$$
with $\lambda_1=1$ \cite{Fine,Hwang}. Therefore,
$$
\lambda_n =\frac{1}{2}x_n^{(1,0)}(1/2)-\frac{1}{2}x_n^{(0,1)}(1/2)+1
$$
where
\begin{align*}
& x_n^{(1,0)}(1/2) =2n-q_{s_n}(n)-n\cdot 2^{-q_{s_n}(n)}-1 -\alpha^{(0,0)}_{n}(1/2)\\[1ex]
& x_n^{(0,1)}(1/2) = \alpha^{(0,0)}_{n}(1/2)+q_{s_n}(n)+n\cdot 2^{-q_{s_n}(n)}-1\\[1ex]
& \alpha_{n}^{(0,0)}(1/2)\\
&\quad = \frac{1}{2}\sum_{i=1}^{s_n}M_i(n)\big(T(0,q_i(n),1/2)-T(0,q_{i-1}(n),1/2)\big)+\sum_{i=1}^{s_n-1}2^{-q_i(n)} (n-M_{i}(n)) -q_{s_n}(n)\\
&\quad = \sum_{i=1}^{s_n}M_i(n)(2^{-q_{i-1}(n)}-2^{-q_{i}(n)})
+\sum_{i=1}^{s_n-1}2^{-q_i(n)} (n-M_{i}(n)) -q_{s_n}(n)\\
&\quad = \sum_{i=1}^{s_n}M_{i}(n)2^{-q_{i-1}(n)}-\sum_{i=1}^{s_n}M_i(n)2^{-q_{i}(n)}+\sum_{i=1}^{s_n-1}2^{-q_i(n)} (n-M_{i}(n)) -q_{s_n}(n)\\
&\quad = \sum_{i=0}^{s_n}M_{i+1}(n)2^{-q_{i}(n)}-\sum_{i=1}^{s_n}M_i(n)2^{-q_{i}(n)}+\sum_{i=1}^{s_n-1}2^{-q_i(n)} (n-M_{i}(n)) -q_{s_n}(n)\\
&\quad = n+\sum_{i=1}^{s_n}(M_{i+1}(n)-M_i(n))2^{-q_{i}(n)}
+\sum_{i=1}^{s_n-1}2^{-q_i(n)} (n-M_{i}(n)) -q_{s_n}(n)\\
&\quad = n-\sum_{i=1}^{s_n}2^{q_i(n)}2^{-q_{i}(n)}
+\sum_{i=1}^{s_n-1}2^{-q_i(n)} (n-M_{i}(n)) -q_{s_n}(n)\\
&\quad = n-s_n+\sum_{i=1}^{s_n-1}2^{-q_i(n)} (n-M_{i}(n)) -q_{s_n}(n)
\end{align*}
Therefore,
\begin{align*}
\lambda_n & =\frac{1}{2}\big(2n-q_{s_n}(n)-n\cdot 2^{-q_{s_n}(n)}-1 -\alpha^{(0,0)}_{n}(1/2)\big)\\
&\qquad -\frac{1}{2}\big(\alpha^{(0,0)}_{n}(1/2)+q_{s_n}(n)+n\cdot 2^{-q_{s_n}(n)}-1\big)+1 \\
& = -\alpha^{(0,0)}_{n}(1/2)+n-q_{s_n}(n)-n\cdot 2^{-q_{s_n}(n)}+1\\
& = -n+s_n-\sum_{i=1}^{s_n-1}2^{-q_i(n)} (n-M_{i}(n)) +q_{s_n}(n)
+n-q_{s_n}(n)-n\cdot 2^{-q_{s_n}(n)}+1\\
& = s_n-\sum_{i=1}^{s_n}2^{-q_i(n)} (n-M_{i}(n))
\end{align*}
in agreement with \cite[Eqn. (5.5)]{Fine}.
\end{example} 

\begin{rem}\label{rem:Stephan}
As R. Stephan points out in \cite{StephanOEIS}, the OEIS contains many sequences defined by recurrences of the form
$$
\left\{\begin{array}{l}
a_{2n}=C\cdot a_n+C\cdot a_{n-1}+P(n)\\
a_{2n+1}=2C\cdot a_n+Q(n)
\end{array}\right.
$$
for some real number $C$ and functions $P,Q$.  It is straightforward to check then that the sequence $x_n\coloneqq a_{n-1}$  satisfies, for $n\geq 2$, the recurrence  (cf. \cite{HJT17})
$$
x_n=C\cdot x_{\lceil n/2\rceil}+C\cdot x_{\lfloor n/2\rfloor}+Q(\lfloor n/2\rfloor-1)+\big(\lceil n/2\rceil-\lfloor n/2\rfloor\big)\big(P(\lfloor n/2\rfloor)-Q(\lfloor n/2\rfloor-1)\big).
$$
 If $P$ and $Q$ are polynomials, these sequences are covered by our main theorem.
\end{rem}

\begin{example}\label{ex:A005536}
Let $a_n$ denote the sequence A005536. By the second last entry in the last table of \cite{StephanOEIS} and the last remark, the sequence $x_n=a_{n-1}$, for $n\geq 1$,  satisfies the recurrence
$$
x_n=-x_{\left\lceil{n}/{2}\right\rceil}-x_{\left\lfloor{n}/{2}\right\rfloor}+\left\lfloor{n}/{2}\right\rfloor
$$
with $x_1=0$; cf.  also \cite[Ex 7.1]{HJT17}. 
By our main theorem, this sequence is
$$
x_n=x_n^{(0,1)}(-1)
$$
where
\begin{align*}
& x_n^{(0,1)}(-1)  =  \alpha^{(0,0)}_{n}(-1)-\frac{1}{3}((-2)^{q_{s_n}(n)}-1)+ (-1)^{q_{s_n}(n)}n-(-2)^{q_{s_n}(n)}\\[1ex]
&\quad = \alpha^{(0,0)}_{n}(-1)-\frac{4}{3}\cdot (-2)^{q_{s_n}(n)}+\frac{1}{3}+ (-1)^{q_{s_n}(n)}n\\
& \alpha^{(0,0)}_{n}(-1)  =\frac{1}{2}\sum_{i=1}^{s_n} M_{i}(n)\big(T(0,q_i(n),-1)-T(0,q_{i-1}(n),-1)\big)\\
&\hphantom{\alpha^{(0,0)}_{n}(-1)  =}\quad +\sum_{i=1}^{s_n-1}(-1)^{q_i(n)} (n-M_{i}(n)) -T(0,q_{s_n}(n),-2)\\
&\quad =\frac{1}{4}\sum_{i=1}^{s_n} M_{i}(n)\big((-1)^{q_{i-1}(n)}-(-1)^{q_i(n)})\big) +\sum_{i=1}^{s_n-1}(-1)^{q_i(n)} (n-M_{i}(n)) +\frac{1}{3}((-2)^{q_{s_n}(n)}-1)\\
&\quad =\frac{1}{4}\sum_{i=1}^{s_n} M_{i}(n)(-1)^{q_{i-1}(n)}-\frac{1}{4}\sum_{i=1}^{s_n} M_{i}(n)(-1)^{q_i(n)}) +\sum_{i=1}^{s_n-1}(-1)^{q_i(n)} (n-M_{i}(n))\\
&\qquad\quad+\frac{1}{3}((-2)^{q_{s_n}(n)}-1)\\
&\quad =\frac{1}{4}\sum_{i=0}^{s_n} M_{i+1}(n)(-1)^{q_{i}(n)}-\frac{1}{4}\sum_{i=1}^{s_n} M_{i}(n)(-1)^{q_i(n)}+\sum_{i=1}^{s_n-1}(-1)^{q_i(n)} (n-M_{i}(n))\\
&\qquad\quad+\frac{1}{3}((-2)^{q_{s_n}(n)}-1)\\
&\quad =\frac{1}{4}n-\frac{1}{4}\sum_{i=1}^{s_n} (-2)^{q_{i}(n)}+\sum_{i=1}^{s_n-1}(-1)^{q_i(n)} (n-M_{i}(n))+\frac{1}{3}((-2)^{q_{s_n}(n)}-1)
\end{align*}
and hence
\begin{align*}
x_n& =\frac{1}{4}n-\frac{1}{4}\sum_{i=1}^{s_n} (-2)^{q_{i}(n)}+\sum_{i=1}^{s_n-1}(-1)^{q_i(n)} (n-M_{i}(n))+\frac{1}{3}((-2)^{q_{s_n}(n)}-1)\\
&\qquad -\frac{4}{3}\cdot (-2)^{q_{s_n}(n)}+\frac{1}{3}+ (-1)^{q_{s_n}(n)}n\\
&= \frac{1}{4}n-\frac{1}{4}\sum_{i=1}^{s_n}(-2)^{q_i(n)}+\sum_{i=1}^{s_n}(-1)^{q_i(n)} (n-M_{i}(n))\\
&= \frac{1}{4}n-\frac{1}{4}\sum_{i=1}^{s_n}(-2)^{q_i(n)}+\sum_{i=1}^{s_n}\Big((-1)^{q_i(n)}\sum_{j=1}^{i-1}2^{q_j(n)}\Big)\\
&= \frac{1}{4}\sum_{i=1}^{s_n}2^{q_i(n)}\Big(1-(-1)^{q_i(n)}+4\sum_{j=i+1}^{s_n}(-1)^{q_j(n)}\Big)
\end{align*}
which is a simpler expression than the one given in \cite[Ex 7.1]{HJT17}.
\end{example}

\begin{example} 
Let $a_n$ denote the sequence A087733. By Remark \ref{rem:Stephan}, the sequence $x_n=a_{n-1}$, for $n\geq 1$,  satisfies the recurrence
$$
x_n=-x_{\left\lceil{n}/{2}\right\rceil}-x_{\left\lfloor{n}/{2}\right\rfloor}+\left\lfloor{n}/{2}\right\rfloor\left\lceil{n}/{2}\right\rceil
$$
with $x_1=0$.
By our main theorem, this sequence is
\begin{align*}
x_n & =x_n^{(1,1)}(-1)\\
& =\sum_{k=1}^{2}\Big(\sum_{j=k}^{2}\frac{\binom{1}{j-2}\binom{j}{k}B_{j-k}}{j(2^{j-1}+1)}\Big)n^k\\
&\quad +\Big(1-\sum_{l=0}^{0}\frac{\binom{1}{l}}{2^{1+l}+1}\Big)\Big(\frac{(-2)^{q_{s_n}(n)}-1}{-3}+n(-1)^{q_{s_n}(n)}-(-2)^{q_{s_n}(n)}\Big)\\
&\quad +\sum_{i=0}^{1}\Big(2^{-i}\binom{2}{i}-2^{-i+1}\binom{1}{i-1}-\sum_{l=i}^{0}\frac{\binom{1}{l}\binom{1+l}{i}}{2^{1+l}+1}\Big)\alpha^{(0,i)}_{n}(-1)
\\
& =\frac{1}{6}n^2-\frac{1}{6}n+\frac{2}{9}(1+3(-1)^{q_{s_n}(n)}n-4(-2)^{q_{s_n}(n)})+\frac{2}{3}\alpha^{(0,0)}_{n}(-1)\\
& =\frac{1}{6}n^2-\frac{1}{6}n+\frac{2}{9}(1+3(-1)^{q_{s_n}(n)}n-4(-2)^{q_{s_n}(n)})\\
&\quad +\frac{2}{3}\Big[\frac{1}{4}n-\frac{1}{4}\sum_{i=1}^{s_n} (-2)^{q_{i}(n)}+\sum_{i=1}^{s_n-1}(-1)^{q_i(n)} (n-M_{i}(n))+\frac{1}{3}((-2)^{q_{s_n}(n)}-1)\Big]\\
& \text{(cf. Example \ref{ex:A005536})}\\
& = \frac{1}{6}n^2-\frac{1}{6}\sum_{i=1}^{s_n}(-2)^{q_i(n)}+\frac{2}{3}\sum_{i=1}^{s_n}(-1)^{q_i(n)}(n-M_i(n))\\
& = \frac{1}{6}n^2-\frac{1}{6}\sum_{i=1}^{s_n}(-2)^{q_i(n)}+\frac{2}{3}\sum_{i=1}^{s_n-1}\Big(2^{q_i(n)}\sum_{j=i+1}^{s_n}(-1)^{q_j(n)}\Big)
\end{align*}
To our knowledge, no closed expression had been published so far for this sequence.
\end{example}

\begin{example}\label{ex:cophenetic}
The minimum total cophenetic index $\Phi_n$ of a rooted bifurcating tree with $n$ leaves (sequence A174605) satisfies the recurrence 
\begin{align*}
\Phi_n & = \Phi_{\left\lceil{n}/{2}\right\rceil} +\Phi_{\left\lfloor{n}/{2}\right\rfloor}+\binom{\left\lceil{n}/{2}\right\rceil}{2}+\binom{\left\lfloor{n}/{2}\right\rfloor}{2}\\
& = \Phi_{\left\lceil{n}/{2}\right\rceil} +\Phi_{\left\lfloor{n}/{2}\right\rfloor}+\frac{1}{2}\left\lceil\frac{n}{2}\right\rceil^2-\frac{1}{2}\left\lceil\frac{n}{2}\right\rceil+
\frac{1}{2}\left\lfloor\frac{n}{2}\right\rfloor^2-\frac{1}{2}\left \lfloor\frac{n}{2}\right\rfloor
\end{align*}
with initial condition $\Phi_1=0$ \cite{cophenetic}. Therefore
$$
\Phi_n=\frac{1}{2} (x_n^{(2,0)}(1)+x_n^{(0,2)}(1)-x_n^{(1,0)}(1)-x_n^{(0,1)}(1))=
\frac{1}{2} (x_n^{(2,0)}(1)+x_n^{(0,2)}(1)-S_n)
$$
where 
$$
S_n=x_n^{(1,0)}(1)+x_n^{(0,1)}(1)=(q_{s_n}(n)+2)n-2^{q_{s_n}(n)+1}
$$
was already computed in Example \ref{ex:Sackin} and
\begin{align*}
x_n^{(2,0)}(1)&  =n^2+(q_{s_n}(n)-2)n+3-2^{q_{s_n}(n)+1}-3\alpha^{(0,0)}_{n}(1)-\alpha^{(0,1)}_{n}(1)
\\[1ex]
x_n^{(0,2)}(1)& = \alpha^{(0,0)}_{n}(1)+\alpha^{(0,1)}_{n}(1)+n-1 
\end{align*}
Therefore
\begin{align*}
\Phi_n& =\frac{1}{2} \Big(n^2+(q_{s_n}(n)-2)n+3-2^{q_{s_n}(n)+1}-3\alpha^{(0,0)}_{n}(1)-\alpha^{(0,1)}_{n}(1) \\
&\qquad +\alpha^{(0,0)}_{n}(1)+\alpha^{(0,1)}_{n}(1)+n-1-\big((q_{s_n}(n)+2)n-2^{q_{s_n}(n)+1}\big)\Big)\\
&=\frac{1}{2}(n^2-3n+2-2\alpha^{(0,0)}_{n}(1))\\
&=\binom{n-1}{2}-\Big(s_nn+\sum_{i=1}^{s_n}2^{q_i(n)-1}(q_i(n)-2i)-n+1\Big)\\
&\text{(by Eqn.\ (\ref{eqn:alpha00}))}\\
& =\binom{n}{2}{-}s_nn-\sum_{i=1}^{s_n}2^{q_i(n)-1}(q_i(n)-2i)
\end{align*}
No closed expression had been published so far for this sequence.
\end{example}

\begin{example}\label{ex:rqi}
The maximum rooted quartet index $\rho_n$ of a rooted bifurcating tree with $n$ leaves (sequence A300445) satisfies the recurrence 
\begin{align*}
\rho_n & = \rho_{\lceil n/2\rceil} + \rho_{\lfloor n/2\rfloor} + \binom{\lceil n/2\rceil}{2}\binom{\lfloor n/2\rfloor}{2}\\
& =\rho_{\lceil n/2\rceil} + \rho_{\lfloor n/2\rfloor} + \frac{1}{4}\Big(\left\lceil\frac{n}{2}\right\rceil\left\lfloor\frac{n}{2}\right\rfloor-\left\lceil\frac{n}{2}\right\rceil^2\left\lfloor\frac{n}{2}\right\rfloor-\left\lceil\frac{n}{2}\right\rceil\left\lfloor\frac{n}{2}\right\rfloor^2+\left\lceil\frac{n}{2}\right\rceil^2\left\lfloor\frac{n}{2}\right\rfloor^2\Big)
\end{align*}
with initial condition $\rho_1=0$ \cite{rQI}. Therefore
$$
\rho_n=\frac{1}{4} (x_n^{(1,1)}(1)-x_n^{(2,1)}(1)-x_n^{(1,2)}(1)+x_n^{(2,2)}(1))
$$
where 
\begin{align*}
x_n^{(1,1)}(1)&=\frac{1}{2}n^2+(1+B_1-1)n +(1-0-1)\alpha^{(0,0)}_{n}(1)+(1-1-0)\alpha^{(0,1)}_{n}(1) = \binom{n}{2}
\\[1ex]
x_n^{(2,1)}(1)&=\sum_{k=2}^{3}\Big(\sum_{j=k}^{3}\frac{\binom{2}{j-2}\binom{j}{k}B_{j-k}}{j(2^{j-1}-1)}\Big)n^k +\Big(1+\sum_{l=1}^{2}\frac{\binom{2}{l-1}(B_l-1)}{2^{l}-1}\Big)n+\sum_{l=0}^{1}\frac{\binom{2}{l}}{2^{1+l}-1}-1\nonumber\\
&\quad +\sum_{i=0}^{2}\Big(2^{-i}\binom{3}{i}-2^{-i+1}\binom{2}{i-1}-\sum_{l=i}^{1}\frac{\binom{2}{l}\binom{1+l}{i}}{2^{1+l}-1}\Big)\alpha^{(0,i)}_{n}(1)\nonumber\\
&=\frac{2}{9}n^3+\frac{1}{6}n^2-\frac{19}{18}n+\frac{2}{3} -\frac{2}{3}\alpha^{(0,0)}_{n}(1) -\frac{5}{6}\alpha^{(0,1)}_{n}(1)-\frac{1}{4}\alpha^{(0,2)}_{n}(1)\nonumber\\[1ex]
x_n^{(1,2)}& =\sum_{k=2}^{3}\Big(\sum_{j=k}^{3}\frac{\binom{1}{j-3}\binom{j}{k}B_{j-k}}{j(2^{j-1}-1)}\Big)n^k+\Big(1+\sum_{l=1}^{2}\frac{\binom{1}{l-2}(B_l-1)}{2^{l}-1}\Big)n+\sum_{l=0}^{0}\frac{\binom{1}{l}}{2^{2+l}-1}-1\nonumber\\
&\quad +\sum_{i=0}^{2}\Big(2^{-i}\binom{3}{i}-2^{-i+1}\binom{1}{i-2}-\sum_{l=i-1}^{0}\frac{\binom{1}{l}\binom{2+l}{i}}{2^{2+l}-1}\Big)\alpha^{(0,i)}_{n}(1)\nonumber\\
& = \frac{1}{9}n^3-\frac{1}{6}n^2+\frac{13}{18}n-\frac{2}{3}+\frac{2}{3}\alpha^{(0,0)}_{n}(1)
+\frac{5}{6}\alpha^{(0,1)}_{n}(1)+\frac{1}{4}\alpha^{(0,2)}_{n}(1)
\\[1ex]
x_n^{(2,2)}& =\sum_{k=2}^{4}\Big(\sum_{j=k}^{4}\frac{\binom{2}{j-3}\binom{j}{k}B_{j-k}}{j(2^{j-1}-1)}\Big)n^k +\Big(1+\sum_{l=1}^{3}\frac{\binom{2}{l-2}(B_l-1)}{2^{l}-1}\Big)n+\sum_{l=0}^{1}\frac{\binom{2}{l}}{2^{2+l}-1}-1\nonumber\\
&\quad +\sum_{i=0}^{3}\Big(2^{-i}\binom{4}{i}-2^{-i+1}\binom{2}{i-2}-\sum_{l=i-1}^{1}\frac{\binom{2}{l}\binom{2+l}{i}}{2^{2+l}-1}\Big)\alpha^{(0,i)}_{n}(1)\nonumber\\
& = \frac{1}{14} n^4-\frac{2}{63}n^3-\frac{2}{21}n^2+\frac{55}{126}n-\frac{8}{21}+\frac{8}{21}\alpha^{(0,0)}_{n}(1)+\frac{10}{21}\alpha^{(0,1)}_{n}(1)+\frac{1}{7}\alpha^{(0,2)}_{n}(1)
\end{align*}
Therefore,
\begin{align}
\rho_n & =\frac{1}{4}\Big[\binom{n}{2}-\Big(\frac{2}{9}n^3+\frac{1}{6}n^2-\frac{19}{18}n+\frac{2}{3} -\frac{2}{3}\alpha^{(0,0)}_{n}(1) -\frac{5}{6}\alpha^{(0,1)}_{n}(1)-\frac{1}{4}\alpha^{(0,2)}_{n}(1)\Big)\nonumber\\
&\qquad -\Big(\frac{1}{9}n^3-\frac{1}{6}n^2+\frac{13}{18}n-\frac{2}{3}+\frac{2}{3}\alpha^{(0,0)}_{n}(1)
+\frac{5}{6}\alpha^{(0,1)}_{n}(1)+\frac{1}{4}\alpha^{(0,2)}_{n}(1)\Big)\nonumber\\
&\qquad+\Big(\frac{1}{14} n^4-\frac{2}{63}n^3-\frac{2}{21}n^2+\frac{55}{126}n-\frac{8}{21}+\frac{8}{21}\alpha^{(0,0)}_{n}(1)+\frac{10}{21}\alpha^{(0,1)}_{n}(1)+\frac{1}{7}\alpha^{(0,2)}_{n}(1)\Big)\Big]\nonumber
\\
& = \frac{1}{504} \big[(n-3) (n-2) (n-1) (9 n+8)+48 \alpha^{(0,0)}_{n}(1)+60 \alpha^{(0,1)}_{n}(1)+18\alpha^{(0,2)}_{n}(1)\big]
\label{eqn:rho1}
\end{align}
Now,
\begin{align*}
& \alpha^{(0,0)}_{n}(1) = \frac{1}{2}\sum_{i=1}^{s_n}M_{i}(n)\big(q_i(n)-q_{i-1}(n)\big) +\sum_{i=1}^{s_n-1}(n-M_{i}(n)) -2^{q_{s_n}(n)}+1\quad \text{(by Eqn. (\ref{eqn:prealpha00}))}\\
&\quad = \frac{1}{2}\sum_{i=1}^{s_n}M_{i}(n)\big(q_i(n)-q_{i-1}(n)\big) +\sum_{i=1}^{s_n}(n-M_{i}(n)) -n+1
\\[1ex]
& \alpha_n^{(0,1)}(1)  =\frac{1}{4}\sum_{i=1}^{s_n}\Big(M_i(n)^2\big(T(0,q_i(n),1/2)-T(0,q_{i-1}(n),1/2)\big)+4B_1M_i(n)(q_i(n)-q_{i-1}(n))\Big)\\ &\qquad  +\sum_{i=1}^{s_n-1}2^{-q_i(n)}(n-M_i(n))M_{i+1}(n)\\
&\quad = \frac{1}{2}\sum_{i=1}^{s_n}M_i(n)^2\left(2^{-q_{i-1}(n)}-2^{-q_{i}(n)}\right)-\frac{1}{2}\sum_{i=1}^{s_n}M_i(n)(q_i(n)-q_{i-1}(n))+\sum_{i=1}^{s_n}2^{-q_i(n)}M_{i+1}(n)\\[1ex]
&\alpha_n^{(0,2)}(1)  =\frac{1}{6}\sum_{i=1}^{s_n}\Big(M_i(n)^3\big(T(0,q_i(n),1/4)-T(0,q_{i-1}(n),1/4)\big)\\ &\qquad +6B_1 M_i(n)^2\big(T(0,q_i(n),1/2)-T(0,q_{i-1}(n),1/2)\big)+12B_2 M_i(n) (q_i(n)-q_{i-1}(n))\Big)\\ & \qquad  +\sum_{i=1}^{s_n-1}4^{-q_i(n)}(n-M_i(n))M_{i+1}(n)^2\\
&\quad = \frac{2}{9}\sum_{i=1}^{s_n}M_i(n)^3 (4^{-q_{i-1}(n)}-4^{-q_{i}(n)})-\sum_{i=1}^{s_n}M_i(n)^2 (2^{-q_{i-1}(n)}-2^{-q_{i}(n)})\\ & \qquad +\frac{1}{3}\sum_{i=1}^{s_n}M_i(n)(q_i-q_{i-1})+\sum_{i=1}^{s_n}4^{-q_i(n)}M_{i+1}(n)^2(n-M_i(n))
\end{align*}
and then, substituting these expressions in Eqn.~(\ref{eqn:rho1}),  
we finally obtain
\begin{align*}
\rho_n & =
\frac{1}{252} \Big[3(9 n+8)\binom{n-1}{3} -24n+24\\
& \qquad\qquad +2 \sum_{i=1}^{s_n}M_i(n)^3 (4^{-q_{i-1}(n)}-4^{-q_{i}(n)})+6\sum_{i=1}^{s_n}M_i(n)^2 (2^{-q_{i-1}(n)}-2^{-q_{i}(n)})\\
& \qquad\qquad +\sum_{i=1}^{s_n}(n-M_i(n))(24+30\cdot 2^{-q_i(n)}M_{i+1}(n)+9\cdot 4^{-q_i(n)}M_{i+1}(n)^2)\Big]
\end{align*}
This expression can be simplified as follows. Notice that
\begin{align*}
& 2 \sum_{i=1}^{s_n}M_i(n)^3 (4^{-q_{i-1}(n)}-4^{-q_{i}(n)})+6\sum_{i=1}^{s_n}M_i(n)^2 (2^{-q_{i-1}(n)}-2^{-q_{i}(n)})\\
&\quad= 2\Big( \sum_{i=0}^{s_n}M_{i+1}(n)^34^{-q_{i}(n)}- \sum_{i=1}^{s_n}M_i(n)^34^{-q_{i}(n)}\Big)\\
&\qquad\quad+6\Big(\sum_{i=0}^{s_n-1}M_{i+1}(n)^2 2^{-q_{i}(n)}-\sum_{i=1}^{s_n}M_i(n)^22^{-q_{i}(n)}\Big)\\
&\quad= 2n^3+2 \sum_{i=1}^{s_n}4^{-q_{i}(n)}(M_{i+1}(n)^3- M_i(n)^3)\\
&\qquad\quad +6n^2+6\sum_{i=1}^{s_n}2^{-q_{i}(n)}(M_{i+1}(n)^2-M_i(n)^2)\\
&\quad= 2n^3+6n^2+2 \sum_{i=1}^{s_n}4^{-q_{i}(n)}(-3 M_i(n)^22^{q_i(n)}+3 M_i(n)4^{q_i(n)}-8^{q_i(n)})\\
&\qquad\quad +6\sum_{i=1}^{s_n}2^{-q_{i}(n)}(-2M_i(n)2^{q_i(n)}+4^{q_i(n)})\\
&\quad= 2n^3+6n^2-2 \sum_{i=1}^{s_n}(3 M_i(n)^22^{-q_i(n)}+3 M_i(n)-2\cdot 2^{{q_i(n)}})\\
&\quad= 2n^3+6n^2+4n-6 \sum_{i=1}^{s_n}M_i(n)(M_i(n)2^{-q_i(n)}+1)
\end{align*}
and
\begin{align*}
&\sum_{i=1}^{s_n}(n-M_i(n))(24+30\cdot 2^{-q_i(n)}M_{i+1}(n)+9\cdot 4^{-q_i(n)}M_{i+1}(n)^2)\\
&\quad=\sum_{i=1}^{s_n}(n-M_i(n))\big(24+30\cdot 2^{-q_i(n)}(M_{i}(n)-2^{q_i(n)})+9\cdot 4^{-q_i(n)}(M_{i}(n)-2^{q_i(n)})^2\big)\\
&\quad=\sum_{i=1}^{s_n}(n-M_i(n))(3+12\cdot 2^{-q_i(n)}M_{i}(n)+9\cdot 4^{-q_i(n)}M_{i}(n)^2)\\
&\quad=3n\sum_{i=1}^{s_n}(1+4\cdot 2^{-q_i(n)}M_{i}(n)+3\cdot 4^{-q_i(n)}M_{i}(n)^2)\\
&\qquad\quad
-3\sum_{i=1}^{s_n}M_{i}(n)\big(1+4\cdot 2^{-q_i(n)}M_{i}(n)+3\cdot 4^{-q_i(n)}M_{i}(n)^2)
\end{align*}
and hence, finally
\begin{align*}
\rho_n & =
\frac{1}{252} \Big[3(9 n+8)\binom{n-1}{3}-24n+24 +2n^3+6n^2+4n\\
&\quad-6 \sum_{i=1}^{s_n}M_i(n)(M_i(n)2^{-q_i(n)}+ 1)  +3n\sum_{i=1}^{s_n}(1+4\cdot 2^{-q_i(n)}M_{i}(n)+3\cdot 2^{-2q_i(n)}M_{i}(n)^2)\\
&\quad
-3\sum_{i=1}^{s_n}M_{i}(n)\big(1+4\cdot 2^{-q_i(n)}M_{i}(n)+3\cdot 2^{-2q_i(n)}M_{i}(n)^2)\Big]\\
& =
\frac{1}{504} \Big[9n^4-42n^3+63n^2-6n+6n\sum_{i=1}^{s_n}(1+M_{i}(n)2^{-q_i(n)})(1+3M_{i}(n)2^{-q_i(n)})\\
&\qquad -18\sum_{i=1}^{s_n}{M_{i}(n)}(1+M_{i}(n)2^{-q_i(n)})^2\Big]
\end{align*}
where notice that $M_i(n)2^{-q_i(n)}=\lfloor n/2^{q_i(n)}\rfloor$.
Again, no closed expression  was known so far for this sequence, either.
\end{example}

\begin{example}
Let $a_{1,n}$ and $a_{3,n}$ denote the sequences A006581 and A006583, respectively.  By Remark \ref{rem:Stephan}, the sequences $\sigma_{i,n}=a_{i,n-1}$, for $i=1,3$, are the solutions with $\sigma_{1,1}=\sigma_{3,1}=0$ of the recurrences
$$
\sigma_{i,n} =2\sigma_{i,\left\lceil{n}/{2}\right\rceil}+2\sigma_{i,\left\lfloor{n}/{2}\right\rfloor}+g_i(n)
$$
where
\begin{align*}
g_1(n) & =\left\lfloor\frac{n}{2}\right\rfloor\Big(\left\lceil\frac{n}{2}\right\rceil-\left\lfloor\frac{n}{2}\right\rfloor\Big)=\left\lfloor\frac{n}{2}\right\rfloor\left\lceil\frac{n}{2}\right\rceil-\left\lfloor\frac{n}{2}\right\rfloor^2\\[1ex]
g_3(n) & =6\Big(\left\lfloor\frac{n}{2}\right\rfloor-1\Big)+\Big(\left\lceil\frac{n}{2}\right\rceil-\left\lfloor\frac{n}{2}\right\rfloor\Big)\Big(5\left\lfloor\frac{n}{2}\right\rfloor-4-6\left\lfloor\frac{n}{2}\right\rfloor+6\Big)\\
& = 4\left\lfloor\frac{n}{2}\right\rfloor+2\left\lceil\frac{n}{2}\right\rceil-6-g_1(n)
\end{align*}
Therefore,
\begin{align*}
\sigma_{1,n}& =x_n^{(1,1)}(2)- x_n^{(0,2)}(2)\\
\sigma_{3,n}& =4x_n^{(0,1)}(2)+2 x_n^{(1,0)}(2)-6x_n^{(0,0)}(2)-\sigma_{1,n}
\end{align*}
where
\begin{align*}
x_n^{(0,0)}(2) & =\frac{1}{3}\big(3n\cdot 2^{q_{s_n}(n)}-2\cdot 4^{q_{s_n}(n)}-1\big)\quad \text{(cf. Example \ref{ex:constant})}\\[1ex]
x_n^{(0,1)}(2) &=\alpha^{(0,0)}_{n}(2)+x_n^{(0,0)}(2)\\[1ex]
x_n^{(0,2)}(2) & =\alpha^{(0,0)}_{n}(2)+ \alpha^{(0,1)}_{n}(2)+x_n^{(0,0)}(2) \\[1ex]
x_n^{(1,0)}(2) & =-n +\frac{2}{3}\big(3n\cdot 2^{q_{s_n}(n)}-2\cdot 4^{q_{s_n}(n)}-1\big)+1-\alpha^{(0,0)}_{n}(2)\\
&=2x_n^{(0,0)}(2)-n+1-\alpha^{(0,0)}_{n}(2)\\
 x_n^{(1,1)}(2) & =\sum_{k=1}^{2}\Big(\sum_{j=k\atop j\neq 2}^{2}\frac{\binom{1}{j-2}\binom{j}{k}B_{j-k}}{j(2^{j-1}-2)}\Big)n^k\\
&\quad +\frac{1}{2}\Big(\dfrac{(q_{s_n}(n)-1)4^{q_{s_n}(n)+1}-q_{s_n}(n)4^{q_{s_n}(n)}+4}{9}+nq_{s_n}(n)2^{q_{s_n}(n)}-q_{s_n}(n)4^{q_{s_n}(n)}\Big)\\
&\quad +\Big(1-\sum_{l=0\atop l\neq 0}^{0}\frac{\binom{1}{l}}{2^{1+l}-2}\Big)\frac{1}{3}\big(3n\cdot 2^{q_{s_n}(n)}-2\cdot 4^{q_{s_n}(n)}-1\big)\nonumber\\
&\quad +\sum_{i=0}^{1}\Big(2^{-i}\binom{2}{i}-2^{-i+1}\binom{1}{i-1}-\sum_{l=i\atop l\neq 0}^{0}\frac{\binom{1}{l}\binom{1+l}{i}}{2^{1+l}-2}\Big)\alpha^{(0,i)}_{n}(2)+\frac{1}{2}\sum_{i=0}^{0}\binom{1}{i}\alpha^{(1,i)}_{n}(a)\\
& =\frac{1}{18}(9nq_{s_n}(n)2^{q_{s_n}(n)}-6q_{s_n}(n)4^{q_{s_n}(n)}-4\cdot 4^{q_{s_n}(n)}+4)\\
&\quad +\frac{1}{3}\big(3n\cdot 2^{q_{s_n}(n)}-2\cdot 4^{q_{s_n}(n)}-1\big)+\alpha^{(0,0)}_{n}(2)+\frac{1}{2}\alpha^{(1,0)}_{n}(2)
\end{align*}
and
\begin{align*}
 \alpha^{(0,0)}_{n}(2) & =\frac{1}{2}\sum_{i=1}^{s_n} M_{i}(n)\big(T(0,q_i(n),2)-T(0,q_{i-1}(n),2)\big)+\sum_{i=1}^{s_n-1}2^{q_i(n)} (n-M_{i}(n)) -T(0,q_{s_n}(n),4)\\
 & = \frac{1}{2}\sum_{i=1}^{s_n} M_{i}(n)(2^{q_i(n)}-2^{q_{i-1}(n)})+\sum_{i=1}^{s_n-1}2^{q_i(n)} (n-M_{i}(n)) -\frac{1}{3}(4^{q_{s_n}(n)}-1)
\\[1ex]
 \alpha^{(0,1)}_{n}(2) & = \frac{1}{4}\sum_{i=1}^{s_n}M_{i}(n)^{2}(q_i(n)-q_{i-1}(n))-
 \frac{1}{2}\sum_{i=1}^{s_n}M_{i}(n)(2^{q_i(n)}-2^{q_{i-1}(n)})\\
 &\quad
+\sum_{i=1}^{s_n-1} (n-M_{i}(n))M_{i+1}(n) 
 \\[1ex]
 \alpha^{(1,0)}_{n}(2) & 
 =\frac{1}{2}\sum_{i=1}^{s_n}M_{i}(n)\big(T(1,q_i(n),2)-T(1,q_{i-1}(n),2)\big)
+\sum_{i=1}^{s_n-1}q_i(n)2^{q_i(n)} (n-M_{i}(n))\\
&\quad -T(1,q_{s_n}(n),4)\\
& =\frac{1}{2}\sum_{i=1}^{s_n}M_{i}(n)((q_{i}(n)-2)2^{q_{i}(n)} -(q_{i-1}(n)-2)2^{q_{i-1}(n)})\\
&\quad
+\sum_{i=1}^{s_n-1}q_i(n)2^{q_i(n)} (n-M_{i}(n)) -\frac{1}{9}\big((3q_{s_n}(n)-4)4^{q_{s_n}(n)}+4\big) 
\end{align*}
Therefore, finally,
\begin{align*}
\sigma_{1,n}& =x_n^{(1,1)}(2)- x_n^{(0,2)}(2)\\ & =\frac{1}{18}(9nq_{s_n}(n)2^{q_{s_n}(n)}-6q_{s_n}(n)4^{q_{s_n}(n)}-4\cdot 4^{q_{s_n}(n)}+4)\\
&\quad +\frac{1}{3}\big(3n\cdot 2^{q_{s_n}(n)}-2\cdot 4^{q_{s_n}(n)}-1\big)+\alpha^{(0,0)}_{n}(2)+\frac{1}{2}\alpha^{(1,0)}_{n}(2)\\
&\quad -\alpha^{(0,0)}_{n}(2)- \alpha^{(0,1)}_{n}(2)-\frac{1}{3}\big(3n\cdot 2^{q_{s_n}(n)}-2\cdot 4^{q_{s_n}(n)}-1\big)\\
& =\frac{1}{18}(9nq_{s_n}(n)2^{q_{s_n}(n)}-6q_{s_n}(n)4^{q_{s_n}(n)}-4\cdot 4^{q_{s_n}(n)}+4) +\frac{1}{2}\alpha^{(1,0)}_{n}(2) - \alpha^{(0,1)}_{n}(2)\\
& =\frac{1}{18}(9nq_{s_n}(n)2^{q_{s_n}(n)}-6q_{s_n}(n)4^{q_{s_n}(n)}-4\cdot 4^{q_{s_n}(n)}+4)\\
&\quad +\frac{1}{4}\sum_{i=1}^{s_n}M_{i}(n)((q_{i}(n)-2)2^{q_{i}(n)} -(q_{i-1}(n)-2)2^{q_{i-1}(n)})\\
&\quad
+\frac{1}{2}\sum_{i=1}^{s_n-1}q_i(n)2^{q_i(n)} (n-M_{i}(n)) -\frac{1}{18}(3q_{s_n}(n)4^{q_{s_n}(n)}-4\cdot 4^{q_{s_n}(n)}+4)\\
&\quad - \frac{1}{4}\sum_{i=1}^{s_n}M_{i}(n)^{2}(q_i(n)-q_{i-1}(n))+
 \frac{1}{2}\sum_{i=1}^{s_n}M_{i}(n)(2^{q_i(n)}-2^{q_{i-1}(n)})\\
 &\quad 
-\sum_{i=1}^{s_n-1} (n-M_{i}(n))M_{i+1}(n) \\
& =\frac{1}{2}(nq_{s_n}(n)2^{q_{s_n}(n)}-q_{s_n}(n)4^{q_{s_n}(n)})+\frac{1}{4}\sum_{i=1}^{s_n}M_{i}(n)(q_{i}(n)2^{q_{i}(n)} -q_{i-1}(n)2^{q_{i-1}(n)})\\
&\quad - \frac{1}{4}\sum_{i=1}^{s_n}M_{i}(n)^{2}(q_i(n)-q_{i-1}(n))
+\frac{1}{2}\sum_{i=1}^{s_n-1}q_i(n)2^{q_i(n)}(n-M_i(n))\\
 &\quad 
-\sum_{i=1}^{s_n-1} (n-M_{i}(n))M_{i+1}(n)\\
& =\frac{1}{4}\sum_{i=1}^{s_n}M_{i}(n)(q_{i}(n)2^{q_{i}(n)} -q_{i-1}(n)2^{q_{i-1}(n)}) - \frac{1}{4}\sum_{i=1}^{s_n}M_{i}(n)^{2}(q_i(n)-q_{i-1}(n))\\
&\quad
+\frac{1}{2}\sum_{i=1}^{s_n}q_i(n)2^{q_i(n)}(n-M_i(n))
-\sum_{i=1}^{s_n} (n-M_{i}(n))M_{i+1}(n)\\
& =\frac{1}{4}\sum_{i=1}^{s_n}M_{i}(n)q_{i}(n)(2^{q_{i}(n)}-M_i(n))
+\frac{1}{4}\sum_{i=1}^{s_n}M_{i}(n)q_{i-1}(n)(M_i(n)+2^{q_{i-1}(n)})
\\
&\quad
-\frac{1}{2}\sum_{i=1}^{s_n}M_{i}(n)q_{i-1}(n)2^{q_{i-1}(n)}
+\frac{1}{2}\sum_{i=1}^{s_n}q_i(n)2^{q_i(n)}(n-M_{i+1}(n)-2^{q_{i}(n)})\\
 &\quad 
-\sum_{i=1}^{s_n} (n-M_{i}(n))M_{i+1}(n)\\
& =-\frac{1}{4}\sum_{i=1}^{s_n}M_{i}(n)q_{i}(n)M_{i+1}(n)
+\frac{1}{4}\sum_{i=1}^{s_n}M_{i}(n)q_{i-1}(n)M_{i-1}(n)
\\
&\quad
-\frac{1}{2}\sum_{i=0}^{s_n}M_{i+1}(n)q_{i}(n)2^{q_{i}(n)}
-\frac{1}{2}\sum_{i=1}^{s_n}q_i(n)2^{q_i(n)}M_{i+1}(n)\\
&\quad 
+\frac{1}{2}\sum_{i=1}^{s_n}q_i(n)2^{q_i(n)}(n-2^{q_{i}(n)})
-\sum_{i=1}^{s_n} (n-M_{i}(n))M_{i+1}(n)
\end{align*}
\begin{align*}
& = 
-\frac{1}{4}\sum_{i=1}^{s_n}M_{i}(n)q_{i}(n)M_{i+1}(n)
+\frac{1}{4}\sum_{i=0}^{s_n-1}M_{i+1}(n)q_{i}(n)M_{i}(n)
\\
&\quad -\sum_{i=1}^{s_n}q_i(n)2^{q_i(n)}M_{i+1}(n)+\frac{1}{2}\sum_{i=1}^{s_n}q_i(n)2^{q_i(n)}(n-2^{q_{i}(n)})
-\sum_{i=1}^{s_n} (n-M_{i}(n))M_{i+1}(n)\\
& = \frac{1}{2}\sum_{i=1}^{s_n}q_i(n)2^{q_i(n)}(n-2^{q_{i}(n)}-2M_{i+1}(n))
-\sum_{i=1}^{s_n} (n-M_{i}(n))M_{i+1}(n)
\end{align*}
and
\begin{align*}
\sigma_{3,n}&=4x_n^{(0,1)}(2)+2 x_n^{(1,0)}(2)-6x_n^{(0,0)}(2)-\sigma_{1,n}\\
&=4(\alpha^{(0,0)}_{n}(2)+x_n^{(0,0)}(2))+2(2x_n^{(0,0)}(2)-n+1-\alpha^{(0,0)}_{n}(2))-6x_n^{(0,0)}(2)-\sigma_{1,n}\\
& =2\alpha^{(0,0)}_{n}(2)+2x_n^{(0,0)}(2))-2n+2-\sigma_{1,n}
\\
& = \sum_{i=1}^{s_n} M_{i}(n)(2^{q_i(n)}-2^{q_{i-1}(n)})+2\sum_{i=1}^{s_n-1}2^{q_i(n)} (n-M_{i}(n)) -\frac{2}{3}(4^{q_{s_n}(n)}-1)\\
&\quad +\frac{2}{3}\big(3n\cdot 2^{q_{s_n}(n)}-2\cdot 4^{q_{s_n}(n)}-1\big)-2n+2-\sigma_{1,n}
\\
&=\sum_{i=1}^{s_n} M_{i}(n)(2^{q_i(n)}-2^{q_{i-1}(n)})+2\sum_{i=1}^{s_n}2^{q_i(n)} (n-M_{i}(n)) -2n+2-\sigma_{1,n}\\
&=-\sum_{i=1}^{s_n} M_{i}(n)2^{q_{i-1}(n)}+n^2+\sum_{i=1}^{s_n}2^{q_i(n)} (n-M_{i}(n)) -2n+2-\sigma_{1,n}\\
&=2\binom{n-1}{2}  -\sigma_{1,n}
\end{align*}
Again, no closed formulas for these sequences had been published so far.
\end{example}

\begin{rem} 
Notice that $x_{2^m}^{(r,t)}=2ax_{2^{m-1}}^{(r,t)}+2^{(m-1)(r+t)}$ and therefore, since $x_1^{(r,t)}=0$,
\begin{align*}
x_{2^m}^{(r,t)}(a) & =\sum_{k=0}^{m-1} (2a)^k2^{(r+t)(m-k-1)}=
2^{(r+t)(m-1)}\sum_{k=0}^{m-1} (2^{-r-t+1}a)^k\\
& =2^{(r+t)(m-1)}T(0,m,2^{-r-t+1}a)=
\left\{\begin{array}{ll}
2^{(r+t)(m-1)}m & \text{if $a=2^{r+t-1}$, i.e., if $\ell=r-1$}\\[1ex]
\dfrac{(2a)^m-2^{m(r+t)}}{2a-2^{r+t}}& \text{if $a\neq 2^{r+t-1}$}
\end{array}\right.
\end{align*}
We have checked that our general formula for $x_{n}^{(r,t)}(a)$ satisfies this equality when $n=2^m$ with \textsl{Mathematica}.
\end{rem}

\section{Proof of the main result}\label{sec:proof}

\subsection{Some notations}

\noindent\textbf{1.} Throughout this paper we shall use the following notations related  to binary decompositions of natural numbers. For every $n\in \NN$, we shall write its binary decomposition as 
$$
n=\sum_{j=1}^{s_n} 2^{q_j(n)},\quad \text{ with }0\leq q_1(n)<\cdots<q_{s_{n}}(n);
$$
if $n=0$, we set $s_0=0$. With these notations, $s_n$ is the \emph{binary weight} of $n$, that is, the number of $1$'s in the binary representation of $n$, and, if $n\geq 1$, $q_{s_n}(n)=\lfloor \log_2(n)\rfloor$. In order to simplify the notations, we shall set $q_0(n)=0$.

Notice that, for every $M\in \NN_{\geq 2}$ and for every natural number $1\leq p<2^{q_1(M)}$, 
$$
M+p=\sum_{i=1}^{s_p} 2^{q_i(p)}+\sum_{i=1}^{s_M} 2^{q_i(M)}
$$
is the binary decomposition of $M+p$, with
 $0\leq q_1(p)<\cdots<q_{s_p}(p)<q_1(M)<\cdots<q_{s_M}(M)$, 
and hence $s_{M+p}=s_p+s_M$ and
\begin{equation}
q_j(M+p)=\left\{\begin{array}{ll} q_j(p) & \text{for } j=1,\ldots,s_p\\
q_{j-s_p}(M)& \text{for } j=s_p+1,\ldots,s_{M+p}
\end{array}\right.
\label{eqn:q_j(M+p)}
\end{equation}

For every $n\in \NN$ and $i=1,\ldots,s_n$, let
$$
M_i(n)=\sum_{j=i}^{s_n} 2^{q_j(n)}=2^{q_i(n)}\left\lfloor \frac{n}{2^{q_i(n)}}\right\rfloor.
$$
\smallskip

\noindent\textbf{2.} For every $n\in \NN$, let
$\varphi_0(n)=\lfloor{n}/{2}\rfloor$ and 
$\varphi_1(n)=\lceil{n}/{2}\rceil$ and,
for every $m\geq 1$ and for every sequence $b_m\ldots b_0\in\{0,1\}^{m+1}$, let
$$
\varphi_{b_m\ldots b_0}(n)=\varphi_{b_m}\big(\varphi_{b_{m-1}\ldots b_0}(n)).
$$
By Thompson's Rounding Lemma \cite{thompson}, for every sequence $b_m\ldots b_0\in\{0,1\}^{m+1}$, 
\begin{equation}
\varphi_{b_m\ldots b_0}(n)=\left\lfloor \frac{n+\sum_{i=0}^m b_i2^i}{2^{m+1}}\right\rfloor.
\label{eqn:thompson}
\end{equation}
\smallskip

\noindent\textbf{3.} For every $m\in \NN$, let $B_m$ denote the $m$-th Bernoulli number of the first kind. Recall that these Bernoulli numbers can be defined, starting with $B_0=1$, by means of the recurrence
\begin{equation}
\sum_{k=0}^m \binom{m+1}{k} B_k=0,\quad  m\geq 1. 
\label{eqn:recurrenceB}
\end{equation}
Let moreover $B_m(x)\in \QQ[x]$ be the Bernoulli polynomial of degree $m$, defined by
$$
B_m(x)=\sum_{k=0}^m \binom{m}{k} B_k x^{m-k}.
$$
We list below several well-known properties of the Bernoulli numbers  and polynomials that we shall use in our proofs, frequently without any further notice. For these and other properties of the Bernoulli numbers  and polynomials, see for instance \cite[Ch. 23]{abram}.
\begin{align}
&B_1=-1/2\text{ and } B_{2i+1}=0\text{ for every $i\geq 1$} \label{eqn:Bimp}\\
& B_m=B_m(0)=(-1)^m B_m(1) \label{eqn:Bn(0)}\\
& B_m(x+y)=\sum_{k=0}^m \binom{m}{k} B_k(x)y^{m-k} \label{eqn:Bn(x+y)}
\end{align}
In particular, 
\begin{align}
B_m(2)& =B_m(1+1)=\sum_{k=0}^m \binom{m}{k} B_k(1)& & \text{(by (\ref{eqn:Bn(x+y)}))} \nonumber\\
&=\sum_{k=0}^m \binom{m}{k} (-1)^k B_k& &\text{(by (\ref{eqn:Bn(0)}))} \nonumber\\
& =\sum_{k=0}^{m} \binom{m}{k} B_k-2mB_1=\sum_{k=0}^m \binom{m}{k} B_k+m& &\text{(by (\ref{eqn:Bimp}))}\nonumber\\
&=B_m(1)+m=(-1)^mB_m+m& & \text{(again by (\ref{eqn:Bn(0)}))}
\label{eqn:B_n(2)}
\end{align}

\noindent\textbf{4.} For every $d\in \NN$, $n\in \NN_{\geq 1}$, and $x\in \RR\setminus\{0\}$, let
$$
T(d,n,x)=\sum_{k=0}^{n-1} k^dx^k
$$
The value of $T(d,n,1)$ is given by \emph{Faulhaber's formula} \cite[(6.78)]{Knuth2}: for every $d\in \NN$ and $n\in\NN_{\geq 1}$,
$$
\sum_{k=0}^{n-1} k^d=\frac{1}{d+1}\sum_{j=0}^d \binom{d+1}{j}B_jn^{d+1-j}
$$
We shall also use this formula in the following, equivalent way:
\begin{equation}
\sum_{k=1}^{n-1} k^d=\frac{1}{d+1}\Big(\sum_{j=0}^{d+1} \binom{d+1}{j}B_jn^{d+1-j}+(-1)^{d}B_{d+1}\Big)
\label{eqn:sumpots}
\end{equation}
The equivalence stems from the fact that if $d\geq 1$, $(-1)^{d}B_{d+1}=-B_{d+1}$, while when $d=0$,
$(-1)^{d}B_{d+1}=B_1=-1/2$.

When $x\neq 1$, the double sequence $T$ satisfies the recurrence
$$
T(d,n,x)=\frac{x}{1-x}\sum_{p=0}^{d-1}\binom{d}{p} T(p,n,x)-\frac{x^n}{1-x}n^d,\quad d\geq 1.
$$
In particular, when $x\neq 1$,
\begin{equation}
T(1,n,x)=\frac{x}{1-x}\big(T(0,n,x)-nx^{n-1}\big)
\label{eqn:recT}
\end{equation}
We shall use these double sequences $T(d,n,x)$ in order to unify some notations and proofs, but, as we have already encountered in the statement of our main result, we actually only need to know closed formulas for them when $d=0,1$:
\begin{equation}
\begin{array}{ll}
\displaystyle T(0,n,1)=n,\qquad &\displaystyle  T(0,n,x)=\frac{x^n-1}{x-1}\quad \text{for $x\neq 1$} \\[2ex]
\displaystyle T(1,n,1)=\binom{n}{2}, &\displaystyle  T(1,n,x)=\frac{nx^n(x-1)-x(x^n-1)}{(x-1)^2}\quad \text{for $x\neq 1$} 
\end{array}
\label{eqn:T}
\end{equation}

\subsection{Some technical lemmas}

\noindent We begin by providing closed formulas for several sums that we shall often encounter in our computations.

\begin{lemma}\label{lem:sum_a^q}
For every $d\in \NN$, $n\in \NN_{\geq 1}$, and $x\in \RR\setminus\{0\}$:
$$
 \sum_{k=1}^{n-1} q_{s_{k}}(k)^dx^{q_{s_{k}}(k)}= T(d,q_{s_n}(n),2x)+nq_{s_{n}}(n)^dx^{q_{s_n}(n)}-q_{s_{n}}(n)^d(2x)^{q_{s_n}(n)}
$$
\end{lemma}

\begin{proof}
\begin{align*}
 \sum_{k=1}^{n-1} q_{s_{k}}(k)^dx^{q_{s_{k}}(k)}&= \sum_{k=1}^{2^{q_{s_n}(n)}-1}
q_{s_{k}}(k)^dx^{q_{s_{k}}(k)}+\sum_{k=2^{q_{s_n}(n)}}^{n-1} q_{s_{k}}(k)^dx^{q_{s_{k}}(k)}\nonumber\\
& = \sum_{j=0}^{q_{s_n}(n)-1}j^d2^j x^j+(n-2^{q_{s_n}(n)})q_{s_{n}}(n)^dx^{q_{s_n}(n)}
\end{align*}
\end{proof}

In particular:
\begin{itemize}
\item If $x\neq 1/2$,
\begin{align*}
& \sum_{k=1}^{n-1} x^{q_{s_{k}}(k)} = \dfrac{(2x)^{q_{s_n}(n)}-1}{2x-1}+n\cdot x^{q_{s_n}(n)}-(2x)^{q_{s_n}(n)}\\
& \sum_{k=1}^{n-1} q_{s_{k}}(k)x^{q_{s_{k}}(k)}=
\dfrac{q_{s_n}(n)(2x)^{q_{s_n}(n)}}{2x-1}-\frac{2x((2x)^{q_{s_n}(n)}-1)}{(2x-1)^2} +nq_{s_n}(n)x^{q_{s_n}(n)}-q_{s_n}(n)(2x)^{q_{s_n}(n)}
\end{align*}
\item If $x= 1/2$,
\begin{align*}
& \sum_{k=1}^{n-1} x^{q_{s_{k}}(k)} = q_{s_n}(n)+n\cdot 2^{-q_{s_n}(n)}-1\\
& \sum_{k=1}^{n-1} q_{s_{k}}(k)x^{q_{s_{k}}(k)}=\dfrac{q_{s_n}(n)(q_{s_n}(n)-3)}{2}+nq_{s_n}(n)2^{-q_{s_n(n)}}
\end{align*}
\end{itemize}

\begin{lemma}\label{lem:sumx^t}
For every $d\in \NN$, $n\in \NN_{\geq 1}$, and $x\in \RR\setminus\{0\}$,
$$
 \sum_{k=0}^{n-1}\sum_{i=1}^{s_{k}}q_{i}(k)^dx^{q_i(k)}   =\sum_{i=1}^{s_n} 2^{q_i(n)-1}T(d,q_{i}(n),x) +\sum_{i=1}^{s_n}q_i(n)^dx^{q_i(n)}(n-M_i(n))
$$
\end{lemma}

\begin{proof}
Let  $x\in \RR\setminus\{0\}$. For every $d\in \NN$ and $n\in \NN_{\geq 1}$, set 
$$
a^{(d)}_n\coloneqq\sum_{k=0}^{n-1}\sum_{i=1}^{s_{k}}q_{i}(k)^dx^{q_i(k)}.
$$
Then,
\begin{align*}
a^{(d)}_{n}& =\sum_{k=0}^{M_2(n)-1} \sum_{i=1}^{s_{k}} q_i(k)^dx^{q_i(k)}+\sum_{k=M_2(n)}^{M_2(n)+2^{q_1(n)}-1} \sum_{i=1}^{s_{k}}q_i(k)^dx^{q_i(k)}\\
& = a^{(d)}_{M_2(n)}+\sum_{p=0}^{2^{q_1(n)}-1}\sum_{i=1}^{s_{M_2(n)+p}} q_i(M_2(n)+p)^dx^{q_i(M_2(n)+p)}\\
& = a^{(d)}_{M_2(n)}+\sum_{p=0}^{2^{q_1(n)}-1}\Big(\sum_{i=1}^{s_p} q_i(p)^dx^{q_i(p)}+\sum_{i=1}^{s_{M_2(n)}} q_i(M_2(n))^dx^{q_i(M_2(n))}\Big) &&\text{(by (\ref{eqn:q_j(M+p)}))} \\
& = a^{(d)}_{M_2(n)}+\sum_{p=0}^{2^{q_1(n)}-1}\sum_{i=1}^{s_p} q_i(p)^dx^{q_i(p)}+2^{q_1(n)}\sum_{i=2}^{s_n} q_i(n)^dx^{q_i(n)}\\
& = a^{(d)}_{M_2(n)}+a^{(d)}_{2^{q_1(n)}}+2^{q_1(n)}\sum_{i=2}^{s_n} q_i(n)^dx^{q_i(n)}.
\end{align*}
So, a simple argument by induction shows that, for every $2\leq r\leq s_n$
$$
a^{(d)}_n=a^{(d)}_{M_r(n)}+\sum_{j=1}^{r-1} a^{(d)}_{2^{q_j(n)}}+\sum_{j=1}^{r-1} \Big(2^{q_j(n)}\sum_{i=j+1}^{s_n} q_i(n)^dx^{q_i(n)}\Big).
$$
Taking $r=s_n$, we obtain
\begin{align}
a^{(d)}_n & =\sum_{j=1}^{s_n} a^{(d)}_{2^{q_j(n)}}+\sum_{j=1}^{s_n-1} \Big(2^{q_j(n)}\sum_{i=j+1}^{s_n} q_i(n)^dx^{q_i(n)}\Big)\nonumber\\
& =\sum_{j=1}^{s_n} a^{(d)}_{2^{q_j(n)}}+\sum_{i=2}^{s_n}\Big(q_i(n)^dx^{q_i(n)}\sum_{j=1}^{i-1} 2^{q_j(n)}\Big)\nonumber
\\ &
=\sum_{j=1}^{s_n} a^{(d)}_{2^{q_j(n)}}+\sum_{i=2}^{s_n}q_i(n)^dx^{q_i(n)}(n-M_i(n))
\label{eqn:lemprel1}
\end{align}
Now, for every $l\in\NN$, we have that
\begin{align}
a^{(d)}_{2^l} & = \sum_{k=0}^{2^l-1}\sum_{i=1}^{s_{k}} q_i(k)^dx^{q_i(k)}  =
\sum_{t=0}^{l-1} t^d\cdot x^{t}\cdot \#\{(k,i)\mid 0\leq k\leq 2^l-1,\ q_i(k)=t\} \nonumber\\
& = 2^{l-1}\sum_{t=0}^{l-1} t^d\cdot x^{t}=2^{l-1}T(d,l,x)
\label{eqn:bdosl}
\end{align}
because each 
$\#\{(k,i)\mid 0\leq k\leq 2^l-1,\ q_i(k)=t\}$ is simply the number of binary words $w\in \{0,1\}^l$ with an 1 in the position $t+1$ starting from the right, which is $2^{l-1}$.

Using this expression in Eqn.~(\ref{eqn:lemprel1}), and taking into account that $M_1(n)=n$, we obtain the formula in the statement.
\end{proof}

In particular:
\begin{itemize}
\item If $x\neq 1$,
\begin{align}
& \hspace*{-2ex}\sum_{k=0}^{n-1}\sum_{i=1}^{s_{k}} x^{q_i(k)}=\sum_{i=1}^{s_n} x^{q_i(n)}(n-M_i(n))+\sum_{i=1}^{s_n} 2^{q_i(n)}\cdot \frac{x^{q_i(n)}-1}{2(x-1)} \label{eqn:sumx^t}\\
& \hspace*{-2ex}\sum_{k=0}^{n-1}\sum_{j=1}^{s_{k}}q_{j}(k)x^{q_j(k)}=\sum_{i=1}^{s_n}q_i(n)x^{q_i(n)}(n-M_i(n))+
\sum_{i=1}^{s_n}2^{q_i(n)}\cdot \dfrac{(x-1)q_i(n)x^{q_i(n)} - x^{q_i(n)+1}+x}{2(x - 1)^2} \label{lem:beta0}
\end{align}

\item When $x=1$,
\begin{align}
& \hspace*{-2ex} \sum_{k=0}^{n-1} s_k=\sum_{k=1}^{n-1}\sum_{i=1}^{s_{k}} 1^{q_i(k)}=\sum_{i=1}^{s_n}2^{q_i(n)-1}(q_i(n)+2(s_n-i))
 \label{eqn:sumx^t1}
\\
& \hspace*{-2ex} \sum_{k=1}^{n-1}\sum_{j=1}^{s_{k}}q_{j}(k)=\sum_{i=1}^{s_n} 2^{q_i(n)-1}\Big(\binom{q_{i}(n)}{2}+{2}\sum_{j=i+1}^{s_n} q_j(n)\Big)\label{eqn:sumqx^t1}
\end{align}
\end{itemize}

\begin{lemma}\label{lem:z}
Let $(z_{l,p})_{(l,p)\in \NN^2}$ be a double sequence satisfying:
\begin{enumerate}[(a)]
\item For every $p\in \NN$, $z_{0,p}=0$ 
\item For every $l,p>0$,
$$
z_{l,p}=2z_{l-1,p}+\sum_{q=0}^{p-1} \binom{p}{q}2^{(p-q)(l-1)}z_{l-1,q}
$$
\end{enumerate}
Then, for every $l\geq 0$ and for every $p>0$,
$$
z_{l,p}=-\frac{2^{l}}{p+1}\sum_{j=1}^{l-1}\Big(\sum_{i=1}^{p} \binom{p+1}{i}(2^{i}-1) 2^{(p-i)l+(i-1)j}B_{i}\Big)z_{j,0}.
$$
\end{lemma}

\begin{proof}
Let us iterate several times the recurrence defining our double sequence:
\begin{align*}
& z_{l,p}  = 2z_{l-1,p}+ \sum_{q=0}^{p-1}\binom{p}{q}2^{(p-q)(l-1)}z_{l-1,q}\\
 &\quad = 2\Big(2z_{l-2,p}+ \sum_{q=0}^{p-1}\binom{p}{q}2^{(p-q)(l-2)}z_{l-2,q}\Big)+ \sum_{q=0}^{p-1}\binom{p}{q}2^{(p-q)(l-1)}z_{l-1,q}\\
 &\quad = 2^2z_{l-2,p}+ \sum_{q=0}^{p-1}\binom{p}{q}\Big(2^{(p-q)(l-2)+1}z_{l-2,q}+2^{(p-q)(l-1)}z_{l-1,q}\Big)\\
 &\quad = 2^2\Big(2z_{l-3,p}+ \sum_{q=0}^{p-1}\binom{p}{q}2^{(p-q)(l-3)}z_{l-3,q}\Big) + \sum_{q=0}^{p-1}\binom{p}{q}\Big(2^{(p-q)(l-2)+1}z_{l-2,q}+2^{(p-q)(l-1)}z_{l-1,q}\Big)\\
 &\quad = 2^3z_{l-3,p} + \sum_{q=0}^{p-1}\binom{p}{q}\Big(2^{(p-q)(l-3)+2}z_{l-3,q}+2^{(p-q)(l-2)+1}z_{l-2,q}+2^{(p-q)(l-1)}z_{l-1,q}\Big)
\end{align*}
Reasoning in this way, and recalling that $z_{0,q}=0$ for every $q$, it is easy to prove by induction on $l$ that
\begin{equation}
z_{l,p}  =\sum_{q=0}^{p-1}\binom{p}{q}\Big(\sum_{j=1}^{l-1} 2^{(p-q)j+(l-1-j)}z_{j,q}\Big)= 2^{l-1}\sum_{q=0}^{p-1}\binom{p}{q}\Big(\sum_{j=1}^{l-1} 2^{(p-q-1)j}z_{j,q}\Big)
\label{eqn:recz^p}
\end{equation}
We have computed explicitly $z_{l,p}$ for several $p\geq 1$ using this recurrence, in order to look for a pattern. The results have been:
\begin{align*}
z_{l,2} = & 2^{l-1}\sum_{j=1}^{l-1}(2^l-2\cdot 2^{j-1})z_{j,0}\\
z_{l,3} = &2^{l-1}\sum_{j=1}^{l-1}\Big(
2^{2l} -3\cdot 2^{l+j - 1}\Big)z_{j,0}\\
z_{l,4} = &2^{l-1}\sum_{j=1}^{l-1}\Big(2^{3l}-4\cdot 2^{2 l+j- 1} +8\cdot 2^{3(j-1)}\Big)z_{j,0}\\
z_{l,5} = & 2^{l-1} \sum_{j=1}^{l-1}(2^{4l}-5\cdot 2^{3 l+j-1}+20\cdot 2^{l+3(j-1)})z_{j,0} \\
z_{l,6} = & 2^{l-1} \sum_{j=1}^{l-1} (2^{5l}-6\cdot 2^{4 l+j-1}+40\cdot 2^{2 l+3(j-1)}-96\cdot 2^{5(j-1)})z_{j,0}\\
z_{l,7} = & 2^{l-1} \sum_{j=1}^{l-1} (2^{6l}-7\cdot 2^{5 l+j-1}+70\cdot 2^{3 l+3(j-1)}-336\cdot 2^{l+5(j-1)})z_{j,0}\\
z_{l,8} = & 2^{l-1} \sum_{j=1}^{l-1} (2^{7l}-8\cdot 2^{6l+j-1}+112\cdot 2^{4l+3(j-1)}-896\cdot 2^{2l+5(j-1)}+2176\cdot 2^{7(j-1)})z_{j,0}\\
z_{l,9} = & 2^{l-1} \sum_{j=1}^{l-1}(2^{8l}-9\cdot 2^{7 l+j-1}+168\cdot 2^{5l+3(j-1)}-2016\cdot 2^{3 l+5(j-1)}+9792\cdot 2^{l+7(j-1)})z_{j,0} \\
z_{l,10} = & 2^{l-1} \sum_{j=1}^{l-1} (2^{9l}-10\cdot 2^{8 l+j-1}+240\cdot 2^{6l+3(j-1)}-4032\cdot 2^{4l+5(j-1)} +32640\cdot 2^{2l+7(j-1)}\\[-2ex] & \qquad\qquad\qquad\qquad-79360\cdot 2^{9(j-1)})z_{j,0}\\
z_{l,11} = & 2^{l-1} \sum_{j=1}^{l-1} (2^{10 l}-11\cdot 2^{9 l+j-1}+330\cdot 2^{7l+3(j-1)}-7392\cdot 2^{5l+5(j-1)} +89760\cdot 2^{3l+7(j-1)}\\[-2ex] & \qquad\qquad\qquad\qquad-436480\cdot 2^{l+9(j-1)})z_{j,0}
,\\
z_{l,12} = & 2^{l-1} \sum_{j=1}^{l-1} (2^{11 l}-12\cdot 2^{10 l+j-1}+440\cdot 2^{8l+3(j-1)}-12672\cdot 2^{6l+5(j-1)}+215424\cdot 2^{4l+7(j-1)}\\[-2ex] & \qquad\qquad\qquad\qquad -1745920\cdot 2^{2l+9(j-1)} +4245504\cdot 2^{11(j-1)})z_{j,0}.
\end{align*}

These results hint that
$$
z_{l,p}=2^{l-1}\sum_{j=1}^{l-1}\Big(2^{(p-1)l}+\sum_{i=1}^{p-1} a^{(p)}_{i}\binom{p}{i} 2^{(p-1-i)l+i(j-1)}\Big)z_{j,0}.
$$
with $a_{2i}^{(p)}=0$ if $i>0$, 
$$
a_1^{(p)}=-1,\ 
a_3^{(p)}=2,\
a_5^{(p)}=-16,\ a_7^{(p)}=272,\
a_9^{(p)}=-7936,\ldots
$$
These values are consistent with
$$
a_i^{(p)}=-\frac{1}{i+1}\cdot 2^{i+1}(2^{i+1}-1)B_{i+1}
,\quad \text{for $1\leq i\leq p-1$}
$$
So, this leads us to conjecture that, if $p>0$,
\begin{align}
z_{l,p} & =2^{l-1}\sum_{j=1}^{l-1}\Big(2^{(p-1)l}-\sum_{i=1}^{p-1} \frac{1}{i+1}\cdot 2^{i+1}(2^{i+1}-1)B_{i+1}\binom{p}{i} 2^{(p-1-i)l+i(j-1)}\Big)z_{j,0}\nonumber\\
& =2^{l-1}\sum_{j=1}^{l-1}\Big(2^{(p-1)l}-\sum_{i=1}^{p-1} \frac{1}{p+1}\cdot 2^{i+1}(2^{i+1}-1)B_{i+1}\binom{p+1}{i+1} 2^{(p-1-i)l+i(j-1)}\Big)z_{j,0}\nonumber\\
& =2^{l-1}\sum_{j=1}^{l-1}\Big(2^{(p-1)l}-\sum_{i=2}^{p} \frac{1}{p+1}\cdot 2^{i}(2^{i}-1)B_{i}\binom{p+1}{i} 2^{(p-i)l+(i-1)(j-1)}\Big)z_{j,0}\nonumber\\
& =2^{l-1}\sum_{j=1}^{l-1}\Big(2^{(p-1)l}-\sum_{i=1}^{p} \frac{1}{p+1}\cdot 2^{i}(2^{i}-1)B_{i}\binom{p+1}{i} 2^{(p-i)l+(i-1)(j-1)}\nonumber\\
&\qquad\qquad+ \frac{1}{p+1}\cdot 2B_{1}(p+1) 2^{(p-1)l}\Big)z_{j,0}\nonumber\\
& =2^{l-1}\sum_{j=1}^{l-1}\Big(2^{(p-1)l}-\sum_{i=1}^{p} \frac{1}{p+1}\cdot 2^{i}(2^{i}-1)B_{i}\binom{p+1}{i} 2^{(p-i)l+(i-1)(j-1)}- 2^{(p-1)l}\Big)z_{j,0}\nonumber\\
& =-\frac{2^{l}}{p+1}\sum_{j=1}^{l-1}\Big(\sum_{i=1}^{p}(2^{i}-1)\binom{p+1}{i} 2^{(p-i)l+(i-1)j}B_{i}\Big)z_{j,0}
\label{eqn:conjecturaz}
\end{align}
and we proceed to prove this equality by induction on $l$.

The case when $l=1$ is true because  the right hand sum of (\ref{eqn:conjecturaz}) is empty and, by  (\ref{eqn:recz^p}),
$z_{1,p}=0$.
Assume now that the equality (\ref{eqn:conjecturaz}) is true 
for $z_{1,p},\ldots,z_{l-1,p}$ and every $p> 0$, and let us prove it for $z_{l,p}$. By (\ref{eqn:recz^p}) and the induction hypothesis,
\begin{align*}
& 2^{1-l}z_{l,p}  = \sum_{q=0}^{p-1}\binom{p}{q}\sum_{k=1}^{l-1} 2^{(p-q-1)k}z_{k,q}\\
&\quad =\sum_{k=1}^{l-1} 2^{(p-1)k}z_{k,0}\\
&\qquad -
\sum_{q=1}^{p-1}\binom{p}{q}\sum_{k=1}^{l-1}2^{(p-q-1)k}\cdot 2^{k-1}\sum_{j=1}^{k-1}\Big(\sum_{i=1}^{q} \frac{2^{i}(2^{i}-1)B_{i}}{q+1}\binom{q+1}{i} 2^{(q-i)k+(i-1)(j-1)}\Big)z_{j,0}\\
& \quad\text{(by the induction hypothesis)}\\
&\quad =\sum_{j=1}^{l-1} 2^{(p-1)j}z_{j,0}\\
&\qquad -\sum_{j=1}^{l-2}\Bigg[
\sum_{q=1}^{p-1}\binom{p}{q}\sum_{k=j+1}^{l-1}2^{(p-q)k-1}\sum_{i=1}^{q} \frac{2^{i}(2^{i}-1)B_{i}}{q+1}\binom{q+1}{i} 2^{(q-i)k+(i-1)(j-1)}\Bigg]z_{j,0}\\
&\quad =\sum_{j=1}^{l-1}\Bigg[2^{(p-1)j}-
\sum_{q=1}^{p-1}\binom{p}{q}\sum_{k=j+1}^{l-1}2^{(p-q)k-1}\sum_{i=1}^{q} \frac{2^{i}(2^{i}-1)B_{i}}{q+1}\binom{q+1}{i} 2^{(q-i)k+(i-1)(j-1)}\Bigg]z_{j,0}
\end{align*}
and then the coefficient of $z_{j,0}$, for $j\leq l-1$, in $2^{-l+1}z_{l,p}$,
which we want to prove to be
$$
-\frac{2}{p+1}\sum_{i=1}^{p}(2^{i}-1)B_{i}\binom{p+1}{i} 2^{(p-i)l+(i-1)j}
$$
is
\begin{align*}
& 2^{(p-1)j}-
\sum_{q=1}^{p-1}\binom{p}{q}\sum_{k=j+1}^{l-1}2^{(p-q)k-1}\sum_{i=1}^{q} \frac{2^{i}(2^{i}-1)B_{i}}{q+1}\binom{q+1}{i} 2^{(q-i)k+(i-1)(j-1)}\nonumber\\
&\quad =2^{(p-1)j}-
\sum_{i=1}^{p-1}\sum_{q=i}^{p-1}\sum_{k=j+1}^{l-1} \frac{(2^{i}-1)B_{i}}{q+1}\binom{p}{q}\binom{q+1}{i} 2^{(p-i)k+(i-1)j}\nonumber\\
&\quad =2^{(p-1)j}-
\sum_{i=1}^{p-1}\sum_{q=i}^{p-1}\sum_{k=j+1}^{l-1}\frac{(2^{i}-1)B_{i}}{p+1}\binom{p+1}{q+1}\binom{q+1}{i} 2^{(p-i)k+(i-1)j}\nonumber\\
&\quad =2^{(p-1)j}-
\sum_{i=1}^{p-1}\sum_{q=i}^{p-1}\sum_{k=j+1}^{l-1}\frac{(2^{i}-1)B_{i}}{p+1}\binom{p+1}{i}\binom{p+1-i}{q+1-i} 2^{(p-i)k+(i-1)j}\nonumber\\
&\quad =2^{(p-1)j}-
\sum_{i=1}^{p-1}\frac{2^{(i-1)j}(2^{i}-1)B_{i}}{p+1}\binom{p+1}{i}\Big(\sum_{q=i}^{p-1}\binom{p+1-i}{q+1-i}\Big)\Big(\sum_{k=j+1}^{l-1}  2^{(p-i)k}\Big)\nonumber\\
&\quad =2^{(p-1)j}-
\sum_{i=1}^{p-1}\frac{2^{(i-1)j}(2^{i}-1)B_{i}}{p+1}\binom{p+1}{i}2(2^{p-i}-1\Big)\frac{2^{(p-i)l}-2^{(p-i)(j+1)}}{2^{p-i}-1}\nonumber\\
&\quad =2^{(p-1)j}-
\sum_{i=1}^{p-1}\frac{(2^{i}-1)B_{i}}{p+1}\binom{p+1}{i}2^{(i-1)j+1}(2^{(p-i)l}-2^{(p-i)(j+1)})\nonumber\\
&\quad =2^{(p-1)j}-
\sum_{i=1}^{p-1}\frac{2(2^{i}-1)B_{i}}{p+1}\binom{p+1}{i}2^{(p-i)l+(i-1)j}+
\sum_{i=1}^{p-1}\frac{2(2^{i}-1)B_{i}}{p+1}\binom{p+1}{i}2^{p-i+(p-1)j}\\
&\quad =2^{(p-1)j}-
\sum_{i=0}^{p}\frac{2(2^{i}-1)B_{i}}{p+1}\binom{p+1}{i}2^{(p-i)l+(i-1)j}+
\sum_{i=0}^{p}\frac{2(2^{i}-1)B_{i}}{p+1}\binom{p+1}{i}2^{p-i+(p-1)j}\nonumber\\
&\quad =2^{(p-1)j}-
\sum_{i=0}^{p}\frac{2(2^{i}-1)B_{i}}{p+1}\binom{p+1}{i}2^{(p-i)l+(i-1)j}+
\sum_{i=0}^{p}\frac{2(2^{i}-1)B_{i}}{p+1}\binom{p+1}{i}2^{p-i+(p-1)j}\nonumber\\
&\quad =2^{(p-1)j}-
\sum_{i=0}^{p}\frac{2(2^{i}-1)B_{i}}{p+1}\binom{p+1}{i}2^{(p-i)l+(i-1)j}
+
\frac{2^{p+1+(p-1)j}}{p+1}\sum_{i=0}^{p}\binom{p+1}{i}B_{i}\\
&\qquad\qquad
- \frac{2^{(p-1)j}}{p+1}\sum_{i=0}^{p}\binom{p+1}{i}2^{p+1-i}B_i
\\
&\quad =2^{(p-1)j}-
\sum_{i=0}^{p}\frac{2(2^{i}-1)B_{i}}{p+1}\binom{p+1}{i}2^{(p-i)l+(i-1)j} - \frac{2^{(p-1)j}}{p+1}\big(B_{p+1}(2)-B_{p+1})\\
&\quad =-
\sum_{i=0}^{p}\frac{2(2^{i}-1)B_{i}}{p+1}\binom{p+1}{i}2^{(p-i)l+(i-1)j}\quad \text{(by Eqn.~(\ref{eqn:B_n(2)}))}
\end{align*}
 as we wanted to prove. 
\end{proof}

For every $l,m,p,d\in\NN$, with $p<m$, and for every given $a\in \RR\setminus\{0\}$, let
$$
\gamma^{(d,p,m)}_l=\sum_{k=1}^{2^l-1}  \sum_{i=1}^{s_{k}} q_i(k)^d(2^{-m}a)^{q_i(k)}M_{i+1}(k)^p
$$
In the last part of our computations we shall use an alternative expression for  $\gamma^{(d,p,m)}_l$ which we derive now as an application of the last lemma.

We have that $\gamma^{(d,p,m)}_0=0$ by definition, and, by Eqn.~(\ref{eqn:bdosl}),
\begin{align}
\gamma^{(d,0,m)}_l &=\sum_{k=1}^{2^l-1}  \sum_{i=1}^{s_{k}}q_i(k)^d(a2^{-m})^{q_i(k)}= 2^{l-1}T(d,l,a2^{-m})
\label{eqn:varphid^0}
\end{align}

The case when $l,p>0$ is covered by the next lemma.

\begin{lemma}\label{lem:varphid}
If $l,p>0$, then
 $$
\gamma^{(d,p,m)}_l =\frac{a^{l-1}2^{-(m-1)(l-1)+pl}}{p+1}\sum_{t=1}^{l-1} (l-t-1)^d(a^{-1}2^{m-p-1})^{t}(B_{p+1}(2^{t})-B_{p+1})
$$
\end{lemma}

\begin{proof}
If $l,p>0$,
\begin{align*}
\gamma^{(d,p,m)}_l & =\gamma^{(d,p,m)}_{l-1}+ \sum_{k=2^{l-1}}^{2^l-1}\sum_{i=1}^{s_{k}}q_i(k)^d(a2^{-m})^{q_i(k)}M_{i+1}(k)^p\\
& =\gamma^{(d,p,m)}_{l-1}+ \sum_{k=2^{l-1}}^{2^l-1}\sum_{i=1}^{s_{k}-1} q_i(k)^d(a2^{-m})^{q_i(k)}\Big(M_{i+1}(k-2^{l-1})+2^{l-1}\Big)^p\\
& =\gamma^{(d,p,m)}_{l-1}+ \sum_{k=0}^{2^{l-1}-1}\sum_{i=1}^{s_{k}}q_i(k)^d(a2^{-m})^{q_i(k)}\Big(M_{i+1}(k)+2^{l-1}\Big)^p\\
& =\gamma^{(d,p,m)}_{l-1}+ \sum_{k=1}^{2^{l-1}-1}\sum_{i=1}^{s_{k}}q_i(k)^d(a2^{-m})^{q_i(k)}\sum_{q=0}^p\binom{p}{q}M_{i+1}(k)^q 2^{(p-q)(l-1)}\\
& =\gamma^{(d,p,m)}_{l-1}+ \sum_{q=0}^p\binom{p}{q}2^{(p-q)(l-1)}\sum_{k=1}^{2^{l-1}-1}\sum_{i=1}^{s_{k}}q_i(k)^d(a2^{-m})^{q_i(k)}M_{i+1}(k)^q \\
& =\gamma^{(d,p,m)}_{l-1}+ \sum_{q=0}^p\binom{p}{q}2^{(p-q)(l-1)}\gamma^{(d,q,m)}_{l-1} =2\gamma^{(d,p,m)}_{l-1}+ \sum_{q=0}^{p-1}\binom{p}{q}2^{(p-q)(l-1)}\gamma^{(d,q,m)}_{l-1}
\end{align*}
Then, by Lemma \ref{lem:z} and Eqn.~(\ref{eqn:varphid^0}),
\begin{align*}
&\gamma_l^{(d,p,m)}  =-\frac{2^{l}}{p+1}\sum_{j=1}^{l-1}\Big(\sum_{i=1}^{p} \binom{p+1}{i}(2^{i}-1) 2^{(p-i)l+(i-1)j}B_{i}\Big)\gamma_j^{(d,0,m)}\\
&\quad=-\frac{2^{l}}{p+1}\sum_{j=1}^{l-1}\Bigg[\sum_{i=1}^{p} \binom{p+1}{i}(2^{i}-1) 2^{(p-i)l+(i-1)j}B_{i} 2^{j-1}\sum_{t=0}^{j-1}t^d(a2^{-m})^t\Bigg]\\
&\quad=-\frac{2^{pl+l-1}}{p+1}\sum_{i=1}^{p} \Bigg[\binom{p+1}{i} (2^{i}-1)2^{-il}\Big(\sum_{t=0}^{l-2}t^d(a2^{-m})^t\sum_{j=t+1}^{l-1}2^{ij}\Big)\Bigg]B_i\\
&\quad=-\frac{2^{pl+l-1}}{p+1}\sum_{i=1}^{p} \Bigg[\binom{p+1}{i} (2^{i}-1)2^{-il}\Big(\sum_{t=0}^{l-2}t^d(a2^{-m})^t\cdot\frac{2^{il}-2^{i(t+1)}}{2^i-1}\Big)\Bigg]B_i\\
&\quad=-\frac{2^{pl+l-1}}{p+1}\sum_{i=1}^{p} \binom{p+1}{i} B_i\cdot  \Big(\sum_{t=0}^{l-2}t^d(a2^{-m})^t\Big)\\
&\qquad\qquad +\frac{2^{pl+l-1}}{p+1}\sum_{i=1}^{p} \binom{p+1}{i}2^{-i(l-1)}\Big(\sum_{t=0}^{l-2}t^d(a2^{i-m})^t\Big)B_i\\
&\quad=\frac{2^{pl+l-1}}{p+1}B_0\cdot  \Big(\sum_{t=0}^{l-2}t^d(a2^{-m})^t\Big) +\frac{2^{pl+l-1}}{p+1}\sum_{i=1}^{p} \binom{p+1}{i}2^{-i(l-1)}\Big(\sum_{t=0}^{l-2}t^d(a2^{i-m})^t\Big)B_i\\
&\quad=\frac{2^{pl+l-1}}{p+1}\sum_{i=0}^{p} \binom{p+1}{i}2^{-i(l-1)}\Big(\sum_{t=0}^{l-2}t^d(a2^{i-m})^t\Big)B_i\\
&\quad=\frac{2^{pl+l-1}}{p+1}\sum_{i=0}^{p} \binom{p+1}{i}2^{-i(l-1)}\Big(\sum_{t=1}^{l-1}(l-1-t)^d(a2^{i-m})^{l-1-t}\Big)B_i\\
&\quad=\frac{a^{l-1}2^{-(m-1)(l-1)+pl}}{p+1}\sum_{t=1}^{l-1} (l-t-1)^d(a^{-1}2^{m-p-1})^{t}\Big(\sum_{i=0}^{p} \binom{p+1}{i}2^{(p+1-i)t}B_i\Big)\\
&\quad=\frac{a^{l-1}2^{-(m-1)(l-1)+pl}}{p+1}\sum_{t=1}^{l-1} (l-t-1)^d(a^{-1}2^{m-p-1})^{t}(B_{p+1}(2^{t})-B_{p+1})
\end{align*}
as we claimed.
\end{proof}

\begin{cor}\label{cor:varphi^0alt}
For every $l,m,p,d\in\NN$, with $p<m$,
\begin{align*}
\gamma^{(d,p,m)}_l &=\frac{a^{l-1}2^{-(m-1)(l-1)+pl}}{p+1}\sum_{t=1}^{l-1} (l-t-1)^d(a^{-1}2^{m-p-1})^{t}(B_{p+1}(2^{t})-B_{p+1})\nonumber\\
&\qquad\qquad+
(a2^{-(m-1)})^{l-1}(l-1)^d\cdot \delta_{p=0,l>0}
\end{align*}
where $\delta_{p=0,l>0}=1$ if $p=0$ and $l>0$, and $\delta_{p=0,l>0}=0$ otherwise.
\end{cor}

\begin{proof}
When $l=0$, both sides of the equality in the statement are equal to 0, 
and the last lemma covers the case when $p>0$ and $l>0$. Finally, when $p=0$ and $l>0$, by (\ref{eqn:varphid^0}),
\begin{align*}
&\gamma^{(d,0,m)}_l  =2^{l-1}\sum_{t=0}^{l-1} t^d(a2^{-m})^t=2^{l-1}\sum_{t=0}^{l-1} (l-1-t)^d(a2^{-m})^{l-1-t}\\
&\quad = a^{l-1}2^{-(m-1)(l-1)}\Big(\sum_{t=1}^{l-1} (l-1-t)^d(a^{-1}2^{m})^{t}+(l-1)^d\Big)\\
&\quad = a^{l-1}2^{-(m-1)(l-1)}\Big(\sum_{t=1}^{l-1} (l-1-t)^d(a^{-1}2^{m-1})^{t}(B_{1}(2^{t})-B_{1})+(l-1)^d\Big)\\
&\quad  =\frac{a^{l-1}2^{-(m-1)(l-1)+0\cdot l}}{0+1}\sum_{t=1}^{l-1} (l-t-1)^d(a^{-1}2^{m-0-1})^{t}(B_{0+1}(2^{t})-B_{0+1})+
(a2^{-(m-1)})^{l-1}(l-1)^d
\end{align*}
\end{proof}

\subsection{Statement of the problem}

\noindent Let $P(x,y)=\sum_{r,t\geq 0} b_{r,t}x^ry^t\in \RR[x,y]$ be a bivariate polynomial and $a\in \RR\setminus\{0\}$, and consider the recurrence equation
\begin{equation}
x_n =a\cdot x_{\left\lceil{n}/{2}\right\rceil}+a\cdot x_{\left\lfloor{n}/{2}\right\rfloor}+P\big(\lceil{n}/{2}\rceil,\lfloor{n}/{2}\rfloor\big),\quad n\geq 2.
\label{eqn:gral}
\end{equation}

\begin{lemma}\label{lem:x1neq0}
If $(x^{0}_n)_n$ is the solution of (\ref{eqn:gral}) with initial condition $x_1^{0}=0$, then
$$
\overline{x}_n=x^{0}_n+\big((2a)^{q_{s_n}(n)}+(2a-1)(na^{q_{s_n}(n)}-(2a)^{q_{s_n}(n)})\big)x_1
$$
is its solution with initial condition $x_1$
\end{lemma}

\begin{proof}
Since $\overline{x}_1=x_1$, we must prove that $(\overline{x}_n)_n$ satisfies (\ref{eqn:gral}). Since $(x^0_n)_n$ already satisfies it,
it is enough to check that, for every $n\geq 2$,
\begin{align*}
& (2a)^{q_{s_n}(n)}+(2a-1)(na^{q_{s_n}(n)}-(2a)^{q_{s_n}(n)})\\
&\qquad=
a\big((2a)^{q_{s_{\lceil n/2\rceil}}(\lceil n/2\rceil)}+(2a-1)(\lceil n/2\rceil a^{q_{s_{\lceil n/2\rceil }}(\lceil n/2\rceil)}-(2a)^{q_{s_{\lceil n/2\rceil }}(\lceil n/2\rceil)})\big)\\
&\qquad\qquad\qquad +
a\big((2a)^{q_{s_{\lfloor n/2\rfloor}}(\lfloor n/2\rfloor)}+(2a-1)(\lfloor n/2\rfloor a^{q_{s_{\lfloor n/2\rfloor }}(\lfloor n/2\rfloor)}-(2a)^{q_{s_{\lfloor n/2\rfloor }}(\lfloor n/2\rfloor)})\big)
\end{align*}
The simplest way to do it is by distinguishing two cases; to simplify the notations, we denote in the rest of this proof $q_{s_n}(n)$ by $q$.
\begin{enumerate}[(i)]
\item If $n\neq 2^{q+1}-1$, then $q_{s_{\lceil n/2\rceil }}(\lceil n/2\rceil)=q_{s_{\lfloor n/2\rfloor}}(\lfloor n/2\rfloor)=q-1$ and then 
\begin{align*}
& a\big((2a)^{q-1}+(2a-1)(\lceil n/2\rceil a^{q-1}-(2a)^{q-1})+(2a)^{q-1}+(2a-1)(\lfloor n/2\rfloor a^{q-1}-(2a)^{q-1})\big)\\
&\qquad= a\big(2(2a)^{q-1}+(2a-1)(na^{q-1}-2(2a)^{q-1})\big)=(2a)^{q}+(2a-1)(na^{q}-(2a)^{q})
\end{align*}
\item If $n= 2^{q+1}-1$,  then $\lceil n/2\rceil =2^q$ and $\lfloor n/2\rfloor=2^q-1$, and then
\begin{align*}
& a\big((2a)^{q}+(2a-1)(2^q a^{q}-(2a)^{q})+(2a)^{q-1}+(2a-1)((2^q-1)a^{q-1}-(2a)^{q-1})\big)\\
&\qquad= 2^{q+1}a^{q+1}-(2a-1)a^q= (2a)^{q}+(2a-1)((2^{q+1}-1)a^{q}-(2a)^{q})
\end{align*}
\end{enumerate}
\end{proof}

Therefore, in the rest of this paper we shall only be concerned with the solution of (\ref{eqn:gral}) with initial condition $x_1=0$.
Now, if, for every $(r,t)\in \NN^2$ such that $b_{r,t}\neq 0$, we know the solution $(x_n^{(r,t)})_n$ of the recurrence
\begin{equation}
x_n^{(r,t)} =a\cdot x_{\left\lceil{n}/{2}\right\rceil}^{(r,t)}+a\cdot x_{\left\lfloor{n}/{2}\right\rfloor}^{(r,t)}+\Big\lceil\frac{n}{2}\Big\rceil^r\cdot \Big\lfloor\frac{n}{2}\Big\rfloor^t,\quad n\geq 2,
\label{eqn:gralrt}
\end{equation}
with initial condition $x_1^{(r,t)}=0$, then the solution of (\ref{eqn:gral}) with  initial condition $x_1=0$ will be
$$
x_n=\sum_{r,t}b_{r,t} x_n^{(r,t)}.
$$
So, we are finally led to consider, for any $r,t\in \NN$,  the equation (\ref{eqn:gralrt}). Let $x_n^{(r,t)}$ be its solution with $x_1^{(r,t)}=0$.

\subsection{The sequence of differences of consecutive terms $y^{(r,t)}_n$}

\noindent Let  $y^{(r,t)}_n=x^{(r,t)}_n-x^{(r,t)}_{n-1}$, with $y^{(r,t)}_1=0$, and
\begin{equation}
x^{(r,t)}_n=\sum_{k=1}^n y^{(r,t)}_k.
\label{eqn:xsumadey}
\end{equation}
Notice that
$$
y^{(r,t)}_2=x^{(r,t)}_2-x^{(r,t)}_1=a\cdot x^{(r,t)}_1+a\cdot x^{(r,t)}_1+1-x^{(r,t)}_1=1.
$$
Let us derive  a recurrence for $y^{(r,t)}_n$, for $n\geq 3$. Taking into account the parity of $n$, we have:
\begin{itemize}
    \item For every $m\geq 2$,
\begin{align*}
& y^{(r,t)}_{2m}  = x^{(r,t)}_{2m}-x^{(r,t)}_{2m-1} \\   &\quad = 2a\cdot x^{(r,t)}_{m}+m^{r+t}-a\cdot x^{(r,t)}_m-a\cdot x^{(r,t)}_{m-1}-m^r(m-1)^t = a\cdot y^{(r,t)}_m+m^{r+t}-m^r(m-1)^t\\
&\quad =a\cdot y^{(r,t)}_{\lceil (2m)/2\rceil} +\lceil (2m)/2\rceil^r\lfloor (2m)/2\rfloor^t-\lfloor (2m)/2\rfloor^r(\lceil (2m)/2\rceil-1)^t
\end{align*}

\item For every $m\geq 1$,
\begin{align*}
& y^{(r,t)}_{2m+1}  =x^{(r,t)}_{2m+1} -x^{(r,t)}_{2m}\\ &\quad = a\cdot x^{(r,t)}_{m+1}+a\cdot x^{(r,t)}_{m}+(m+1)^rm^t-2a\cdot x^{(r,t)}_m-m^{r+t} = a\cdot y^{(r,t)}_{m+1}+(m+1)^rm^t-m^{r+t}\\ &\quad =a\cdot y^{(r,t)}_{\lceil (2m+1)/2\rceil} +\lceil (2m+1)/2\rceil^r\lfloor (2m+1)/2\rfloor^t-\lfloor (2m+1)/2\rfloor^r(\lceil (2m+1)/2\rceil-1)^t
\end{align*}
\end{itemize}
So, in summary, $y^{(r,t)}_n$ satisfies the recurrence
$$
y^{(r,t)}_n=a\cdot y^{(r,t)}_{\varphi_1(n)} +\varphi_1(n)^r\varphi_0(n)^t-\varphi_0(n)^r(\varphi_1(n)-1)^t,\quad n\geq 3.
$$
Therefore, for every $n\geq 3$ and  $m\geq 1$,
$$
y^{(r,t)}_n=a^m\cdot y^{(r,t)}_{\varphi_{\scriptsize\underbrace{1\ldots 1}_m}(n)}+\sum_{k=0}^{m-1} a^k\Big(\varphi_{\scriptsize1\underbrace{1\ldots 1}_k}(n)^r\varphi_{\scriptsize0\underbrace{1\ldots 1}_k}(n)^t-\varphi_{\scriptsize0\underbrace{1\ldots 1}_k}(n)^r(\varphi_{\scriptsize1\underbrace{1\ldots 1}_k}(n)-1)^t\Big)
$$
where, by Eqn.~(\ref{eqn:thompson}), for every $k\geq 0$
\begin{align*}
& \varphi_{\scriptsize 1\underbrace{1\ldots 1}_k}(n)=\left\lfloor\frac{n+2^{k+1}-1}{2^{k+1}}\right\rfloor =\left\lfloor\frac{n-1}{2^{k+1}}\right\rfloor+1\\
& \varphi_{\scriptsize 0\underbrace{1\ldots 1}_k}(n)=\left\lfloor\frac{n+2^{k}-1}{2^{k+1}}\right\rfloor=\left\lfloor\frac{n-1}{2^{k+1}}+\frac{1}{2}\right\rfloor
\end{align*}
This implies that, for every $n\geq 3$, 
$$
y^{(r,t)}_n=a^{L_n}y^{(r,t)}_2+\sum_{k=1}^{L_n} a^{k-1}\Bigg(\Big(1+\left\lfloor\frac{n-1}{2^{k}}\right\rfloor\Big)^r\left\lfloor\frac{n-1}{2^{k}}+\frac{1}{2}\right\rfloor^t-\left\lfloor\frac{n-1}{2^{k}}+\frac{1}{2}\right\rfloor^r\left\lfloor\frac{n-1}{2^{k}}\right\rfloor^t\Bigg)
$$
where $L_n$ is such that 
$2=\varphi_{\scriptsize\underbrace{1\ldots 1}_{L_n}}(n)=1+\lfloor(n-1)/2^{L_n}\rfloor$
that is,
$L_n=q_{s_{n-1}}(n-1)$. Thus, in summary:
\begin{equation}
y^{(r,t)}_n  =a^{q_{s_{n-1}}(n-1)}+\hspace*{-2ex}\sum_{k=1}^{q_{s_{n-1}}(n-1)} a^{k-1}\Bigg[\Big(1+\left\lfloor\frac{n-1}{2^{k}}\right\rfloor\Big)^r\left\lfloor\frac{n-1}{2^{k}}+\frac{1}{2}\right\rfloor^t-\left\lfloor\frac{n-1}{2^{k}}+\frac{1}{2}\right\rfloor^r\left\lfloor\frac{n-1}{2^{k}}\right\rfloor^t\Bigg]
\label{eqn:sumy2}
\end{equation}

Next result gives an expression for $y^{(r,t)}_n$ that will be more suitable to our purposes. To simplify the notations, in its statement and henceforth, for every $d,m,n\in \NN$, we set
$$
S^{(d,m)}_n=\sum_{j=1}^{s_{n}-1}q_j(n)^d(2^{-m}a)^{q_j(n)}M_{j+1}(n)^m.
$$
Moreover, we shall let $\ell=\log_2(a)-t$ when $a>0$,  and we shall use the following Kronecker delta:
$$
\delta_{\ell}=\left\{\begin{array}{ll}
1 & \text{ if $a>0$, $r>0$, and $\ell\in \{0,\ldots,r-1\}$}\\
0 & \text{ otherwise}
\end{array}\right.
$$

\begin{proposition}\label{lem:qib1}
For every $n\geq 2$, 
\begin{align*}
& y^{(r,t)}_n   =a^{q_{s_{n-1}}(n-1)}+\sum_{l=0\atop l\neq\ell}^{r-1}\frac{\binom{r}{l}}{2^{t+l}-a}\big((n-1)^{t+l}-a^{q_{s_{n-1}}(n-1)}\big)\\
&\qquad\qquad  
 +\sum_{i=0}^{r+t-1}\Big(2^{-i}\binom{r+t}{i}-2^{-i+1}\binom{r}{i-t}-\sum_{l=i-t+1\atop l\neq \ell}^{r-1}\frac{\binom{r}{l}\binom{t+l}{i}}{2^{t+l}-a}\Big)S_{n-1}^{(0,i)}\\
&\qquad\qquad  
+\delta_{\ell}\cdot\frac{1}{a}\binom{r}{\ell}\sum_{j=0}^{s_{n-1}-1}M_{j+1}(n-1)^{t+\ell}(q_{j+1}(n-1)-q_{j}(n-1))\\
 \end{align*}
\end{proposition}

\begin{proof}
Let $n\in \NN_{geq 2}$.  To simplify the notations, in this proof we denote $s_{n-1}$ by $s$ and each $q_i(n-1)$ and $M_i(n-1)$, $i=1,\ldots,s$, by $q_i$ and $M_i$, respectively.  Recall that $\ell=\log_2(a)-t$ if $a>0$. This number will play a role in sums of geometric sequences of the form $\sum_{k=1}^{m} (2^{-(t+l)}a)^k$ with $l\in \NN$, which yield two different expressions depending on whether $l=\ell$, that is, $2^{-(t+l)}a=1$, or not.

It is easy to check that when $n=2$ the right-hand side in the formula for $y^{(r,t)}_n$ given in the statement is $1$, using that $s_1=1$ and $q_{s_1}(1)=0$. Assume now that $n\geq 3$. In this case, 
by (\ref{eqn:sumy2}), 
$$
y^{(r,t)}_n=\widehat{y}_n+a^{q_s}
$$
where
$$
\widehat{y}_n =\sum_{k=1}^{q_s} a^{k-1}\Bigg(\Big(1+\left\lfloor\frac{n-1}{2^{k}}\right\rfloor\Big)^r\left\lfloor\frac{n-1}{2^{k}}+\frac{1}{2}\right\rfloor^t-\left\lfloor\frac{n-1}{2^{k}}+\frac{1}{2}\right\rfloor^r\left\lfloor\frac{n-1}{2^{k}}\right\rfloor^t\Bigg)
$$

Let us compute this sum. To begin with, notice that:
\begin{itemize}
\item If $k\leq q_1$,
\begin{align*}
& \Big(1+\left\lfloor\frac{n-1}{2^{k}}\right\rfloor\Big)^r\left\lfloor\frac{n-1}{2^{k}}+\frac{1}{2}\right\rfloor^t-\left\lfloor\frac{n-1}{2^{k}}+\frac{1}{2}\right\rfloor^r\left\lfloor\frac{n-1}{2^{k}}\right\rfloor^t\\
&\quad =
\Big(1+\frac{n-1}{2^{k}}\Big)^r\Big(\frac{n-1}{2^{k}}\Big)^t-\Big(\frac{n-1}{2^{k}}\Big)^r\Big(\frac{n-1}{2^{k}}\Big)^t\\
&\quad =(1+2^{-k}M_1)^r(2^{-k}M_1)^t-(2^{-k}M_1)^{r+t}=\sum_{l=0}^{r-1}\binom{r}{l}(2^{-k}M_1)^{t+l}
\end{align*}

\item If $q_j+1<k\leq q_{j+1}$, for some $j=1,\ldots,s-1$,
\begin{align*}
& \Big(1+\left\lfloor\frac{n-1}{2^{k}}\right\rfloor\Big)^r\left\lfloor\frac{n-1}{2^{k}}+\frac{1}{2}\right\rfloor^t-\left\lfloor\frac{n-1}{2^{k}}+\frac{1}{2}\right\rfloor^r\left\lfloor\frac{n-1}{2^{k}}\right\rfloor^t\\
&\quad =\Big(1+ \left\lfloor\sum_{i=1}^s 2^{q_i-k}\right\rfloor\Big)^r\left\lfloor\sum_{i=1}^s 2^{q_i-k}+\frac{1}{2}\right\rfloor^t -\left\lfloor\sum_{i=1}^s 2^{q_i-k}+\frac{1}{2}\right\rfloor^r\left\lfloor\sum_{i=1}^s 2^{q_i-k}\right\rfloor^t\\
&\quad =\Big(1+ \sum_{i=j+1}^s 2^{q_i-k}\Big)^r\Big(\sum_{i=j+1}^s 2^{q_i-k}\Big)^t -\Big(\sum_{i=j+1}^s 2^{q_i-k}\Big)^r\Big(\sum_{i=j+1}^s 2^{q_i-k}\Big)^t\\
&\quad =(1+ 2^{-k}M_{j+1})^r(2^{-k}M_{j+1})^t-(2^{-k}M_{j+1})^{r+t}=\sum_{l=0}^{r-1}\binom{r}{l}(2^{-k}M_{j+1})^{t+l}
\end{align*}

\item If $k=q_j+1$, for some $j=1,\ldots,s-1$,
\begin{align*}
& \Big(1+\left\lfloor\frac{n-1}{2^{k}}\right\rfloor\Big)^r\left\lfloor\frac{n-1}{2^{k}}+\frac{1}{2}\right\rfloor^t-\left\lfloor\frac{n-1}{2^{k}}+\frac{1}{2}\right\rfloor^r\left\lfloor\frac{n-1}{2^{k}}\right\rfloor^t\\
&\quad =\Big(1+ \left\lfloor\sum_{i=1}^s 2^{q_i-q_j-1}\right\rfloor\Big)^r\left\lfloor\sum_{i=1}^s 2^{q_i-q_j-1}+\frac{1}{2}\right\rfloor^t-\left\lfloor\sum_{i=1}^s 2^{q_i-q_j-1}+\frac{1}{2}\right\rfloor^r\left\lfloor\sum_{i=1}^s 2^{q_i-q_j-1}\right\rfloor^t\\
&\quad =\Big(1+ \sum_{i=j+1}^s 2^{q_i-q_j-1}\Big)^r\Big(\sum_{i=j+1}^s 2^{q_i-q_j-1}+1\Big)^t -\Big(\sum_{i=j+1}^s 2^{q_i-q_j-1}+1\Big)^r\Big(\sum_{i=j+1}^s 2^{q_i-q_j-1}\Big)^t\\
&\quad =(1+ 2^{-q_j-1}M_{j+1})^{r+t}-(1+2^{-q_j-1}M_{j+1})^r(2^{-q_j-1}M_{j+1})^t\\
    &\quad = \sum_{i=0}^{r+t-1}\binom{r+t}{i}(2^{-q_j-1}M_{j+1})^i - \sum_{l=0}^{r-1}\binom{r}{l}(2^{-q_j-1}M_{j+1})^{t+l}\\
\end{align*}
\end{itemize}

Then,
\begin{align*}
& \sum_{k=1}^{q_1} a^{k-1}\Bigg(\Big(1+\left\lfloor\frac{n-1}{2^{k}}\right\rfloor\Big)^r\left\lfloor\frac{n-1}{2^{k}}+\frac{1}{2}\right\rfloor^t-\left\lfloor\frac{n-1}{2^{k}}+\frac{1}{2}\right\rfloor^r\left\lfloor\frac{n-1}{2^{k}}\right\rfloor^t\Bigg)\\
&\quad =\sum_{k=1}^{q_1}a^{k-1}\sum_{l=0}^{r-1}\binom{r}{l}(2^{-k}M_1)^{t+l}=\frac{1}{a}\sum_{l=0}^{r-1}\binom{r}{l}M_1^{t+l}\sum_{k=1}^{q_1}(2^{-(t+l)}a)^k=(*)
\end{align*}

Now, if $\ell\notin \{0,\ldots,r-1\}$, then $2^{-(t+l)}a\neq 1$ for every $l=0,\ldots,r-1$ and hence
$$
(*)=\sum_{l=0}^{r-1}\binom{r}{l}M_1^{t+l}\cdot \frac{1-(2^{-(t+l)}a)^{q_1}}{2^{t+l}-a}
$$
while if $\ell\in \{0,\ldots,r-1\}$, then
\begin{align*}
(*) &=\frac{1}{a}\sum_{l=0\atop l\neq \ell}^{r-1}\binom{r}{l}M_1^{t+l}\sum_{k=1}^{q_1}(2^{-(t+l)}a)^k+
\frac{1}{a}\binom{r}{\ell}M_1^{t+\ell}\sum_{k=1}^{q_1}1^k
\\ & =
\sum_{l=0\atop l\neq \ell}^{r-1}\binom{r}{l}M_1^{t+l}\cdot \frac{1-(2^{-(t+l)}a)^{q_1}}{2^{t+l}-a}+\frac{1}{a}\binom{r}{\ell}M_1^{t+\ell}q_1
\end{align*}

We shall summarize these two cases in a single expression by writing
\begin{equation}
(*)=\sum_{l=0\atop l\neq \ell}^{r-1}\binom{r}{l}M_1^{t+l}\cdot \frac{1-(2^{-(t+l)}a)^{q_1}}{2^{t+l}-a}+ \delta_{\ell}\cdot \frac{1}{a}\binom{r}{\ell}M_1^{t+\ell}q_1.
\label{eqn:estrelleta1}
\end{equation}

Notice that if $r=0$, then $(*)=0$ and, since in this case $\delta_{\ell}=0$, the right-hand side expression of (\ref{eqn:estrelleta1}) is indeed 0. Without making use of $\delta_\ell$, (\ref{eqn:estrelleta1}) would not be true when $r=0$. In the next sum we shall apply a similar argument to split each sum $\sum\limits_{l=0}^{r-1}$ into a sum $\sum\limits_{l=0\atop l\neq \ell}^{r-1}$ and a term that must only be considered when $r>0$ and  $\ell\in \{0,\ldots,r-1\}$, and hence it appears multiplied by $ \delta_{\ell}$.

For every $j=1,\ldots,s-1$:
\begin{align*}
& \sum_{k=q_j+1}^{q_{j+1}}a^{k-1}\Bigg(\Big(1+\left\lfloor\frac{n-1}{2^{k}}\right\rfloor\Big)^r\left\lfloor\frac{n-1}{2^{k}}+\frac{1}{2}\right\rfloor^t-\left\lfloor\frac{n-1}{2^{k}}+\frac{1}{2}\right\rfloor^r\left\lfloor\frac{n-1}{2^{k}}\right\rfloor^t\Bigg)\\
&\quad=\sum_{k=q_j+2}^{q_{j+1}}a^{k-1}\sum_{l=0}^{r-1}\binom{r}{l}(2^{-k}M_{j+1})^{t+l}+a^{q_j}\sum_{i=0}^{r+t-1}\binom{r+t}{i}(2^{-q_j-1}M_{j+1})^i  -a^{q_j} \sum_{l=0}^{r-1}\binom{r}{l}(2^{-q_j-1}M_{j+1})^{t+l}\\
&\quad=\frac{1}{a}\sum_{l=0}^{r-1}\binom{r}{l}M_{j+1}^{t+l}\sum_{k=q_j+1}^{q_{j+1}}(2^{-(t+l)}a)^k+\frac{1}{a}\sum_{i=0}^{r+t-1}\binom{r+t}{i}M_{j+1}^i(2^{-i}a)^{q_j+1}  -\frac{2}{a} \sum_{l=0}^{r-1}\binom{r}{l}M_{j+1}^{t+l}(2^{-(t+l)}a)^{q_j+1}\\
&\quad=\frac{1}{a}\sum_{l=0\atop l\neq \ell}^{r-1}\binom{r}{l}M_{j+1}^{t+l}\sum_{k=q_j+1}^{q_{j+1}}(2^{-(t+l)}a)^k +\delta_{\ell}\frac{1}{a}\binom{r}{\ell}M_{j+1}^{t+\ell}\sum_{k=q_j+1}^{q_{j+1}}1^k +\frac{1}{a}\sum_{i=0}^{r+t-1}\binom{r+t}{i}M_{j+1}^i(2^{-i}a)^{q_j+1}\\
&\qquad    -\frac{2}{a} \sum_{l=0}^{r-1}\binom{r}{l}M_{j+1}^{t+l}(2^{-(t+l)}a)^{q_j+1}\\
&\quad=\sum_{l=0\atop l\neq \ell}^{r-1}\binom{r}{l}M_{j+1}^{t+l}\cdot  \frac{(2^{-(t+l)}a)^{q_{j}}-(2^{-(t+l)}a)^{q_{j+1}}}{2^{t+l}-a}+\delta_{\ell}\frac{1}{a}\binom{r}{\ell}M_{j+1}^{t+\ell}(q_{j+1}-q_{j})\\
&\quad\qquad +\frac{1}{a}\sum_{i=0}^{r+t-1}\binom{r+t}{i}M_{j+1}^i(2^{-i}a)^{q_j+1}-\frac{2}{a} \sum_{l=0}^{r-1}\binom{r}{l}M_{j+1}^{t+l}(2^{-(t+l)}a)^{q_j+1}
\end{align*}

Then (and using $q_0=0$)
\begin{align}
\widehat{y}_n & = \sum_{k=1}^{q_s} a^{k-1}\Bigg(\Big(1+\left\lfloor\frac{n-1}{2^{k}}\right\rfloor\Big)^r\left\lfloor\frac{n-1}{2^{k}}+\frac{1}{2}\right\rfloor^t-\left\lfloor\frac{n-1}{2^{k}}+\frac{1}{2}\right\rfloor^r\left\lfloor\frac{n-1}{2^{k}}\right\rfloor^t\Bigg)\nonumber\\ 
&  =
\sum_{j=0}^{s-1}\sum_{k=q_j+1}^{q_{j+1}} a^{k-1}\Bigg(\Big(1+\left\lfloor\frac{n-1}{2^{k}}\right\rfloor\Big)^r\left\lfloor\frac{n-1}{2^{k}}+\frac{1}{2}\right\rfloor^t-\left\lfloor\frac{n-1}{2^{k}}+\frac{1}{2}\right\rfloor^r\left\lfloor\frac{n-1}{2^{k}}\right\rfloor^t\Bigg)\nonumber\\
&=\delta_{\ell}\cdot\frac{1}{a}\binom{r}{\ell}M_1^{t+\ell}q_1+\sum_{l=0\atop l\neq \ell}^{r-1}\binom{r}{l}M_1^{t+l}\cdot \frac{1-(2^{-(t+l)}a)^{q_1}}{2^{t+l}-a}\nonumber\\
&\quad+\sum_{j=1}^{s-1}\Bigg[\delta_{\ell}\cdot\frac{1}{a}\binom{r}{\ell}M_{j+1}^{t+\ell}(q_{j+1}-q_{j})+\frac{1}{a}\sum_{i=0}^{r+t-1}\binom{r+t}{i}M_{j+1}^i(2^{-i}a)^{q_j+1}\nonumber \\
&\qquad +\sum_{l=0\atop l\neq \ell}^{r-1}\binom{r}{l}M_{j+1}^{t+l}\cdot \frac{(2^{-(t+l)}a)^{q_{j}}-(2^{-(t+l)}a)^{q_{j+1}}}{2^{t+l}-a}-\frac{2}{a}\sum_{l=0}^{r-1}\binom{r}{l}M_{j+1}^{t+l}(2^{-(t+l)}a)^{q_j+1}\Bigg]\nonumber\\
&=\delta_{\ell}\cdot\frac{1}{a}\binom{r}{\ell}\sum_{j=0}^{s-1}M_{j+1}^{t+\ell}(q_{j+1}-q_{j})+\sum_{l=0\atop l\neq \ell}^{r-1}\frac{\binom{r}{l}}{2^{t+l}-a}\sum_{j=0}^{s-1}M_{j+1}^{t+l}\cdot \big((2^{-(t+l)}a)^{q_{j}}-(2^{-(t+l)}a)^{q_{j+1}}\big)\nonumber\\
&\quad+\sum_{i=0}^{r+t-1}\binom{r+t}{i}\sum_{j=1}^{s-1}2^{-i}M_{j+1}^i(2^{-i}a)^{q_j}-2\sum_{i=t}^{r+t-1}\binom{r}{i-t}\sum_{j=1}^{s-1}2^{-i}M_{j+1}^{i}(2^{-i}a)^{q_j}
\label{eqn:ynCesc1}
\end{align}

Now, 
\begin{align}
 & \sum_{j=0}^{s-1} ((2^{-(t+l)}a)^{q_{j}}-(2^{-(t+l)}a)^{q_{j+1}})M_{j+1}^{t+l}\nonumber\\
 & \qquad =M_1^{t+l}+\sum_{j=1}^{s-1}(2^{-(t+l)}a)^{q_{j}}(M_{j+1}^{t+l}-M_{j}^{t+l})-(2^{-(t+l)}a)^{q_{s}}M_s^{t+l}\nonumber\\
 & \qquad =(n-1)^{t+l}+\sum_{j=1}^{s-1}(2^{-(t+l)}a)^{q_{j}}(M_{j+1}^{t+l}-(M_{j+1}+2^{q_j})^{t+l})-(2^{-(t+l)}a)^{q_{s}}(2^{q_s})^{t+l}\nonumber\\
 & \qquad =(n-1)^{t+l}-a^{q_{s}}-\sum_{j=1}^{s-1}(2^{-(t+l)}a)^{q_{j}}\sum_{i=0}^{t+l-1}\binom{t+l}{i}M_{j+1}^{i}2^{(t+l-i)q_j}\nonumber\\
 & \qquad =(n-1)^{t+l}-a^{q_{s}}-\sum_{i=0}^{t+l-1}\binom{t+l}{i}\sum_{j=1}^{s-1}M_{j+1}^{i}(2^{-i}a)^{q_j}
 \label{eqn:simpl1}
\end{align}

Using this in Eqn.~(\ref{eqn:ynCesc1}), we obtain
\begin{align*}
\widehat{y}_n & =\delta_{\ell}\cdot\frac{1}{a}\binom{r}{\ell}\sum_{j=0}^{s-1}M_{j+1}^{t+\ell}(q_{j+1}-q_{j})\\
&\quad+\sum_{l=0\atop l\neq \ell}^{r-1}\frac{\binom{r}{l}}{2^{t+l}-a}\Big((n-1)^{t+l}-a^{q_{s}}-\sum_{i=0}^{t+l-1}\binom{t+l}{i}\sum_{j=1}^{s-1}M_{j+1}^{i}(2^{-i}a)^{q_j}\Big)\\
&\quad+\sum_{i=0}^{r+t-1}\binom{r+t}{i}\sum_{j=1}^{s-1}2^{-i}M_{j+1}^i(2^{-i}a)^{q_j}-2\sum_{i=t}^{r+t-1}\binom{r}{i-t}\sum_{j=1}^{s-1}2^{-i}M_{j+1}^{i}(2^{-i}a)^{q_j}\\
& =\delta_{\ell}\cdot\frac{1}{a}\binom{r}{\ell}\sum_{j=0}^{s-1}M_{j+1}^{t+\ell}(q_{j+1}-q_{j})+\sum_{l=0\atop l\neq \ell}^{r-1}\frac{\binom{r}{l}}{2^{t+l}-a}\big((n-1)^{t+l}-a^{q_{s}}\big)\\
&\quad -\sum_{l=0\atop l\neq \ell}^{r-1}\sum_{i=0}^{t+l-1}\frac{\binom{r}{l}\binom{t+l}{i}}{2^{t+l}-a}S_{n-1}^{(0,i)}+\sum_{i=0}^{r+t-1}\binom{r+t}{i}2^{-i}S_{n-1}^{(0,i)}-2\sum_{i=t}^{r+t-1}\binom{r}{i-t}2^{-i}S_{n-1}^{(0,i)}\\
& =\delta_{\ell}\cdot\frac{1}{a}\binom{r}{\ell}\sum_{j=0}^{s-1}M_{j+1}^{t+\ell}(q_{j+1}-q_{j})+\sum_{l=0\atop l\neq \ell}^{r-1}\frac{\binom{r}{l}}{2^{t+l}-a}\big((n-1)^{t+l}-a^{q_{s}}\big)\\
&\quad-\sum_{i=0}^{r+t-1}\sum_{l=i-t+1\atop l\neq \ell}^{r-1}\frac{\binom{r}{l}\binom{t+l}{i}}{2^{t+l}-a}S_{n-1}^{(0,i)}+\sum_{i=0}^{r+t-1}\binom{r+t}{i}2^{-i}S_{n-1}^{(0,i)}-2\sum_{i=0}^{r+t-1}\binom{r}{i-t}2^{-i}S_{n-1}^{(0,i)}\\
& =\delta_{\ell}\cdot\frac{1}{a}\binom{r}{\ell}\sum_{j=0}^{s-1}M_{j+1}^{t+\ell}(q_{j+1}-q_{j})+\sum_{l=0\atop l\neq \ell}^{r-1}\frac{\binom{r}{l}}{2^{t+l}-a}\big((n-1)^{t+l}-a^{q_{s}}\big)\\
&\quad+\sum_{i=0}^{r+t-1}\Big(2^{-i}\binom{r+t}{i}-2^{-i+1}\binom{r}{i-t}-\sum_{l=i-t+1\atop l\neq \ell}^{r-1}\frac{\binom{r}{l}\binom{t+l}{i}}{2^{t+l}-a}\Big)S_{n-1}^{(0,i)}
\end{align*}
as we claimed.
\end{proof}

\begin{rem}
When $r=t=0$, the formula in the last proposition simply says
$y_{n}^{(r,t)}=a^{q_{s_{n-1}}(n-1)}$.
\end{rem}

In the expression for $y_n^{(r,t)}$ given in the last proposition, the term
$$
\delta_{\ell}\cdot\frac{1}{a}\sum_{j=0}^{s_{n-1}-1}M_{j+1}(n-1)^{t+\ell}(q_{j+1}(n-1)-q_{j}(n-1))
$$
is different from 0 only when $r>0$ and $a\in\{2^t,\ldots,2^{r+t-1}\}$. When $t+\ell=0$, that is, when $a=1$, the sum in it simplifies to
\begin{equation}
\sum_{j=0}^{s_{n-1}-1}M_{j+1}(n-1)^{0}(q_{j+1}(n-1)-q_{j}(n-1))=q_{s_{n-1}}(n-1)
\label{eqn:whent=0}
\end{equation}
And notice that this sum contributes in this $a=1$ case  to $y_n^{(r,t)}$ only when $r>0$ and $a=1\in\{2^t,\ldots,2^{r+t-1}\}$, that is, when $r>0$ and $t=0$. When $t+\ell>0$, and recalling that $2^{t+\ell}=a$, this sum is
\begin{align}
& \sum_{j=0}^{s_{n-1}-1}M_{j+1}(n-1)^{t+\ell}(q_{j+1}(n-1)-q_{j}(n-1))\\ &\quad = \sum_{j=1}^{s_{n-1}-1}q_{j}(n-1)(M_{j}(n-1)^{t+\ell}-M_{j+1}(n-1)^{t+\ell})+2^{(t+\ell)q_{s_{n-1}}(n-1)}q_{s_{n-1}}(n-1)\nonumber\\
&\quad = \sum_{j=1}^{s_{n-1}-1}q_{j}(n-1)((M_{j+1}(n-1)+2^{q_j(n-1)})^{t+\ell}-M_{j+1}(n-1)^{t+\ell})+a^{q_{s_{n-1}}(n-1)}q_{s_{n-1}}(n-1)\nonumber\\
&\quad = a^{q_{s_{n-1}}(n-1)}q_{s_{n-1}}(n-1)+\sum_{j=1}^{s_{n-1}-1}q_{j}(n-1)\sum_{i=0}^{t+\ell-1}\binom{t+\ell}{i}M_{j+1}(n-1)^i2^{q_j(n-1)(t+\ell-i)}\nonumber\\
&\quad = a^{q_{s_{n-1}}(n-1)}q_{s_{n-1}}(n-1)+\sum_{i=0}^{t+\ell-1}\binom{t+\ell}{i}\sum_{j=1}^{s_{n-1}-1}q_{j}(n-1)M_{j+1}(n-1)^i(a2^{-i})^{q_j(n-1)}\nonumber\\
&\quad = a^{q_{s_{n-1}}(n-1)}q_{s_{n-1}}(n-1)+\sum_{i=0}^{t+\ell-1}\binom{t+\ell}{i}S^{(1,i)}_{n-1}
\label{eqn:whent>0}
\end{align}
where, for every $n,i\in \NN$,
$$
S^{(1,i)}_n=\sum_{j=1}^{s_{n}-1}q_j(n)(2^{-i}a)^{q_j(n)}M_{j+1}(n)^i.
$$
Notice that when $\delta_{\ell}=1$ and $a=1$, Eqn.~(\ref{eqn:whent>0}) becomes Eqn.~(\ref{eqn:whent=0}).

\subsection{Proof of the main result, up to computing the $\alpha$'s}

\noindent For every $d,n,m\in \NN$,
$$
\alpha^{(d,m)}_n = \sum_{k=1}^{n-1}S^{(d,m)}_{k}=\sum_{k=1}^{n-1}\sum_{j=1}^{s_{k}-1}q_j(k)^d(2^{-m}a)^{q_j(k)}M_{j+1}(k)^m
$$
By Proposition \ref{lem:qib1} and Eqn.~(\ref{eqn:xsumadey}),
we have
\begin{align}
x^{(r,t)}_n & =\sum_{k=2}^n y^{(r,t)}_k= \sum_{k=1}^{n-1}\Bigg[a^{q_{s_k}(k)}+\delta_{\ell}\cdot\frac{1}{a}\binom{r}{\ell}\sum_{j=0}^{s_k-1}M_{j+1}(k)^{t+\ell}(q_{j+1}(k)-q_{j}(k))+\sum_{l=0\atop l\neq \ell}^{r-1}\frac{\binom{r}{l}}{2^{t+l}-a}\cdot k^{t+l}\nonumber\\
 &\quad -\sum_{l=0\atop l\neq \ell}^{r-1}\frac{\binom{r}{l}}{2^{t+l}-a}\cdot a^{q_{s_k}(k)} +\sum_{i=0}^{r+t-1}\Big(2^{-i}\binom{r+t}{i}-2^{-i+1}\binom{r}{i-t}-\sum_{l=i-t+1\atop l\neq \ell}^{r-1}\frac{\binom{r}{l}\binom{t+l}{i}}{2^{t+l}-a}\Big)S_{k}^{(0,i)}\Bigg]\nonumber\\
 &=\delta_{\ell}\cdot\frac{1}{a}\binom{r}{\ell}\sum_{k=1}^{n-1}\sum_{j=0}^{s_k-1}M_{j+1}(k)^{t+\ell}(q_{j+1}(k)-q_{j}(k)) +\Big(1-\sum_{l=0\atop l\neq \ell}^{r-1}\frac{\binom{r}{l}}{2^{t+l}-a}\Big)\Big(\sum_{k=1}^{n-1} a^{q_{s_k}(k)}\Big)\nonumber\\
 &\quad  +\sum_{l=0\atop l\neq \ell}^{r-1}\frac{\binom{r}{l}}{2^{t+l}-a}\sum_{k=1}^{n-1}k^{t+l}+\sum_{i=0}^{r+t-1}\Big(2^{-i}\binom{r+t}{i}-2^{-i+1}\binom{r}{i-t}-\sum_{l=i-t+1\atop l\neq \ell}^{r-1}\frac{\binom{r}{l}\binom{t+l}{i}}{2^{t+l}-a}\Big)\sum_{k=1}^{n-1}S_{k}^{(0,i)}\nonumber\\
 &=\delta_{\ell}\cdot\frac{1}{a}\binom{r}{\ell}\sum_{k=1}^{n-1}\sum_{j=0}^{s_k-1}M_{j+1}(k)^{t+\ell}(q_{j+1}(k)-q_{j}(k)) +\Big(1-\sum_{l=0\atop l\neq \ell}^{r-1}\frac{\binom{r}{l}}{2^{t+l}-a}\Big)\Big(\sum_{k=1}^{n-1} a^{q_{s_k}(k)}\Big)\nonumber\\
 &\quad  +\sum_{l=0\atop l\neq \ell}^{r-1}\frac{\binom{r}{l}}{(t+l+1)(2^{t+l}-a)}\Big(\sum_{j=0}^{t+l+1} \binom{t+l+1}{j}B_jn^{t+l+1-j}+(-1)^{t+l}B_{t+l+1}\Big)
\nonumber\\
 &\quad +\sum_{i=0}^{r+t-1}\Big(2^{-i}\binom{r+t}{i}-2^{-i+1}\binom{r}{i-t}-\sum_{l=i-t+1\atop l\neq \ell}^{r-1}\frac{\binom{r}{l}\binom{t+l}{i}}{2^{t+l}-a}\Big)\alpha^{(0,i)}_{n}\nonumber\\
& \text{(by Faulhaber’s Formula (\ref{eqn:sumpots}))}\nonumber\\
&=\delta_{\ell}\cdot\frac{1}{a}\binom{r}{\ell}\sum_{k=1}^{n-1}\sum_{j=0}^{s_k-1}M_{j+1}(k)^{t+\ell}(q_{j+1}(k)-q_{j}(k)) +\Big(1-\sum_{l=0\atop l\neq \ell}^{r-1}\frac{\binom{r}{l}}{2^{t+l}-a}\Big)\Big(\sum_{k=1}^{n-1} a^{q_{s_k}(k)}\Big)\nonumber\\
 &\quad  +\sum_{k=1}^{r+t}\Big(\sum_{i=k\atop i\neq t+\ell+1}^{r+t}\frac{\binom{r}{i-t-1}\binom{i}{k}B_{i-k}}{i(2^{i-1}-a)}\Big)n^k
+\sum_{i=t+1\atop i\neq t+\ell+1}^{r+t}\frac{\binom{r}{i-t-1}}{i(2^{i-1}-a)}(B_i-(-1)^{i}B_{i})\nonumber \\
&\quad +\sum_{i=0}^{r+t-1}\Big(2^{-i}\binom{r+t}{i}-2^{-i+1}\binom{r}{i-t}-\sum_{l=i-t+1\atop l\neq \ell}^{r-1}\frac{\binom{r}{l}\binom{t+l}{i}}{2^{t+l}-a}\Big)\alpha^{(0,i)}_{n}
 \label{eqn:sumycapk}
 \end{align}
because
\begin{align*}
& \sum_{l=0\atop l\neq \ell}^{r-1}\frac{\binom{r}{l}}{(t+l+1)(2^{t+l}-a)}\Big(\sum_{j=0}^{t+l+1} \binom{t+l+1}{j}B_jn^{t+l+1-j}+(-1)^{t+l}B_{t+l+1}\Big)\\
&\qquad =\sum_{l=0\atop l\neq \ell}^{r-1}\frac{\binom{r}{l}}{(t+l+1)(2^{t+l}-a)}\Big(\sum_{k=0}^{t+l+1} \binom{t+l+1}{k}B_{t+l+1-k}n^k+(-1)^{t+l}B_{t+l+1}\Big)\\
&\qquad =\sum_{i=t+1\atop i\neq t+\ell+1}^{r+t}\frac{\binom{r}{i-t-1}}{i(2^{i-1}-a)}\Big(\sum_{k=0}^{i} \binom{i}{k}B_{i-k}n^k-(-1)^{i}B_{i}\Big)\\
& \qquad =\sum_{k=0}^{r+t}\Big(\sum_{i=k\atop i\neq t+\ell+1}^{r+t}\frac{\binom{r}{i-t-1}\binom{i}{k}B_{i-k}}{i(2^{i-1}-a)}\Big)n^k
-\sum_{i=t+1\atop i\neq t+\ell+1}^{r+t}\frac{\binom{r}{i-t-1}}{i(2^{i-1}-a)}(-1)^{i}B_{i}\\
& \qquad =\sum_{i=0\atop i\neq t+\ell+1}^{r+t}\frac{\binom{r}{i-t-1}}{i(2^{i-1}-a)}B_{i}
+\sum_{k=1}^{r+t}\Big(\sum_{i=k\atop i\neq t+\ell+1}^{r+t}\frac{\binom{r}{i-t-1}\binom{i}{k}B_{i-k}}{i(2^{i-1}-a)}\Big)n^k-\sum_{i=t+1\atop i\neq t+\ell+1}^{r+t}\frac{\binom{r}{i-t-1}}{i(2^{i-1}-a)}(-1)^{i}B_{i}\\
& \qquad =\sum_{k=1}^{r+t}\Big(\sum_{i=k\atop i\neq t+\ell+1}^{r+t}\frac{\binom{r}{i-t-1}\binom{i}{k}B_{i-k}}{i(2^{i-1}-a)}\Big)n^k
+\sum_{i=t+1\atop i\neq t+\ell+1}^{r+t}\frac{\binom{r}{i-t-1}}{i(2^{i-1}-a)}(B_i-(-1)^{i}B_{i})
\end{align*}
Now, in Eqn.~(\ref{eqn:sumycapk}) we have that, by Lemma \ref{lem:sum_a^q},
\begin{align*}
\sum_{k=1}^{n-1} a^{q_{s_{k}}(k)} & =T(0,q_{s_n}(n),2a)+na^{q_{s_n}(n)}-(2a)^{q_{s_n}(n)}\\
& = \left\{\begin{array}{ll}
q_{s_n}(n)+n\cdot 2^{-q_{s_n}(n)}-1 & \text{ if $a=1/2$}\\[1ex]
\dfrac{(2a)^{q_{s_n}(n)}-1}{2a-1}+n\cdot a^{q_{s_n}(n)}-(2a)^{q_{s_n}(n)}& \text{ if $a\neq 1/2$}
\end{array}\right.
\end{align*}
and, by (\ref{eqn:Bimp}),
$$
\sum_{i=t+1\atop i\neq t+\ell+1}^{r+t}\frac{\binom{r}{i-t-1}}{i(2^{i-1}-a)}(B_i-(-1)^{i}B_{i})=\left\{\begin{array}{ll}
1/(a-1)& \text{ if $r>0$, $t=0$, $a\neq 1$}\\
0 & \text{otherwise}
\end{array}\right.
$$
Concerning the first term of Eqn.~(\ref{eqn:sumycapk}), which only contributes to $x^{(r,t)}_n$ when $r>0$ and $a=2^{t+\ell}$ with $\ell\in\{0,\ldots,r-1\}$, by (\ref{eqn:whent>0}) and Lemma \ref{lem:sum_a^q} this sum can be expressed as
\begin{align*}
&\sum_{k=1}^{n-1}\sum_{j=0}^{s_k-1}M_{j+1}(k)^{t+\ell}(q_{j+1}(k)-q_{j}(k))= \sum_{k=1}^{n-1}\Big(a^{q_{s_{k}}(k)}q_{s_{k}}(k)+\sum_{i=0}^{t+\ell-1}\binom{t+\ell}{i}S^{(1,i)}_{k}\Big)\\
&\quad = \sum_{k=1}^{n-1} a^{q_{s_{k}}(k)}q_{s_{k}}(k)+\sum_{i=0}^{t+\ell-1}\binom{t+\ell}{i}\sum_{k=1}^{n-1}S^{(1,i)}_{k}\\
& \quad =T(1,q_{s_n}(n),2a)+nq_{s_{n}}(n)a^{q_{s_n}(n)}-q_{s_{n}}(n)(2a)^{q_{s_n}(n)}+\sum_{i=0}^{t+\ell-1}\binom{t+\ell}{i}\alpha^{(1,i)}_{n}
\end{align*}

In summary, this proves the following theorem. The expressions given in Section \ref{sec:mr} are obtained by rewriting the formula in this theorem for the different relevant cases of the exponents $r,t$ and the coefficient $a$.

\begin{thm}
For every $n\geq 2$,
\begin{align*}
&x^{(r,t)}_n = \Bigg(1-\sum_{l=0\atop l\neq \ell}^{r-1}\frac{\binom{r}{l}}{2^{t+l}-a} \Bigg)\big(T(0,q_{s_n}(n),2a)+na^{q_{s_n}(n)}-(2a)^{q_{s_n}(n)}\big)\nonumber\\
 &\qquad  +\sum_{k=1}^{r+t} \Bigg(\sum_{i=k\atop i\neq t+\ell+1}^{r+t}\frac{\binom{r}{i-t-1}\binom{i}{k}B_{i-k}}{i(2^{i-1}-a)} \Bigg)n^k
+\frac{1}{a-1}\cdot\delta_{r>0,t=0,a\neq 1}\nonumber \\
&\qquad +\sum_{i=0}^{r+t-1} \Bigg(2^{-i}\binom{r+t}{i}-2^{-i+1}\binom{r}{i-t}-\sum_{l=i-t+1\atop l\neq \ell}^{r-1}\frac{\binom{r}{l}\binom{t+l}{i}}{2^{t+l}-a} \Bigg)\alpha^{(0,i)}_{n}\\
&\qquad +\delta_{\ell}\cdot\frac{1}{a}\binom{r}{\ell}\Bigg(T(1,q_{s_n}(n),2a)+nq_{s_{n}}(n)a^{q_{s_n}(n)}-q_{s_{n}}(n)(2a)^{q_{s_n}(n)}+\sum_{i=0}^{t+\ell-1}\binom{t+\ell}{i}\alpha^{(1,i)}_{n}\Bigg)
\end{align*}
where $\delta_{r>0,t=0,a\neq 1}=1$ if $r>0$, $t=0$, and $a\neq 1$, and it is 0 otherwise.
\end{thm}

So, this proves our main result, up to the computation of the terms 
$$
\alpha^{(d,i)}_n  =\sum_{k=1}^{n-1}\sum_{j=1}^{s_{k}-1}q_{j}(k)^d(2^{-i}a)^{q_j(k)}M_{j+1}(k)^i,\qquad d=0,1.
$$
And notice that we only need to compute the $\alpha^{(1,i)}$'s when $r\geq 1$ and $a=2^{t+\ell}$, for some $\ell\in\{0,\ldots,r-1\}$ with $t+\ell>0$. In particular, in the computation of $\alpha^{(1,i)}$ we can assume that $a\neq 1,1/2$, in order to avoid discussing cases that are unrelevant for our main result.

\subsection{Computation of the $\alpha$'s}

\noindent Recall that, for every $d,m\in \NN$ and $n\in \NN_{\geq 1}$, we set
$$
S^{(d,m)}_n=\sum_{j=1}^{s_{n}-1}q_j(n)^d(2^{-m}a)^{q_j(n)}M_{j+1}(n)^m
$$
and
$$
\alpha^{(d,m)}_n=\sum_{k=1}^{n-1}S^{(d,m)}_k=\sum_{k=1}^{n-1}\sum_{j=1}^{s_{k}-1}q_j(k)^d(a2^{-m})^{q_j(k)}M_{j+1}(k)^m
$$
We are interested in closed formulas for $\alpha^{(d,m)}_n$ when $d=0$ as well as when $d=1$ and specific values of $a$ which in particular exclude the cases $a=1,1/2$.

When $m=0$, by Lemmas \ref{lem:sum_a^q}  and \ref{lem:sumx^t} we have that
\begin{align}
\alpha^{(d,0)}_n & =\sum_{k=1}^{n-1}\sum_{j=1}^{s_{k}-1}q_j(k)^da^{q_j(k)}=\sum_{k=1}^{n-1}\sum_{j=1}^{s_{k}}q_j(k)^da^{q_j(k)}-\sum_{k=1}^{n-1}q_{s_k}(k)^da^{q_{s_k}(k)}\nonumber\\
& =\sum_{i=1}^{s_n} 2^{q_i(n)-1}T(d,q_{i}(n),a) +\sum_{i=1}^{s_n}q_i(n)^da^{q_i(n)}(n-M_i(n))\nonumber\\
&\qquad\qquad-T(d,q_{s_n}(n),2a)-(n-2^{q_{s_n}(n)})q_{s_{n}}(n)^da^{q_{s_n}(n)}\nonumber\\
& = \sum_{i=1}^{s_n} 2^{q_i(n)-1}T(d,q_{i}(n),a) +\sum_{i=1}^{s_n-1}q_i(n)^da^{q_i(n)}(n-M_i(n))-T(d,q_{s_n}(n),2a)
\label{eqn:alphad0na}
 \end{align}
and then, in particular,
\begin{align*}
\alpha^{(0,0)}_n& =\sum_{i=1}^{s_n} 2^{q_i(n)-1}T(0,q_{i}(n),a) +\sum_{i=1}^{s_n-1}a^{q_i(n)}(n-M_i(n))-T(0,q_{s_n}(n),2a)
\\
&=\left\{\begin{array}{ll}
 \displaystyle \sum_{j=1}^{s_n}2^{q_j(n)-1}(q_j(n)-2j)+(s_n-1)n+1 &  \text{ if $a= 1$}\\[2ex]
 \displaystyle\sum_{j=1}^{s_n-1} 2^{-q_j(n)}(n-M_j(n))+n-q_{s_n}-s_n & \text{ if $a=1/2$}\\[2ex]
\displaystyle \sum_{j=1}^{s_n-1} a^{q_j(n)}(n-M_j(n))+\frac{1}{2(a-1)}\Big(\sum_{j=1}^{s_n} (2a)^{q_j(n)}-n\Big) - \frac{(2a)^{q_{s_n}(n)}-1}{2a-1} &  \text{ if $a\neq 1,1/2$}
\end{array}\right.
\end{align*}
and, when $a\neq 1,1/2$,
\begin{align*}
\alpha^{(1,0)}_n & =\sum_{i=1}^{s_n} 2^{q_i(n)-1}T(1,q_{i}(n),a) +\sum_{i=1}^{s_n-1}q_i(n)a^{q_i(n)}(n-M_i(n))-T(1,q_{s_n}(n),2a)\\
&= \sum_{i=1}^{s_n-1}q_i(n)a^{q_i(n)}(n-M_i(n))+\frac{1}{2(a-1)}\sum_{i=1}^{s_n}(2a)^{q_i(n)} q_i(n) -\frac{a}{2(a-1)^2}\sum_{i=1}^{s_n}(2a)^{q_i(n)}\\ 
&\qquad\qquad -\frac{1}{2a-1}(2a)^{q_{s_n}(n)}q_{s_n}(n)+\frac{2a}{(2a-1)^2}(2a)^{q_{s_n}(n)}+\frac{a}{2(a-1)^2}n-\frac{2a}{(2a-1)^2}
\end{align*}
Therefore, in the sequel we only consider the case $m>0$.

\begin{lemma}\label{lem:prealpha1}
For every $m\in \NN_{\geq 1}$, $d\in\NN$, $M\in \NN_{\geq 1}$, and $0\leq l \leq q_1(M)$,
$$
\aadm_{M+2^l}=\aadm_M+\aadm_{2^l}+\sum_{p=0}^{m-1}\binom{m}{p}M^{m-p}\gamma^{(d,p,m)}_l +2^lS_M^{(d,m)}
$$
\end{lemma}

\begin{proof}
Let $M=\sum_{i=1}^s 2^{q_i}$, with $s\geq 1$ and  $q_s>\cdots>q_1>0$, and let $0\leq l \leq q_1$, so that $s_M=s$ and $q_i(M)=q_i$ for each $i=1,\ldots,s$.
Then
\begin{align*}
\aadm_{M+2^l} & =
\sum_{k=1}^{M-1}\sum_{i=1}^{s_{k}-1}q_i(k)^d(a2^{-m})^{q_i(k)} M_{i+1}(k)^m+\sum_{k=M}^{M+2^l-1}\sum_{i=1}^{s_{k}-1}q_i(k)^d(a2^{-m})^{q_i(k)} M_{i+1}(k)^m\\
& =
\aadm_{M}+\sum_{k=M}^{M+2^l-1}\sum_{i=1}^{s_{k}-1}q_i(k)^d(a2^{-m})^{q_i(k)} M_{i+1}(k)^m
\end{align*}
where
\begin{align*}
& \sum_{k=M}^{M+2^l-1}\sum_{i=1}^{s_{k}-1}q_i(k)^d(a2^{-m})^{q_i(k)} M_{i+1}(k)^m\nonumber\\
&\quad =
\sum_{k=M}^{M+2^l-1}  \sum_{i=1}^{s_{k}-s}q_i(k)^d(a2^{-m})^{q_i(k)} M_{i+1}(k)^m +\sum_{k=M}^{M+2^l-1}  \sum_{i=s_k-s+1}^{s_k-1}q_i(k)^d(a2^{-m})^{q_i(k)}M_{i+1}(k)^m\nonumber\\
&\quad =
\sum_{k=M}^{M+2^l-1}  \sum_{i=1}^{s_{k}-s}q_i(k)^d(a2^{-m})^{q_i(k)} (M_{i+1}(k-M)+M)^m +\sum_{k=M}^{M+2^l-1}  \sum_{i=1}^{s-1} q_i^d (a2^{-m})^{q_i}M_{i+1}(M)^m\nonumber\\
&\quad =
\sum_{k=0}^{2^l-1}  \sum_{i=1}^{s_{k}} q_i(k)^d(a2^{-m})^{q_i(k)} (M_{i+1}(k)+M)^m +2^l  \sum_{i=1}^{s-1} q_i^d(a2^{-m})^{q_i} M_{i+1}(M)^m\nonumber\\
&\quad =
\sum_{k=1}^{2^l-1}  \sum_{i=1}^{s_{k}} q_i(k)^d(a2^{-m})^{q_i(k)}\sum_{p=0}^m\binom{m}{p}M_{i+1}(k)^pM^{m-p}+2^l  S_{M}^{(d,m)}\nonumber\\
&\quad= \sum_{p=0}^{m-1}\binom{m}{p}M^{m-p}\sum_{k=1}^{2^l-1}  \sum_{i=1}^{s_{k}} q_i(k)^d(a2^{-m})^{q_i(k)}M_{i+1}(k)^p
+\sum_{k=1}^{2^l-1}  \sum_{i=1}^{s_{k}} q_i(k)^d(a2^{-m})^{q_i(k)} M_{i+1}(k)^m\\
&\qquad\qquad+2^l  S_{M}^{(d,m)}\nonumber\\
&\quad =\sum_{p=0}^{m-1}\binom{m}{p}M^{m-p}\gamma^{(d,p,m)}_l +\aadm_{2^l}+2^l  S_{M}^{(d,m)}
\end{align*}
yielding the expression in the statement.
\end{proof}

\begin{cor}\label{lem:alphad2l}
For every $m\in \NN_{\geq 1}$ and for every $l\geq 0$,
$$
\aadm_{2^l}=\frac{2^{ml+l-1}}{m+1}\sum_{p=0}^{m}\binom{m+1}{p} 2^{-p(l-1)} B_p\cdot T(d,l-1,2^{p-m}a)
$$
\end{cor}

\begin{proof}
Taking $M=2^l$ in the last lemma (and recalling that $S_{2^l}^{(d,m)}=0$ because $s_{2^{l}}=1$), we obtain
$$
\aadm_{2^{l+1}}=2 \aadm_{2^l}+\sum_{p=0}^{m-1}\binom{m}{p}2^{l(m-p)}\gamma^{(d,p,m)}_l
$$
The solution of this recurrence with $\aadm_{2^{0}}=\aadm_{1}=0$ is
$$
\aadm_{2^{l}} =\sum_{k=1}^{l-1}2^{l-k-1}\sum_{p=0}^{m-1}\binom{m}{p}2^{k(m-p)}\gamma^{(d,p,m)}_k
$$

Then, using Corollary \ref{cor:varphi^0alt},
\begin{align*}
& \aadm_{2^{l}} =\sum_{k=1}^{l-1}2^{l-k-1}\sum_{p=0}^{m-1}\binom{m}{p}2^{k(m-p)}\frac{a^{k-1}2^{-(m-1)(k-1)+pk}}{p+1}\sum_{t=1}^{k-1} (k-t-1)^d(a^{-1}2^{m-p-1})^{t}(B_{p+1}(2^{t})-B_{p+1})\nonumber\\
&\qquad\qquad+
\sum_{k=1}^{l-1}2^{l-k-1}2^{km} (a2^{-(m-1)})^{k-1}(k-1)^d\nonumber\\
&\quad= \frac{2^{m+l-2}}{m+1}\sum_{t=1}^{l-2}\Bigg[\Big(\sum_{p=0}^{m-1}\binom{m+1}{p+1}2^{(m-p-1)t}(B_{p+1}(2^t)-B_{p+1})\Big)\Big(\sum_{k=t+1}^{l-1}(k-t-1)^da^{k-t-1}\Big)\Bigg]\nonumber\\
&\qquad\qquad +
2^{m+l-2}\sum_{k=1}^{l-1}(k-1)^da^{k-1} \nonumber\\
  &=\frac{2^{m+l-2}}{m+1}\sum_{t=1}^{l-2}\Bigg[\Big(\sum_{p=1}^{m}\binom{m+1}{p}2^{(m-p)t}(B_{p}(2^t)-B_{p})\Big)\Big(\sum_{k=0}^{l-t-2}k^da^{k}\Big)\Bigg]+2^{m+l-2}\sum_{k=0}^{l-2} k^da^{k} \nonumber\\
&=\frac{2^{m+l-2}}{m+1}\sum_{t=0}^{l-2}\Bigg[\Big(\sum_{p=0}^{m}\binom{m+1}{p}2^{(m-p)t}(B_{p}(2^t)-B_{p})\Big)T(d,l-t-1,a)\Bigg]
\end{align*}
by Eqn. (\ref{eqn:Bn(0)}). Now,
\begin{align*}
&\sum_{p=0}^{m}\binom{m+1}{p}2^{(m-p)t}(B_{p}(2^t)-B_{p})\\
&\qquad =
2^{-t}\sum_{p=0}^{m}\binom{m+1}{p}2^{t(m+1-p)}B_{p}(2^t)-\sum_{p=0}^{m}\binom{m+1}{p}2^{(m-p)t}B_{p}\\
&\qquad =
2^{-t}(B_{m+1}(2^{t+1})-B_{m+1}(2^t))-\sum_{p=0}^{m}\binom{m+1}{p}2^{(m-p)t}B_{p}\\
&\qquad =
2^{-t}\sum_{p=0}^{m+1}\binom{m+1}{p}2^{(m+1-p)(t+1)}B_{p}-2^{-t}\sum_{p=0}^{m+1}\binom{m+1}{p}2^{(m+1-p)t}B_{p}-\sum_{p=0}^{m}\binom{m+1}{p}2^{(m-p)t}B_{p}
\\
&\qquad =2\sum_{p=0}^{m}\binom{m+1}{p} (2^{(m-p)(t+1)}-2^{(m-p)t})B_p
\end{align*}
and therefore 
\begin{align*}
&\aadm_{2^{l}}  = \frac{2^{m+l-2}}{m+1}\sum_{t=0}^{l-2}\Bigg[\Big(2\sum_{p=0}^{m}\binom{m+1}{p} (2^{(m-p)(t+1)}-2^{(m-p)t})B_p\Big)T(d,l-t-1,a)\Bigg]\\
&\quad = \frac{2^{m+l-1}}{m+1}\sum_{p=0}^{m}\binom{m+1}{p} B_p\cdot \Bigg[\sum_{t=0}^{l-2} 2^{(m-p)(t+1)}T(d,l-t-1,a)-\sum_{t=0}^{l-2}2^{(m-p)t}T(d,l-t-1,a)\Bigg]\\
&\quad = \frac{2^{m+l-1}}{m+1}\sum_{p=0}^{m}\binom{m+1}{p} B_p\cdot \Bigg[\sum_{t=1}^{l-1} 2^{(m-p)t}T(d,l-t,a)-\sum_{t=0}^{l-2}2^{(m-p)t}T(d,l-t-1,a)\Bigg]\\
&\quad = \frac{2^{m+l-1}}{m+1}\sum_{p=0}^{m}\binom{m+1}{p} B_p\cdot \Bigg[\sum_{t=1}^{l-1} 2^{(m-p)t}(l-t-1)^da^{l-t-1}+\sum_{t=1}^{l-2} 2^{(m-p)t}T(d,l-t-1,a)\\
&\quad\qquad\hphantom{= \frac{2^{m+l-1}}{m+1}\sum_{p=0}^{m}\binom{m+1}{p} B_p\cdot \Bigg[}-\sum_{t=0}^{l-2}2^{(m-p)t}T(d,l-t-1,a)\Bigg]\\
&\quad = \frac{2^{m+l-1}}{m+1}\sum_{p=0}^{m}\binom{m+1}{p} B_p\cdot \Bigg[\sum_{t=1}^{l-1} 2^{(m-p)t}(l-t-1)^da^{l-t-1}-T(d,l-1,a)\Bigg]\\
&\quad = \frac{2^{m+l-1}}{m+1}\sum_{p=0}^{m}\binom{m+1}{p} B_p\cdot \Big(\sum_{t=1}^{l-1} 2^{(m-p)t}(l-t-1)^da^{l-t-1}\Big) \qquad \text{(by Eqn. (\ref{eqn:recurrenceB}))}\\
&\quad = \frac{2^{ml+l-1}}{m+1}\sum_{p=0}^{m}\binom{m+1}{p} 2^{-p(l-1)} B_p\cdot \Big(\sum_{t=1}^{l-1} (l-t-1)^d(2^{p-m}a)^{l-t-1}\Big)\\
&\quad = \frac{2^{ml+l-1}}{m+1}\sum_{p=0}^{m}\binom{m+1}{p} 2^{-p(l-1)} B_p\cdot \Big(\sum_{k=0}^{l-2} k^d(2^{p-m}a)^k\Big)\\
&\quad = \frac{2^{ml+l-1}}{m+1}\sum_{p=0}^{m}\binom{m+1}{p} 2^{-p(l-1)} B_p\cdot T(d,l-1,2^{p-m}a)
\end{align*}
as we claimed.
\end{proof}

\begin{proposition}\label{lem:experiment} 
For every $n,m\geq 1$, 
\begin{align*}
\aadm_{n} & = \frac{1}{2(m+1)}\sum_{i=1}^{s_n}\sum_{j=0}^{m}\binom{m+1}{j}B_j 2^{j}M_{i}(n)^{m+1-j}\big(T(d,q_i(n),a2^{j-m})-T(d,q_{i-1}(n),a2^{j-m})\big)\nonumber\\
&\qquad\qquad +\sum_{i=1}^{s_n-1}q_i(n)^d(a2^{-m})^{q_i(n)} (n-M_{i}(n))M_{i+1}(n)^m 
\end{align*}
\end{proposition}

\begin{proof}
To simplify the notations, let $s=s_n$ and $q_i=q_i(n)$  and $M_i=M_i(n)=\sum_{h=i}^s 2^{q_h}$, for every $i=1,\ldots,s$.
Then, Lemma \ref{lem:prealpha1} and a simple argument by induction implies that
\begin{align}
&\aadm_{n} 
= \sum_{k=1}^{s}\aadm_{2^{q_k}}+
\sum_{p=0}^{m-1}\binom{m}{p}\sum_{i=1}^{s-1} M_{i+1}^{m-p}\gamma^{(d,p,m)}_{q_i}  +\sum_{k=1}^{s-1}2^{q_k}S_{M_{k+1}}^{(d,m)}
\nonumber\\
&\ = \sum_{k=1}^{s}\aadm_{2^{q_k}}+
\sum_{p=0}^{m-1}\binom{m}{p}\sum_{i=1}^{s-1} M_{i+1}^{m-p}\gamma^{(d,p,m)}_{q_i}  +\sum_{k=1}^{s-1}2^{q_k}\sum_{i=k+1}^{s-1}q_i^d(a2^{-m})^{q_i}M_{i+1}^m\nonumber\\
&\  = \sum_{k=1}^{s}\aadm_{2^{q_k}}+
\sum_{p=0}^{m-1}\binom{m}{p}\sum_{i=1}^{s-1} M_{i+1}^{m-p}\gamma^{(d,p,m)}_{q_i}  +\sum_{i=2}^{s-1}q_i^d(a2^{-m})^{q_i} M_{i+1}^m\sum_{k=1}^{i-1}2^{q_k}\nonumber \\
&\  = \sum_{k=1}^{s}\aadm_{2^{q_k}}+
\sum_{p=0}^{m-1}\binom{m}{p}\sum_{i=1}^{s-1} M_{i+1}^{m-p}\gamma^{(d,p,m)}_{q_i}  +\sum_{i=1}^{s-1}q_i^d(a2^{-m})^{q_i} (n-M_{i})M_{i+1}^m
\label{eqn:prealpham1}
\end{align}
where, with the convention that $M_{s+1}^0=1$, 
\begin{align*}
& \sum_{k=1}^{s}\aadm_{2^{q_k}}=   \sum_{k=1}^{s}\frac{2^{mq_k+q_k-1}}{m+1}\sum_{p=0}^{m}\binom{m+1}{p} 2^{-p(q_k-1)} B_p\cdot T(d,q_k-1,2^{p-m}a)\\
&\quad=   \sum_{k=1}^{s}\frac{2^{mq_k+q_k-1}}{m+1}\sum_{p=0}^{m}\binom{m+1}{p} 2^{-p(q_k-1)} B_p\cdot \Big(\sum_{t=0}^{q_k-2}t^d(2^{p-m}a)^{t}\Big)\\
&\quad=   \frac{2^m}{m+1}\sum_{p=0}^{m}\binom{m+1}{p}B_p\cdot\Big(\sum_{k=1}^{s} 2^{(m-p+1)(q_k-1)}\sum_{t=0}^{q_k-2}t^d(2^{p-m}a)^{t}\Big)\\
&\quad=   \frac{2^m}{m+1}\sum_{p=0}^{m-1}\binom{m+1}{m+1}\binom{m+1}{p}B_p\cdot\Big(\sum_{k=1}^{s} M_{k+1}^{m-m}2^{(m-p+1)(q_k-1)}\sum_{t=0}^{q_k-2}t^d(2^{p-m}a)^{t}\Big)\\
&\qquad\qquad + 2^{m-1}B_m\cdot\Big(\sum_{k=1}^{s} 2^{q_k}\sum_{t=0}^{q_k-2}t^da^{t}\Big)
\end{align*}
and, by Corollary \ref{cor:varphi^0alt}, setting $\delta_{q_1\geq 1}=\min\{q_1,1\}$
\begin{align*}
& \sum_{p=0}^{m-1}\binom{m}{p}\sum_{i=1}^{s-1} M_{i+1}^{m-p}\gamma^{(d,p,m)}_{q_i}=\sum_{p=0}^{m-1}\binom{m}{p}\sum_{i=1}^{s} M_{i+1}^{m-p}\gamma^{(d,p,m)}_{q_i}\\
&\quad  = \sum_{p=0}^{m-1}\binom{m}{p}\sum_{i=1}^{s} M_{i+1}^{m-p}\frac{a^{q_i-1}2^{-(m-1)(q_i-1)+pq_i}}{p+1}\sum_{t=1}^{q_i-1} (q_i-t-1)^d(a^{-1}2^{m-p-1})^{t}(B_{p+1}(2^{t})-B_{p+1})\\
&\qquad\qquad+{M_{2}^{m}(a2^{-(m-1)})^{q_1-1}(q_1-1)^d\delta_{q_1 \geq 1}+\sum_{i=2}^{s} M_{i+1}^{m}(a2^{-(m-1)})^{q_i-1}(q_i-1)^d}\\
&\quad  = \frac{1}{m+1}\sum_{p=0}^{m-1}\binom{m+1}{p+1}2^p\sum_{i=1}^{s} M_{i+1}^{m-p}\sum_{t=1}^{q_i-1} (q_i-t-1)^d(2^{p-m+1}a)^{q_i-t-1}(B_{p+1}(2^{t})-B_{p+1})\\
&\qquad\qquad+M_{2}^{m}(a2^{-(m-1)})^{q_1-1}(q_1-1)^d\delta_{q_1 \geq 1}+\sum_{i=2}^{s} M_{i+1}^{m}(a2^{-(m-1)})^{q_i-1}(q_i-1)^d\\
&\quad  = \frac{1}{m+1}\sum_{p=0}^{m-1}\binom{m+1}{p+1}2^p\sum_{i=1}^{s} M_{i+1}^{m-p}\sum_{t=0}^{q_i-2} t^d(2^{p-m+1}a)^{t}(B_{p+1}(2^{q_i-t-1})-B_{p+1})\\
&\qquad\qquad+M_{2}^{m}(a2^{-(m-1)})^{q_1-1}(q_1-1)^d\delta_{q_1 \geq 1}+\sum_{i=2}^{s} M_{i+1}^{m}(a2^{-(m-1)})^{q_i-1}(q_i-1)^d\\
&\quad  = \frac{1}{m+1}\sum_{p=0}^{m-1}\binom{m+1}{p+1}2^p\sum_{i=1}^{s} M_{i+1}^{m-p}\sum_{t=0}^{q_i-2} t^d2^{(p-m+1)t}a^{t}\sum_{j=0}^p\binom{p+1}{j}2^{(p+1-j)(q_i-t-1)}B_j\\
&\qquad\qquad+M_{2}^{m}(a2^{-(m-1)})^{q_1-1}(q_1-1)^d\delta_{q_1 \geq 1}+\sum_{i=2}^{s} M_{i+1}^{m}(a2^{-(m-1)})^{q_i-1}(q_i-1)^d\\
&\quad  =
\frac{1}{m+1}\sum_{j=0}^{m-1}\sum_{p=j}^{m-1}\binom{m+1}{p+1}\binom{p+1}{j}2^p\sum_{i=1}^{s} M_{i+1}^{m-p}2^{(p+1-j)(q_i-1)}\sum_{t=0}^{q_i-2} t^d(a2^{-(m-j)})^tB_j\\
&\qquad\qquad+M_{2}^{m}(a2^{-(m-1)})^{q_1-1}(q_1-1)^d\delta_{q_1 \geq 1}+\sum_{i=2}^{s} M_{i+1}^{m}(a2^{-(m-1)})^{q_i-1}(q_i-1)^d
\end{align*}
Thus, returning back to (\ref{eqn:prealpham1}) (and still with the convention that $M_{s+1}^0=1$), we have
\begin{align}
&\aadm_{n}=\frac{2^m}{m+1}\sum_{j=0}^{m-1}\binom{m+1}{m+1}\binom{m+1}{j}\Big(\sum_{i=1}^{s} M_{i+1}^{m-m}2^{(m-j+1)(q_i-1)}\sum_{t=0}^{q_i-2}t^d(2^{j-m}a)^{t}\Big)B_j\nonumber\\
&\qquad\qquad + 2^{m-1}B_m\cdot\Big(\sum_{i=1}^{s} 2^{q_i}\sum_{t=0}^{q_i-2}t^da^{t}\Big) \nonumber\\
&\qquad\qquad +\frac{1}{m+1}\sum_{j=0}^{m-1}\sum_{p=j}^{m-1}\binom{m+1}{p+1}\binom{p+1}{j}2^p\sum_{i=1}^{s} M_{i+1}^{m-p}2^{(p+1-j)(q_i-1)}\sum_{t=0}^{q_i-2} t^d(a2^{-(m-j)})^tB_j \nonumber\\
&\qquad\qquad+M_{2}^{m}(a2^{-(m-1)})^{q_1-1}(q_1-1)^d\delta_{q_1 \geq 1}+\sum_{i=2}^{s} M_{i+1}^{m}(a2^{-(m-1)})^{q_i-1}(q_i-1)^d\nonumber\\ &\qquad\qquad+\sum_{i=1}^{s-1}q_i^d(a2^{-m})^{q_i} (n-M_{i})M_{i+1}^m \nonumber\\
&\quad= \frac{1}{m+1}\sum_{j=0}^{m-1}\sum_{p=j}^{m}\binom{m+1}{p+1}\binom{p+1}{j}2^p\sum_{i=1}^{s} M_{i+1}^{m-p}2^{(p+1-j)(q_i-1)}\sum_{t=0}^{q_i-2} t^d(a2^{-(m-j)})^tB_j \nonumber\\
&\qquad\qquad + 2^{m-1}B_m\cdot\Big(\sum_{i=1}^{s} 2^{q_i}\sum_{t=0}^{q_i-2}t^da^{t}\Big) \nonumber\\
&\qquad\qquad+M_{2}^{m}(a2^{-(m-1)})^{q_1-1}(q_1-1)^d\delta_{q_1 \geq 1}+\sum_{i=2}^{s} M_{i+1}^{m}(a2^{-(m-1)})^{q_i-1}(q_i-1)^d\nonumber\\ &\qquad\qquad+\sum_{i=1}^{s-1}q_i^d(a2^{-m})^{q_i} (n-M_{i})M_{i+1}^m \nonumber\\
&\quad= \frac{1}{m+1}\sum_{j=0}^{m}\sum_{p=j}^{m}\binom{m+1}{p+1}\binom{p+1}{j}2^p\sum_{i=1}^{s} M_{i+1}^{m-p}2^{(p+1-j)(q_i-1)}\sum_{t=0}^{q_i-2} t^d(a2^{-(m-j)})^tB_j \nonumber\\
&\qquad\qquad+M_{2}^{m}(a2^{-(m-1)})^{q_1-1}(q_1-1)^d\delta_{q_1 \geq 1}+\sum_{i=2}^{s} M_{i+1}^{m}(a2^{-(m-1)})^{q_i-1}(q_i-1)^d\nonumber\\ &\qquad\qquad+\sum_{i=1}^{s-1}q_i^d(a2^{-m})^{q_i} (n-M_{i})M_{i+1}^m
\label{eqn:adm-sl}
\end{align}
Now notice that, for every $i=1,\ldots,s$
\begin{align*}
& \frac{1}{m+1}\sum_{j=0}^{m}\sum_{p=j}^{m}\binom{m+1}{p+1}\binom{p+1}{j}2^pM_{i+1}^{m-p}2^{(p+1-j)(q_i-1)}(q_i-1)^d(a2^{-(m-j)})^{q_i-1}B_j\\
&\qquad =\frac{1}{m+1}(q_i-1)^da^{q_i-1}\sum_{p=0}^{m}\binom{m+1}{p+1}M_{i+1}^{m-p}2^{(p+1-m)(q_i-1)+p}\sum_{j=0}^{p}\binom{p+1}{j}B_j\\
&\qquad =\frac{1}{m+1}(q_i-1)^da^{q_i-1}(m+1)M_{i+1}^{m}2^{-(m-1)(q_i-1)}\qquad \text{(by Eqn.~(\ref{eqn:recurrenceB}))}\\
&\qquad =(q_i-1)^dM_{i+1}^{m}(a2^{-(m-1)})^{q_i-1}
\end{align*}
Therefore, if $q_1\geq 1$, (\ref{eqn:adm-sl}) becomes
\begin{align*}
\aadm_{n} & =\frac{1}{m+1}\sum_{j=0}^{m}\sum_{p=j}^{m}\binom{m+1}{p+1}\binom{p+1}{j}2^p\sum_{i=1}^{s} M_{i+1}^{m-p}2^{(p+1-j)(q_i-1)}\sum_{t=0}^{q_i-2} t^d(a2^{-(m-j)})^tB_j \nonumber\\
&\qquad\qquad+\sum_{i=1}^{s} M_{i+1}^{m}(a2^{-(m-1)})^{q_i-1}(q_i-1)^d+\sum_{i=1}^{s-1}q_i^d(a2^{-m})^{q_i} (n-M_{i})M_{i+1}^m\\
&=\frac{1}{m+1}\sum_{j=0}^{m}\sum_{p=j}^{m}\binom{m+1}{p+1}\binom{p+1}{j}2^p\sum_{i=1}^{s} M_{i+1}^{m-p}2^{(p+1-j)(q_i-1)}\sum_{t=0}^{q_i-1} t^d(a2^{-(m-j)})^tB_j\\
&\qquad\qquad +\sum_{i=1}^{s-1}q_i^d(a2^{-m})^{q_i} (n-M_{i})M_{i+1}^m
\end{align*}
while, if $q_1=0$,
\begin{align*}
\aadm_{n} & =\frac{1}{m+1}\sum_{j=0}^{m}\sum_{p=j}^{m}\binom{m+1}{p+1}\binom{p+1}{j}2^p\sum_{i=1}^{s} M_{i+1}^{m-p}2^{(p+1-j)(q_i-1)}\sum_{t=0}^{q_i-2} t^d(a2^{-(m-j)})^tB_j \nonumber\\
&\qquad\qquad+\sum_{i=2}^{s} M_{i+1}^{m}(a2^{-(m-1)})^{q_i-1}(q_i-1)^d+\sum_{i=1}^{s-1}q_i^d(a2^{-m})^{q_i} (n-M_{i})M_{i+1}^m\\
& =\frac{1}{m+1}\sum_{j=0}^{m}\sum_{p=j}^{m}\binom{m+1}{p+1}\binom{p+1}{j}2^p\sum_{i=2}^{s} M_{i+1}^{m-p}2^{(p+1-j)(q_i-1)}\sum_{t=0}^{q_i-2} t^d(a2^{-(m-j)})^tB_j \nonumber\\
&\qquad\qquad+\sum_{i=2}^{s} M_{i+1}^{m}(a2^{-(m-1)})^{q_i-1}(q_i-1)^d+\sum_{i=1}^{s-1}q_i^d(a2^{-m})^{q_i} (n-M_{i})M_{i+1}^m\\&=\frac{1}{m+1}\sum_{j=0}^{m}\sum_{p=j}^{m}\binom{m+1}{p+1}\binom{p+1}{j}2^p\sum_{i=2}^{s} M_{i+1}^{m-p}2^{(p+1-j)(q_i-1)}\sum_{t=0}^{q_i-1} t^d(a2^{-(m-j)})^tB_j\\
&\qquad\qquad +\sum_{i=1}^{s-1}q_i^d(a2^{-m})^{q_i} (n-M_{i})M_{i+1}^m\\
&=\frac{1}{m+1}\sum_{j=0}^{m}\sum_{p=j}^{m}\binom{m+1}{p+1}\binom{p+1}{j}2^p\sum_{i=1}^{s} M_{i+1}^{m-p}2^{(p+1-j)(q_i-1)}\sum_{t=0}^{q_i-1} t^d(a2^{-(m-j)})^tB_j\\
&\qquad\qquad +\sum_{i=1}^{s-1}q_i^d(a2^{-m})^{q_i} (n-M_{i})M_{i+1}^m
\end{align*}
So, in summary, for every $n\geq 1$,
\begin{align}
\aadm_{n} & =\frac{1}{m+1}\sum_{j=0}^{m}\sum_{p=j}^{m}\binom{m+1}{p+1}\binom{p+1}{j}2^p\sum_{i=1}^{s} M_{i+1}^{m-p}2^{(p+1-j)(q_i-1)}\sum_{t=0}^{q_i-1} t^d(a2^{-(m-j)})^tB_j\nonumber\\
&\qquad\qquad +\sum_{i=1}^{s-1}q_i^d(a2^{-m})^{q_i} (n-M_{i})M_{i+1}^m\label{eqn:aadmprov}
\end{align}
It remains to simplify the first term in the right-hand side expression:
\begin{align*}
& \sum_{j=0}^{m}\sum_{p=j}^{m}\binom{m+1}{p+1}\binom{p+1}{j}2^p\sum_{i=1}^{s} M_{i+1}^{m-p}2^{(p+1-j)(q_i-1)}\sum_{t=0}^{q_i-1} t^d(a2^{-(m-j)})^tB_j \\
&\quad =\sum_{i=1}^{s}\sum_{t=0}^{q_i-1} t^da^t\sum_{j=0}^{m}\sum_{p=j}^{m}\binom{m+1}{j}\binom{m+1-j}{m-p}M_{i+1}^{m-p}2^{(p+1-j)q_i+j-1-(m-j)t}B_j \\
&\quad =\sum_{i=1}^{s}\sum_{t=0}^{q_i-1} t^da^t\sum_{j=0}^{m}\sum_{p=0}^{m-j}\binom{m+1}{j}\binom{m+1-j}{p}M_{i+1}^{p}2^{(m-p+1-j)q_i+j-1-(m-j)t}B_j \\
&\quad =\sum_{i=1}^{s}\sum_{t=0}^{q_i-1} t^da^t\sum_{j=0}^{m}\Big(\sum_{p=0}^{m-j}\binom{m+1-j}{p}M_{i+1}^{p}2^{(m+1-j-p)q_i}\Big)\binom{m+1}{j}2^{j-1-(m-j)t}B_j \\
&\quad =\sum_{i=1}^{s}\sum_{t=0}^{q_i-1} t^da^t\sum_{j=0}^{m}(M_{i}^{m+1-j}-M_{i+1}^{m+1-j})\binom{m+1}{j}2^{j-1-(m-j)t}B_j \\
&\quad =\frac{1}{2}\sum_{j=0}^{m}\binom{m+1}{j}B_j 2^{j}\sum_{i=1}^{s}(M_{i}^{m+1-j}-M_{i+1}^{m+1-j})T(d,q_i,a2^{j-m})\\
&\quad =\frac{1}{2}\sum_{j=0}^{m}\binom{m+1}{j}B_j 2^{j}\sum_{i=1}^{s}M_{i}^{m+1-j}\big(T(d,q_i,a2^{j-m})-T(d,q_{i-1},a2^{j-m})\big)
\end{align*}
So, returning back to Eqn.~(\ref{eqn:aadmprov}), we finally obtain
\begin{align*}
\aadm_{n} & =\frac{1}{2(m+1)}\sum_{i=1}^{s}\sum_{j=0}^{m}\binom{m+1}{j}B_j 2^{j}M_{i}^{m+1-j}\big(T(d,q_i,a2^{j-m})-T(d,q_{i-1},a2^{j-m})\big)\nonumber\\
&\qquad\qquad +\sum_{i=1}^{s-1}q_i^d(a2^{-m})^{q_i} (n-M_{i})M_{i+1}^m 
\end{align*}
as we claimed.
\end{proof}

\begin{rem}
Recall from Eqn. (\ref{eqn:alphad0na}) that
$$
\alpha^{(d,0)}_n = \sum_{i=1}^{s_n} 2^{q_i(n)-1}T(d,q_{i}(n),a) +\sum_{i=1}^{s_n-1}q_i(n)^da^{q_i(n)}(n-M_i(n))-T(d,q_{s_n}(n),2a)
$$
Now, if we set $m=0$ in the right-hand side term of the expression for $\alpha^{(d,m)}_n$, when $m>0$, established in the last proposition, we obtain
\begin{align*}
& \frac{1}{2}\sum_{i=1}^{s_n}M_{i}(n)\big(T(d,q_i(n),a)-T(d,q_{i-1}(n),a)\big)+\sum_{i=1}^{s_n-1}q_i(n)^da^{q_i(n)} (n-M_{i}(n)) \\
&\qquad =\frac{1}{2}\sum_{i=1}^{s_n}T(d,q_i(n),a)(M_{i}(n)-M_{i+1}(n))\big)+\sum_{i=1}^{s_n-1}q_i(n)^da^{q_i(n)} (n-M_{i}(n))\\
&\qquad =\sum_{i=1}^{s_n}T(d,q_i(n),a)2^{q_i(n)-1}+\sum_{i=1}^{s_n-1}q_i(n)^da^{q_i(n)} (n-M_{i}(n))\\
&\qquad =\alpha^{(d,0)}_n+T(d,q_{s_n}(n),2a)
\end{align*}
Therefore, we can combine Eqn. (\ref{eqn:alphad0na}) and Proposition \ref{lem:experiment} into a single formula:
\begin{align*}
\aadm_{n} & =\frac{1}{2(m+1)}\sum_{i=1}^{s}\sum_{j=0}^{m}\binom{m+1}{j}B_j 2^{j}M_{i}^{m+1-j}\big(T(d,q_i,a2^{j-m})-T(d,q_{i-1},a2^{j-m})\big)\nonumber\\
&\qquad\qquad +\sum_{i=1}^{s-1}q_i^d(a2^{-m})^{q_i} (n-M_{i})M_{i+1}^m -T(d,q_{s_n}(n),2a)\cdot \delta_{m=0}
\end{align*}
with $\delta_{m=0}=1$ if $m=0$ and $\delta_{m=0}=0$ if $m>0$.
\end{rem}

\begin{rem}
When $n=2^l$, the formula in Proposition \ref{lem:experiment}  says
\begin{align*}
\aadm_{2^l} & =\frac{1}{2(m+1)}\sum_{j=0}^{m}\binom{m+1}{j}B_j 2^{j+(m+1-j)l}T(d,l,a2^{j-m})\\
& =\frac{2^{(m+1)l-1}}{m+1}\sum_{j=0}^{m}\binom{m+1}{j}B_j 2^{-j(l-1)}T(d,l,a2^{j-m})
\end{align*}
and this expression is equivalent to the one established in Corollary \ref{lem:alphad2l},
$$
\aadm_{2^l} =\frac{2^{(m+1)l-1}}{m+1}\sum_{j=0}^{m}\binom{m+1}{j}B_j 2^{-j(l-1)}T(d,l-1,a2^{j-m}),
$$
because
$$
\sum_{j=0}^{m}\binom{m+1}{j} B_j 2^{-j(l-1)}(a2^{j-m})^{l-1}=
2^{-m(l-1)}\sum_{j=0}^{m}\binom{m+1}{j} B_j =0
$$
\end{rem}


\end{document}